\documentclass[11pt]{article}

\usepackage[letterpaper,margin=1in]{geometry}
\usepackage{amsmath,amssymb,amsthm,mathtools}
\usepackage{booktabs,tabularx,array}
\usepackage{enumitem}
\usepackage[dvipsnames]{xcolor}
\usepackage{tcolorbox}
\tcbuselibrary{breakable}
\usepackage[hypertexnames=false,colorlinks=true,linkcolor=MidnightBlue,citecolor=BrickRed,urlcolor=MidnightBlue]{hyperref}
\usepackage{microtype}
\usepackage{lmodern}

\DeclareFontEncoding{T5}{}{}
\DeclareFontSubstitution{T5}{lmr}{m}{n}
\DeclareMicrotypeSet{papertext}{encoding={OT1,T1,TS1},family={rm*,sf*}}
\UseMicrotypeSet[protrusion]{papertext}
\UseMicrotypeSet[expansion]{papertext}

\newtheorem{theorem}{Theorem}[section]
\newtheorem{lemma}[theorem]{Lemma}
\newtheorem{proposition}[theorem]{Proposition}
\newtheorem{corollary}[theorem]{Corollary}

\theoremstyle{definition}

\newtheorem{algorithm}[theorem]{Algorithm}
\newenvironment{algorithmblock}[1]
 {\begin{tcolorbox}[
    breakable,
    colback=black!1,
    colframe=black!55,
    boxrule=0.5pt,
    arc=1mm,
    left=1.5mm,
    right=1.5mm,
    top=1mm,
    bottom=1mm]
  \begin{algorithm}[#1]}
 {\end{algorithm}
  \end{tcolorbox}}
\theoremstyle{remark}
\newtheorem{remark}[theorem]{Remark}

\newcommand{\R}{\mathbb{R}}
\newcommand{\E}{\mathbb{E}}
\newcommand{\Prb}{\mathbb{P}}
\newcommand{\tr}{\operatorname{tr}}
\newcommand{\rank}{\operatorname{rank}}
\newcommand{\sr}{\operatorname{sr}}
\newcommand{\TV}{d_{\mathrm{TV}}}
\newcommand{\op}{\mathrm{op}}
\newcommand{\Sone}{S_1}
\newcommand{\Sp}[1]{S_{#1}}

\newcommand{\eps}{\varepsilon}
\newcommand{\norm}[1]{\left\lVert #1\right\rVert}
\newcommand{\abs}[1]{\left\lvert #1\right\rvert}
\newcommand{\ip}[2]{\left\langle #1,#2\right\rangle}
\newcommand{\defeq}{:=}
\newcommand{\Nguyen}{Nguy{\microtypesetup{protrusion=false,expansion=false}%
  \fontencoding{T5}\selectfont\symbol{185}}n}

\setlist[itemize]{leftmargin=1.45em,itemsep=0.35em,topsep=0.4em}
\setlist[enumerate]{leftmargin=1.7em,itemsep=0.35em,topsep=0.4em}

\allowdisplaybreaks

\title{Near-Optimal Bounds for Sketching the Schatten Norms}
\author{Lin F. Yang\\
University of California, Los Angeles\\
\texttt{linyang@ee.ucla.edu}%
\thanks{This paper was prepared with AI assistance from GPT-5.6-sol.}}
\date{}

\begin{document}
\maketitle

\begin{abstract}
Let $k_{1,\varepsilon}(n)$ be the smallest number of real linear measurements needed by a
randomized oblivious sketch that estimates the nuclear norm of every fixed real $n\times n$
matrix within a factor $1\pm\varepsilon$, with probability at least $2/3$.  For every fixed
$0<\varepsilon<1$, we prove
\[
  \frac{n^2}{(\log n)^{A_\varepsilon}}
  \le k_{1,\varepsilon}(n)
  \le C_\varepsilon
       \frac{n^2\{\log\log(e^e n)\}^2}{\log(e n)}.
\]
Previously, the best unrestricted bounds for general linear sketches of the Schatten--1 norm
were $\Omega(n)$ and the trivial $O(n^2)$ upper bound~\cite{LiNguyenWoodruff}, leaving a
polynomial gap.  Our bounds close that gap up to polylogarithmic factors and give a nontrivial
logarithmic saving below the $n^2$-measurement storage bound.

The result extends much further.  Write $k_{p,\varepsilon}(n)$ for the
analogous sketch dimension for the Schatten--$p$ norm.  For every fixed finite
$p>0$ that is not a positive even integer, there are positive constants
$A_{p,\varepsilon},C_{p,\varepsilon},c_p$ such that
\[
 \frac{n^2}{(\log n)^{A_{p,\varepsilon}}}
 \le k_{p,\varepsilon}(n)
 \le C_{p,\varepsilon}\frac{n^2}{(\log n)^{c_p}},
\]
so $k_{p,\varepsilon}(n)=n^{2-o(1)}$ throughout the non-even regime.  Together with the known
tight bounds $\Theta_{p,\varepsilon}(n^{2-4/p})$ for positive even $p$ and
$\Theta_\varepsilon(n^2)$ for $p=\infty$~\cite{LiWoodruff}, our results close the remaining
polynomial gap across the Schatten family and complete, up to polylogarithmic factors, the
polynomial-order classification of general linear sketches for all Schatten-$p$ norms.
\end{abstract}

\tableofcontents

\section{Introduction}

\subsection{The problem}

For $A\in\R^{n\times n}$, the singular values
$\sigma_1(A)\ge\cdots\ge\sigma_n(A)\ge0$ are the square roots of the eigenvalues of $A^TA$.
Three norms will be used:
\[
 \norm{A}_{\Sone}=\sum_{i=1}^n\sigma_i(A),\qquad
 \norm{A}_F=\left(\sum_{i=1}^n\sigma_i(A)^2\right)^{1/2},\qquad
 \norm{A}_{\op}=\sigma_1(A).
\]
The first is the \emph{nuclear norm}, also called the Schatten--1 norm.  A randomized
$k$-dimensional linear sketch samples a linear map
\[
 S_\omega:\R^{n\times n}\longrightarrow\R^k
\]
without seeing $A$, stores $S_\omega A$, and applies a decoder
$d_\omega:\R^k\to\R$ that may depend on $\omega$.  The required ``for each'' guarantee is
\begin{equation}
 \Prb_\omega\!\left\{
   (1-\eps)\norm{A}_{\Sone}\le d_\omega(S_\omega A)
   \le(1+\eps)\norm{A}_{\Sone}
 \right\}\ge\frac23
 \quad\text{for every fixed }A.                         \tag{1.1}\label{eq:1.1}
\end{equation}
This is the general linear-sketch model studied by Li, \Nguyen, and
Woodruff~\cite{LiNguyenWoodruff} and by Li and Woodruff~\cite{LiWoodruff}.  The model
uses exact real arithmetic.  Therefore a bit-counting or finite-set encoding argument
does not by itself prove a lower bound on $k$.

The relation to a Turing-machine streaming model requires care.  In that model the input is
normally a stream of integer or rational updates of bounded bit length, and the algorithm stores
a finite binary string.  By contrast, one coordinate of $S_\omega A$ in \eqref{eq:1.1} is an exact real
number and may contain arbitrarily many bits.  Thus $k$ measures the number of data-dependent
real scalars, not their binary representation.  In an idealized real-RAM interpretation, the
theorem bounds the number of data-dependent real words.  It also lower-bounds the number of
linear counters in any finite-precision scheme that obtains its state by rounding $k$ fixed real
linear measurements and is required to work for every real input: a decoder given the exact
measurements could perform the same rounding.  Neither observation specifies how many bits each
counter needs.

There is a subtler connection for bounded-integer turnstile streams.  The earlier theorem of Li,
\Nguyen, and Woodruff replaces a one-pass streaming algorithm by an integer linear sketch while
nearly preserving the number of attainable states~\cite{LiNguyenWoodruffLinearization}; by itself,
that statement does not convert an exact-real row-dimension lower bound into the same bit lower
bound.  Two newer ingredients do.  Gribelyuk, Lin, Woodruff, Yu, and Zhou lift real-sketch
dimension lower bounds to bounded-entry integer sketches~\cite{GribelyukEtAlLifting}, and Jiang,
Liu, and Yu convert an $S$-bit polynomial-length turnstile algorithm for a smooth problem into a
bounded-entry integer sketch with $O(S/\log R)$ rows~\cite{JiangLiuYuLinearization}.  The
Gaussian-smoothed hard pairs constructed here for finite non-even $p$ can be scaled, rounded,
and truncated so that these two results apply.  Corollary~\ref{cor:streaming} records the
resulting finite-memory consequence; Appendix~\ref{app:streaming} proves the required robustness
and parameter compatibility.

Define $k_\eps(n)$ to be the least integer $k\ge0$ for which such a distribution of linear maps
and measurable decoders exists.  If no such sketch existed we would set $k_\eps(n)=\infty$,
although storing all $n^2$ entries shows that $k_\eps(n)\le n^2$.

More generally, for $0<p<\infty$ write
\[
 \norm{A}_{\Sp{p}}=\left(\sum_{i=1}^n\sigma_i(A)^p\right)^{1/p}
\]
and let $k_{p,\eps}(n)$ denote the analogous minimum sketch dimension for a
$(1\pm\eps)$ estimate of $\norm{A}_{\Sp{p}}$.  For $0<p<1$ this functional is a
quasi-norm; the estimation guarantee is unchanged.  Put
$\norm A_{\Sp{\infty}}=\norm A_{\op}$ and define $k_{\infty,\eps}(n)$ analogously.
Thus $k_{1,\eps}(n)=k_\eps(n)$.

If $A$ is promised positive semidefinite, then $\norm{A}_{\Sone}=\tr A$, so one linear
measurement is enough.  The theorem concerns arbitrary real matrices.

\begin{theorem}[Main result]\label{thm:main}
For every fixed $0<\eps<1$, there are finite constants $A_\eps,C_\eps$ such that, for all
sufficiently large $n$,
\begin{equation}
 \boxed{
 \frac{n^2}{(\log n)^{A_\eps}}
 \le k_\eps(n)
 \le C_\eps\frac{n^2\{\log\log(e^e n)\}^2}{\log(e n)}.}
                                                               \tag{1.2}\label{eq:1.2}
\end{equation}
The constants may depend on $\eps$ but not on $n$.
\end{theorem}

The lower-bound construction is not specific to the square root that represents the nuclear
norm.  Every fixed noninteger power of the squared singular values is invisible to the matched
integer moments.  The following consequence is proved in
Appendix~\ref{app:noneven}.

\begin{corollary}[All non-even Schatten exponents]\label{cor:noneven}
Fix $p>0$ with $p\notin\{2,4,6,\ldots\}$ and fix $0<\eps<1$.  There is a finite constant
$A_{p,\eps}$ such that, for all sufficiently large $n$,
\begin{equation}
 k_{p,\eps}(n)\ge \frac{n^2}{(\log n)^{A_{p,\eps}}}.       \tag{1.2a}\label{eq:1.2a}
\end{equation}
The constant may depend on $p$ and $\eps$ but not on $n$; in particular, the statement is not
uniform as $p$ approaches a positive even integer.
\end{corollary}

The upper-bound construction also extends beyond the nuclear norm.  The fractional-power tail
is estimated from the same integer spectral moments, while the leading singular-value power is
read directly from the regularized row sketch.  Appendix~\ref{app:noneven-upper} proves the
following consequence.

\begin{corollary}[Polylogarithmic saving for all non-even exponents]
\label{cor:noneven-upper}
Fix $p>0$ with $p\notin\{2,4,6,\ldots\}$ and fix $0<\eps<1$.  There is a finite constant
$C_{p,\eps}$ such that, for all sufficiently large $n$,
\begin{equation}
 k_{p,\eps}(n)\le C_{p,\eps}\frac{n^2}{(\log n)^{c_p}},
 \qquad
 c_p=\begin{cases}
       p/8,&0<p<2,\\
       1/p^2,&p>2.
     \end{cases}                                         \tag{1.2b}\label{eq:1.2b}
\end{equation}
The exponents are conservative; at $p=1$, Theorem~\ref{thm:main} gives the sharper logarithmic
saving in \eqref{eq:1.2}.
\end{corollary}

The exceptional exponents are structural.  When $p=2q$ is a positive even integer,
$\norm A_{\Sp{p}}^p=\tr(A^TA)^q$ is a polynomial in the entries and is exactly determined by the
matched $q$th spectral moment.  For every other fixed finite $p>0$, the proof identifies a unique
unmatched anchor in the scalar construction and amplifies its contribution to a constant
relative gap.
For positive even $p$, known general linear sketches use
$O_\eps(n^{2-4/p})$ measurements~\cite{LiWoodruff}.  Corollary~\ref{cor:noneven} therefore
exhibits a polynomial separation between the positive even exponents and every neighboring
fixed non-even exponent.
Combining the two corollaries with the tight even-$p$ and
operator-norm results of Li and Woodruff~\cite{LiWoodruff} gives, for every fixed
$0<p\le\infty$ and fixed accuracy,
\begin{equation}
 k_{p,\eps}(n)=
 \begin{cases}
   n^{2-o(1)}, & 0<p<\infty,\quad p\notin\{2,4,6,\ldots\},\\
   \Theta_{p,\eps}(n^{2-4/p}), & p\in\{2,4,6,\ldots\},\\
   \Theta_\eps(n^2), & p=\infty.
 \end{cases}                                                       \tag{1.2c}\label{eq:classification}
\end{equation}
Thus the polynomial order of the general-linear sketch dimension is resolved for the entire
Schatten family; the only remaining losses in the non-even finite regime are polylogarithmic.

\begin{corollary}[Polynomial-length turnstile lower bound]\label{cor:streaming}
Fix $p>0$ with $p\notin\{2,4,6,\ldots\}$ and fix $0<\eps<1$.  There are finite constants
$B_{p,\eps},C_{\mathrm{str}}(p,\eps),c_{\mathrm{str}}(p,\eps)>0$ such that the following holds
for all sufficiently large $n$.  Let $W\ge n^{C_{\mathrm{str}}(p,\eps)}$.  Suppose a randomized
one-pass algorithm receives unit updates
\[
 A_{ij}\leftarrow A_{ij}+\Delta,
 \qquad \Delta\in\{-1,1\},
\]
and, on every stream of length at most $W$, outputs a $(1\pm\eps)$ estimate of
$\norm A_{\Sp{p}}$ for the final matrix $A\in\mathbb Z^{n\times n}$ with probability at least
$2/3$.  Then its memory is at least
\begin{equation}
 c_{\mathrm{str}}(p,\eps)
 \frac{n^2\log W}{(\log n)^{B_{p,\eps}}}
 \quad\text{bits}.                                      \tag{1.3}\label{eq:1.3}
\end{equation}
In particular, for every fixed exponent $C\ge C_{\mathrm{str}}(p,\eps)$, taking $W=n^C$ gives
an $n^{2-o(1)}$-bit lower bound.  The exponent $B_{p,\eps}$ may be taken to be the exponent
furnished by the non-even-$p$ lower-bound proof after changing only its fixed smoothing
constants.
\end{corollary}

The ambient vector space has dimension $n^2$.  At $p=1$, the upper bound saves roughly
$\log n/(\log\log n)^2$ over storing the whole matrix; every other fixed finite non-even exponent
also admits a fixed polylogarithmic saving.  The lower bounds say that no sketch can save a fixed
polynomial in $n$.  Thus the bounds are tight up to polylogarithmic factors throughout the
non-even regime, and $k_{p,\eps}(n)=n^{2-o(1)}$ there.

\subsection{Background and prior work}

Linear sketching asks how many linear measurements are needed to estimate a function of a
high-dimensional input.  It arose in the turnstile streaming model.  In the vector version, an
underlying vector $x\in\R^N$ receives updates of the form
$x_i\leftarrow x_i+\Delta$; because a linear sketch can update its stored state from $i$ and
$\Delta$ alone, it never needs to store $x$.  The classical frequency moments are
\[
 F_p(x)=\sum_{i=1}^N |x_i|^p=\norm{x}_p^p.
\]
Alon, Matias, and Szegedy introduced their systematic streaming study and, in particular, the
foundational second-moment sketch~\cite{AlonMatiasSzegedy}.  For $0<p\le2$, Indyk used random
projections with $p$-stable entries to estimate $\ell_p$ norms in space polylogarithmic in $N$
for fixed accuracy~\cite{IndykStable}; Kane, Nelson, and Woodruff later determined the sharp bit
complexity for small $p$, including the dependence on the accuracy and numerical range of the
stream~\cite{KaneNelsonWoodruff}.  For $p>2$, Indyk and Woodruff gave algorithms using, up to
logarithmic factors, $N^{1-2/p}$ space~\cite{IndykWoodruff}.  Thus the vector problem has a sharp
change of behavior at $p=2$.

The matrix turnstile model similarly receives updates
$A_{ij}\leftarrow A_{ij}+\Delta$.  Schatten norms apply the corresponding vector norm to the
singular-value vector, but a vector norm sketch cannot simply be run on that vector: one entry
update can change every singular value, and the singular values are not linear functions of the
entries.  An early matrix-streaming result of Andoni and \Nguyen\ used left and right random
projections to recover large eigenvalues and estimate the residual Frobenius error in one
pass~\cite{AndoniNguyen}.  This illustrates the useful but restricted \emph{bilinear} sketching
form.

Four distinctions in this literature are relevant to the present result.  A \emph{general linear sketch}
may use an arbitrary linear map on the $n^2$ entries of $A$, whereas a \emph{bilinear sketch}
stores a matrix $SAT$ for chosen left and right projection matrices $S,T$.  Every bilinear sketch
is linear in $A$, but not every linear sketch has this form.  A \emph{bit-space} bound counts a
finite-memory streaming state and therefore depends on precision assumptions, whereas a
\emph{sketch-dimension} bound counts exact real measurements.  Finally, arbitrary-order entry
updates, row-order input, and multiple passes are different access models.  An improvement in one
of these models need not improve the others.

Li, \Nguyen, and Woodruff initiated a systematic study of Schatten-norm sketching across
general versus bilinear sketches and bit space versus exact-real sketch dimension
\cite{LiNguyenWoodruffSODA,LiNguyenWoodruff}.  The conference version proved, for
constant-factor Schatten--1 approximation, an $\Omega(\sqrt n)$ lower bound for general linear
sketches and an $\Omega(n^{1-\gamma})$ lower bound for bilinear sketches for every fixed
$\gamma>0$; the upper bound in both models was the trivial $O(n^2)$ bound obtained by storing all
entries~\cite{LiNguyenWoodruffSODA}.  The journal version strengthened the general-linear lower
bound to $\Omega(n)$~\cite{LiNguyenWoodruff}.  Thus, immediately before the present work, the
best unrestricted constant-factor bounds in the exact-real general-linear model were $\Omega(n)$
and $O(n^2)$.  Their work left open where in this polynomial range the correct complexity lies.
For the rest of the Schatten family, the picture was split by arithmetic properties of the
exponent.  Li and Woodruff obtained tight general-linear dimension bounds
$\Theta_{p,\eps}(n^{2-4/p})$ for every positive even integer $p$ and
$\Theta_\eps(n^2)$ for the operator norm~\cite{LiWoodruff}.  Their arguments exploit polynomial
trace moments or the large-singular-value regime and do not settle $p=1$, or more generally any
finite non-even $p$.  Before the present result there was therefore no polynomial-order
classification of exact-real general linear sketches across all Schatten exponents.

Several streaming algorithms and lower bounds address related, but not identical, matrix models.
For matrices with only constantly many nonzero entries in every row and column, Li and Woodruff
gave a one-pass turnstile algorithm using $n^{1-2/p}$ times factors polynomial in the inverse
accuracy and logarithmic in $n$ when $p$ is an even integer~\cite{LiWoodruffSingularValues}.  For
every fixed non-even $p>0$, they proved a nearly linear bit-space lower bound on such sparse
instances.  In particular, at $p=1$ and for sufficiently small relative error $\delta$, a
$(1+\delta)$-approximation requires
$\Omega_\delta(n^{1-g(\delta)})$ bits, where $g(\delta)\to0$ as $\delta\to0$, even on matrices
with $O(1)$ nonzero entries per row and column.  These are finite-bit streaming lower bounds
rather than lower bounds on the number of exact real measurements, and their nearly linear scale
does not imply the nearly quadratic classification in \eqref{eq:classification}.

There are also nonconstant-approximation upper bounds specific to the Schatten--1 norm.  The
inequalities $\norm A_F\le \norm A_{\Sone}\le\sqrt{\min\{n,d\}}\norm A_F$ give a
$\sqrt{\min\{n,d\}}$-approximation from a standard Frobenius-norm sketch.  More generally, Li and
Woodruff constructed oblivious left-right embeddings that give a factor-$D$ estimate using
$\widetilde O(\min\{n,d\}^2/D^4)$ bits of streaming space for an $n\times d$ matrix
\cite{LiWoodruffEmbeddings}.  This interpolates between the trivial near-quadratic storage bound
at constant $D$ and the $O(1)$-word Frobenius estimate at $D$ of order
$\min\{n,d\}^{1/2}$.  It is a useful space--approximation tradeoff, but for constant $D$ it does
not give a subquadratic bound.

Changing how the matrix is presented can lead to stronger algorithms.  Braverman, Chestnut,
Krauthgamer, Li, Woodruff, and Yang developed faster algorithms, multi-pass algorithms, and
row-order algorithms for Schatten norms, and showed separations from the one-pass entry-update
model~\cite{BravermanEtAlMatrixNorms}.  For doubly sparse matrices presented in row order,
Braverman, Krauthgamer, Krishnan, and Sinoff subsequently obtained the first Schatten-norm
algorithms whose memory is independent of the matrix dimension; their guarantees use sparsity
and multiple passes, and they also prove that multiple passes are necessary in their setting
\cite{BravermanEtAlSparsity}.  These results are algorithmically important, but their structural
promises do not apply to the unrestricted, one-pass, general-sketch problem studied here.

A complementary geometric line of work relates small sketches for norms to low-distortion
embeddings of normed spaces.  Andoni, Krauthgamer, and Razenshteyn used this connection to show,
among other examples, that in their nonlinear distance-threshold model, an $s$-bit factor-$D$
sketch for the trace-norm space must satisfy~\cite{AndoniKrauthgamerRazenshteyn}
\[
 sD=\Omega\!\left(\frac{\sqrt n}{\log n}\right).
\]
This rules out dimension-independent constant-size sketches.
Because it concerns finite-bit sketches for a distance decision problem rather than exact-real
linear measurements of one matrix, it is complementary to the present result and does not imply
the nearly quadratic dependence on $n$ proved here.

The nuclear norm behaves differently on positive semidefinite and unrestricted inputs.  If
$A\succeq0$, then $\norm A_{\Sone}=\tr A$, so one measurement is exact.  The difficulty therefore
comes from the interaction between left and right singular directions for a general matrix.  This
paper concerns unrestricted real square matrices and general linear sketches.

Theorem~\ref{thm:main} nearly resolves the preceding Schatten--1 question.  Its lower and upper
bounds are tight up to polylogarithmic factors, closing the polynomial gap between the previous
$\Omega(n)$ and $O(n^2)$ bounds: for every fixed $\eps$ and every fixed $c>0$, a sketch using
$O(n^{2-c})$ measurements is impossible, while the upper bound achieves a logarithmic saving
below $n^2$.  Corollary~\ref{cor:noneven} shows that the same nearly quadratic obstruction holds
for every finite non-even $p>0$, and Corollary~\ref{cor:noneven-upper} gives a nontrivial
polylogarithmic saving throughout that regime.  Together with the tight positive-even and operator-norm results
reviewed above, this yields the complete polynomial-order classification
\eqref{eq:classification}.  The paper keeps Schatten--1 as its main storyline because that case
receives the sharpest upper bound; the general lower- and upper-bound extensions are proved in
Appendices~\ref{app:noneven} and~\ref{app:noneven-upper}.

\subsection{Contributions}

The paper develops the Schatten--1 problem as its main storyline.  Its lower-bound mechanism then
extends to every non-even exponent, completing the broader classification when combined with
prior results.

\begin{itemize}
 \item The lower bound is uniform over all rank-$k$ linear measurements and all measurable
       decoders.  Its hard pair has deterministic singular values up to a small Gaussian
       perturbation.  Low even moments of the two spectra agree exactly, while their nuclear norms
       differ by a fixed relative amount.  A noninteger-moment variant gives the same
       $n^{2-o(1)}$ lower bound for every fixed finite Schatten exponent
       $p\notin\{2,4,6,\ldots\}$.
 \item The upper bound uses four independent Gaussian matrix products, all sampled before the
       input is seen.  The decoder chooses a spectral cutoff from the stored data, estimates the
       corresponding head without reconstructing $A$, and estimates a certified high-stable-rank
       tail from unbiased spectral moments.  This gives the explicit logarithmic saving below
       $n^2$ for Schatten--1.
 \item A three-block variant estimates the exact leading Schatten power directly from the row
       sketch and approximates $x^{p/2}$ on the certified tail.  It gives a polylogarithmic saving
       below $n^2$ for every fixed finite non-even $p>0$.
 \item Combining the non-even bounds with the known
       tight results for positive even $p$ and $p=\infty$ resolves the polynomial order of
       $k_{p,\eps}(n)$ for every fixed $0<p\le\infty$, as summarized in
       \eqref{eq:classification}.
 \item The non-even-$p$ hard pairs, combined with recent turnstile linearization and
       integer-to-real lifting theorems, give the bit-space lower bound in
       Corollary~\ref{cor:streaming}.  The problem-specific work is to prove that scaling, integer
       rounding, truncation, and an additional discrete-Gaussian mollification preserve both the
       norm gap and real-sketch indistinguishability.
\end{itemize}

The rate calculations are explicit.  Constants are not optimized, and computational efficiency
of the decoder is not part of the exact-real sketching model.

\subsection{Organization and proof structure}

The main body constructs two random matrices whose nuclear norms are well separated but whose
images under every low-rank linear map are statistically close.  It then gives one fixed Gaussian
sketch and a decoder that splits the spectrum into a low-rank head and a well-spread tail.
Appendix~\ref{app:noneven} replaces the square-root endpoint functional by $x^{p/2}$ and extends
the lower bound to every finite non-even $p>0$; Appendix~\ref{app:noneven-upper} extends the
head--tail upper-bound architecture to the same range.

Sections~1--6 present the main arguments and their implications.  The specialized scalar construction, tensor estimate,
and Gaussian calculation enter through precise intermediate statements.  Appendices
\ref{app:scalar-path}--\ref{app:noneven-upper} give their complete proofs, including the scalar and
particle-path constructions, the random-rotation estimate, the Fisher-information calculation,
the detailed verification of the decoders, the transfer to finite-memory turnstile streams, and
the general fractional-power upper bound.
The standard orthogonal Weingarten formula is the
only external group-integration formula used in Appendices~\ref{app:tensors} and~\ref{app:score};
all deductions needed from it are elementary pairing and contraction calculations.

The scope of the self-contained exposition is as follows.

\begin{itemize}
 \item Every object, parameter, and probability statement used in the final rate calculation is
       defined here.
 \item Elementary algebra, calculus, linear algebra, testing, and optimization steps are proved
       here.
 \item General published tools---Steinitz rearrangement, the Markov brothers' polynomial
       inequality, the McCarthy Schatten quasi-norm inequality, the orthogonal Weingarten formula, the
       rotation-group spectral gap,
       Gaussian covariance concentration, polynomial approximation, Gaussian cycle variance, and
       inverse-Wishart expectation, together with the mollified turnstile-to-sketch transfer and
       integer-to-real Gaussian lifting theorems---are quoted at the exact strength used and cited
       at the point of application.  All problem-specific deductions from them are proved here.
\end{itemize}

The main proof chains are
\[
\begin{aligned}
\text{moment-matched spectra}
&\to\text{random singular vectors}\to\text{Gaussian smoothing}\\
&\to\text{small total variation}\to\text{lower bound},\\[1mm]
\text{multiscale cutoff}
&\to\text{implicit head recovery}+\text{stable-rank tail estimation}\\
&\to\text{Schatten--1 upper bound},\\[1mm]
\text{power-dependent cutoff}
&\to\text{exact spectral head}+\text{fractional-power tail}\\
&\to\text{non-even-$p$ upper bounds},\\[1mm]
\text{real-sketch hard pair}
&\to\text{integer rounding}+\text{discrete-Gaussian mollification}\\
&\to\text{turnstile transfer}+\text{lifting}\to\text{bit-space lower bound}.
\end{aligned}
\]

\section{Technical overview}\label{sec:overview}

This section explains the lower and upper bounds and the turnstile consequence without using
their most technical details.  The complete Schatten--1 proofs begin in
Sections~\ref{sec:prelim}--\ref{sec:upper}; the non-even-$p$ lower extension, streaming transfer,
and general upper extension are completed in Appendices~\ref{app:noneven},
\ref{app:streaming}, and~\ref{app:noneven-upper}, respectively.

\subsection{Lower-bound overview}

Fix a deterministic $k$-dimensional measurement map.  After removing redundant output
coordinates, it is enough to consider an orthogonal projection onto a $k$-dimensional subspace of
the $n^2$-dimensional matrix space.  We construct two random matrices $X_0,X_1$ with the following
properties:
\[
 \norm{X_0}_{\Sone}\quad\hbox{and}\quad\norm{X_1}_{\Sone}
 \quad\hbox{lie in disjoint relative-error ranges,}
\]
but the two projected distributions are close in total variation.  Le Cam's testing identity then
implies that no decoder can reliably decide which distribution generated the input, whereas a
valid nuclear-norm estimate would make that decision.  Yao's principle converts this deterministic
testing statement into a lower bound for randomized sketches.

The singular values are built from two scalar probability laws.  If $x$ denotes a squared singular
value, then $x^q$ controls an even spectral moment and $\sqrt x$ contributes to the nuclear norm.
Let $K\ge2$ be an integer, called the \emph{moment-matching order}.  It is the largest exponent
$q$ for which we require the two scalar moments $\int x^q\,d\nu_i(x)$ to agree.  We construct laws
$\nu_0,\nu_1$ that satisfy
\[
 \int x^q\,d\nu_0(x)=\int x^q\,d\nu_1(x)\quad(0\le q\le K),
 \qquad
 \int\sqrt x\,d\nu_0(x)\ne\int\sqrt x\,d\nu_1(x).
\]
Thus the two laws have identical moments of degrees $0,1,\ldots,K$, while their square-root
expectations differ.  The lower-bound optimization will choose
$K$ on the order of $\log n/\log\log n$.
Their probabilities need not be multiples of $1/n$, so they are not immediately empirical
spectra of $n\times n$ matrices.  A rounding-and-correction argument replaces them by two
$n$-point empirical laws and, more importantly, joins those laws by a path along which all first
$K$ moments remain constant.

At path time $t$, let $D_t$ be the diagonal matrix of square roots of the empirical atoms and set
$F_t=\sqrt n\,UD_tV^T$, where $U,V$ are independent Haar orthogonal matrices.  We finally add a
small Gaussian matrix.  The noise gives every measured distribution a smooth density, so the
total variation between the endpoints can be bounded by integrating the square root of Fisher
information along the path.

Moment preservation controls this integral through repeated integration by parts on the two
orthogonal groups, which transfers the path derivative to the Gaussian density.  Every tensor order through $2K$
cancels because the corresponding spectral moments are constant.  The first surviving order is
the odd number $m=2K+1$.  An odd-order random-rotation lemma shows that a rank-$k$ projection sees
at most, up to factors exponential in $K\log K$,
\[
 \left(\frac{k}{n^2}\right)^{K+1}
\]
of its squared energy.  Taking $K$ of order $\log n/\log\log n$ makes the endpoint distributions
indistinguishable whenever $k\le n^2/(\log n)^{A_\eps}$.

The same mechanism covers every fixed finite non-even exponent $p>0$.  For the Schatten--$p$
functional, the scalar endpoint statistic is $x^{p/2}$ in place of $\sqrt x$.  When $p/2$ is
not an integer, the support exponent in the scalar construction can be chosen so that one anchor
uniquely dominates this moment, producing a fixed relative gap while all integer moments through
degree $K$ still match.  The moment-preserving path and observation comparison are unchanged.
Gaussian smoothing is controlled by the Schatten triangle inequality for $p\ge1$, by its
$p$-power analogue for $0<p<1$, and, above $p=2$, by an operator-norm bound for the noise.  This
gives the $n^{2-o(1)}$ lower bound for all finite non-even $p$.  At a positive even exponent,
$x^{p/2}$ is itself one of the matched polynomial moments, explaining exactly why this
construction has no gap there.

\subsection{Upper-bound overview}

The upper bound separates the singular values into a head and a tail, but the useful cutoff is not
known before the input arrives.  We inspect cutoffs spaced by a block size $R$.  A deterministic
dichotomy shows that one inspected tail is either negligible in nuclear norm or has stable rank
at least a constant multiple of $R$.

A Gaussian row sketch $GA$ identifies candidate right singular subspaces.  A regularized
covariance estimate shows that the selected rank-$h$ subspace leaves a residual
$C=A(I-P_h)$ whose nuclear norm differs from the optimal tail by only a small multiple of
$\norm A_{\Sone}$.  A second independent block $WA$ certifies, using only stored data, whether
$C$ has high stable rank.  If no cutoff is certified, the last residual is negligible.

The head $AQ_h$ is estimated from a fresh right sketch $AR_0$ by Gaussian regression.  For a
certified residual, high stable rank bounds its normalized squared singular values in a known
interval.  We approximate the square-root function on that interval by a degree-$K$ polynomial.
Each power-sum coefficient of this polynomial is estimated without bias from a fourth independent
block $YA$ using Gaussian cycle statistics.  The degree can be taken of order
$\log n/\log\log n$, while the number of Gaussian rows remains below $n$ by a logarithmic factor.

Balancing the head cost $n^2\log b/b$ against the tail cost
$n^{2-c/b}\operatorname{poly}(b)$ at
$b\asymp\log n/\log\log n$ gives
\[
 k=O_\eps\!\left(\frac{n^2(\log\log n)^2}{\log n}\right).
\]
All cutoff selection occurs in the decoder; the four stored maps are fixed and linear.

For a general fixed non-even $p$, the certified tail uses the normalized statistic
\[
 q_p=\frac1n\sum_i\lambda_i^{p/2},
 \qquad
 \norm C_{\Sp{p}}^p=n^{1-p/2}\norm C_F^p q_p.
\]
A Bernstein polynomial for $x^{p/2}$ reduces its estimation to the same integer spectral moments.
For the head, the generalized decoder reads the leading power sum directly from the singular
values of $GA$.  The regularized covariance event proves simultaneously that this is close to the
exact leading power sum and that $A(I-P_h)$ is close, in Schatten--$p$ power, to the exact tail.
This avoids combining the two projected matrices through a norm-additivity assertion, which would
be false away from $p=1$.  A conservative optimization gives the upper bound in
Corollary~\ref{cor:noneven-upper}; Appendix~\ref{app:noneven-upper} contains the proof.

\subsection{Turnstile consequence}

An exact-real sketch-dimension lower bound does not by itself imply a bit-space lower bound for
streaming algorithms.  To make that transfer, the Gaussian-smoothed hard pair is scaled to a
bounded range, rounded entrywise to integer matrices, truncated to bounded support, and mollified
by discrete Gaussian noise.  The norm gap and measurement indistinguishability survive these
operations.  The recent polynomial-length turnstile-to-sketch theorem and integer-to-real lifting
framework can then be applied to any one-pass unit-update algorithm estimating a fixed finite
non-even Schatten exponent $p>0$.  For streams of length at most $W$, with $W$ a sufficiently
large polynomial in $n$, this yields the bit-space lower bound
\[
 \Omega_{p,\eps}\!\left(\frac{n^2\log W}{(\log n)^{B_{p,\eps}}}\right)
\]
For polynomial $W$, this is $n^{2-o(1)}$ bits.
Appendix~\ref{app:streaming} verifies the robustness and parameter compatibility needed for this
transfer.

\section{Preliminaries}\label{sec:prelim}

For a matrix $B$ of rank at most $r$, Cauchy--Schwarz gives
\begin{equation}
 \norm{B}_{\Sone}\le\sqrt r\,\norm{B}_F.                    \tag{3.1}\label{eq:3.1}
\end{equation}
More generally $\norm{B}_{\Sone}\le\sqrt n\,\norm{B}_F$.

Let $A_h$ be a best rank-$h$ approximation to $A$, obtained by keeping its first $h$ singular
directions.  Define
\begin{equation}
 T_h(A)=\sum_{i>h}\sigma_i(A),\qquad
 F_h(A)^2=\sum_{i>h}\sigma_i(A)^2,\qquad
 a_h=\sigma_{h+1}(A).                                      \tag{3.2}\label{eq:3.2}
\end{equation}
Here and throughout, singular values past the matrix dimension are defined to be zero; in
particular, $a_n=0$.
We record the part of the singular-value truncation theorem used later.  If $B$ has rank at most
$h$, then
\begin{equation}
 \sigma_i(A-B)\ge\sigma_{i+h}(A)                          \tag{3.2a}\label{eq:3.2a}
\end{equation}
whenever the right side is defined.  To see this, let $E$ be the span of the first $i+h$ right
singular vectors of $A$.  Since $\ker B$ has codimension at most $h$, the intersection
$E\cap\ker B$ has dimension at least $i$.  On every unit vector in this intersection,
$(A-B)x=Ax$ and $\norm{Ax}\ge\sigma_{i+h}(A)$.  The max--min characterization of singular
values gives \eqref{eq:3.2a}.  Summing \eqref{eq:3.2a}, or summing its squares, shows that the truncated singular
value decomposition $A_h$ minimizes both every Ky--Fan norm
$\sum_{i=1}^r\sigma_i(A-B)$ and the Frobenius norm of $A-B$.

We also use that multiplication by an orthogonal projection cannot increase an individual
singular value.  Indeed, the elementary approximation formula
$\sigma_i(M)=\min_{\rank R<i}\norm{M-R}_{\op}$ gives
\[
 \sigma_i(AP)\le\norm{AP-A_{i-1}P}_{\op}
 \le\norm{A-A_{i-1}}_{\op}=\sigma_i(A).
\]

The stable rank of a nonzero matrix is
\begin{equation}
 \sr(B)=\frac{\norm{B}_F^2}{\norm{B}_{\op}^2}.              \tag{3.3}\label{eq:3.3}
\end{equation}
It lies between $1$ and $\rank(B)$.  A large stable rank means that the Frobenius energy is not
concentrated in one singular direction.

For probability laws $\mu,\nu$, total variation is
\begin{equation}
 \TV(\mu,\nu)=\sup_E\abs{\mu(E)-\nu(E)}
 =\frac12\int\abs{d\mu-d\nu}.                              \tag{3.4}\label{eq:3.4}
\end{equation}
If the two laws have equal prior probability, the best possible success probability of any label
test is $(1+\TV(\mu,\nu))/2$; this is Le Cam's binary testing identity
\cite[Proposition~2.3.1]{Duchi}.

All logarithms are natural.  Every constant in the proof has a fixed name: subscripts identify
the estimate in which the constant is introduced, and the value never changes afterward.  At the
end of a proof we may take the maximum of finitely many such constants and give that maximum a
single descriptive name.  The theorem-level constants $C_\eps,A_\eps$ may depend on the
fixed accuracy $\eps$, but never on $n$.

\section{Lower bound}\label{sec:lower}

\paragraph{Proof structure.}
The proof builds two random-matrix models with different nuclear norms but nearly identical
$k$-dimensional observations.  It has four stages:
\begin{enumerate}[leftmargin=2.2em]
 \item Normalize the sketch and reduce estimation to binary testing
       (Lemma~\ref{lem:yao}).
 \item Construct two spectra whose first $K$ power moments agree while their sums of square roots
       differ (Proposition~\ref{prop:scalar-pair} and Lemma~\ref{lem:path}).
 \item Randomize the singular vectors and add a small Gaussian perturbation, which creates smooth
       observation densities without closing the nuclear-norm gap
       (Lemma~\ref{lem:smoothing}).
 \item Join the two models by the moment-preserving path and bound the total variation of every
       $k$-dimensional observation.  The Fisher--Rao bound converts path speed into total
       variation, while Lemma~\ref{lem:twirl} bounds the portion visible to the measurement.
\end{enumerate}
The last step verifies the two hypotheses of Lemma~\ref{lem:yao}: separated nuclear-norm ranges
and small total variation after every rank-$k$ linear measurement.

\subsection{Step 1: reduce an arbitrary sketch to a projection}

Identify $\R^{n\times n}$ with $\R^{n^2}$ using the Frobenius inner product
$\ip{A}{B}_F=\tr(A^TB)$.  A deterministic linear map
$S:\R^{n^2}\to\R^k$ is represented by a $k\times n^2$ matrix.  If its rank is $r\le k$, it
has a factorization
\[
 S=TL,\qquad
 L\in\R^{r\times n^2},\quad T\in\R^{k\times r},\quad LL^*=I_r.
\]
Here $L:\R^{n^2}\to\R^r$ has orthonormal rows, and
$T:\R^r\to\R^k$ is one-to-one (injective): $Tx=0$ implies $x=0$.  When $r<k$, $T$ is not
onto all of $\R^k$; it is a bijection only from $\R^r$ to the $r$-dimensional subspace
$\operatorname{range}(T)\subseteq\R^k$.  Consequently $T^{-1}$ is well defined on that range,
and applying it to $SA=T(LA)$ recovers $LA$ without losing information.  (For example, this
factorization follows immediately from the compact singular value decomposition of $S$.)
Therefore, for lower bounds it suffices to consider a \emph{row coisometry} $L$ with
$LL^*=I_r$.  The matrix $P=L^*L$ is an orthogonal projection of rank $r$ on the $n^2$-dimensional
matrix space.

The following notation will be used in the testing reduction.  An \emph{input law}
$\Pi$ is a probability distribution on the matrix space $\R^{n\times n}$.  Writing
$A\sim\Pi$ means that the random input matrix $A$ is sampled from this distribution.  If
$L:\R^{n\times n}\to\R^r$ is a linear measurement map, then
\begin{equation}
 L_\#\Pi\defeq\operatorname{Law}(LA) \quad\text{for }A\sim\Pi. \tag{4.0}\label{eq:4.0}
\end{equation}
This is called the \emph{pushforward} of $\Pi$ by $L$.  Equivalently, for every measurable set
$B\subseteq\R^r$,
\[
 (L_\#\Pi)(B)=\Pi\{A\in\R^{n\times n}:LA\in B\}.
\]
Thus $\TV(L_\#\Pi_0,L_\#\Pi_1)$ measures how well one can distinguish the two random-matrix
models after seeing only the measurement vector $LA$.

\paragraph{Statistical reduction.}
The theorem concerns every randomized sketching algorithm, whereas the construction below gives
two probability distributions on inputs.  Lemma~\ref{lem:yao} translates between these forms.  If
a valid estimate of a nonnegative measurable functional would reveal which input distribution
was used, but the measurement vector cannot reveal that label with comparable probability, then
such an estimator cannot exist.  The lower bound below applies the lemma to the nuclear norm.

\begin{lemma}[Yao--Le Cam reduction]\label{lem:yao}
Let $\Phi:\R^{n\times n}\to[0,\infty)$ be measurable, and let $\Pi_0,\Pi_1$ be two probability
laws on $\R^{n\times n}$.  For $i\in\{0,1\}$, let $A_i\sim\Pi_i$.  Suppose there are
deterministic sets $I_0,I_1\subseteq[0,\infty)$ such that
\[
 \Prb\{\Phi(A_i)\in I_i\}\ge1-\zeta
 \quad(i=0,1),
\]
and their multiplicative enlargements
\[
 (1\pm\eps)I_i
 \defeq\{y\ge0:(1-\eps)x\le y\le(1+\eps)x
                 \text{ for some }x\in I_i\}
\]
are disjoint.  If
for every rank-at-most-$k$ row coisometry $L$,
\[
 \TV(L_\#\Pi_0,L_\#\Pi_1)\le\delta
 \quad\text{and}\quad \delta+2\zeta<\frac13,
\]
then there is no randomized $k$-dimensional linear sketch $(S_\omega,d_\omega)$ satisfying
\[
 \Prb_\omega\!\left\{
  (1-\eps)\Phi(A)\le d_\omega(S_\omega A)\le(1+\eps)\Phi(A)
 \right\}\ge\frac23
 \qquad\text{for every fixed }A\in\R^{n\times n}.
\]
\end{lemma}

This is the standard Yao minimax reduction followed by Le Cam's binary testing bound.  For
completeness, Appendix~\ref{app:yao-proof} gives the proof with the constants used here.

\subsection{Step 2: build two spectra with equal low moments and different nuclear norms}

\paragraph{Scalar moment design.}
For the hard matrices constructed below, $x$ represents a squared singular value.  Consequently
$x^q$ controls even-degree matrix moments, while $\sqrt{x}$ contributes to the nuclear norm.
We therefore seek two scalar laws that agree on $x^q$ for $q\le K$ but disagree on
$\sqrt{x}$.  Proposition~\ref{prop:scalar-pair} constructs such a pair of scalar laws.  It does
not yet construct $n\times n$ matrices; that conversion is the role of Lemma~\ref{lem:path}.

\begin{proposition}[Moment-matched scalar hard pair]\label{prop:scalar-pair}
For every $0<\eps<1$, there exist an integer $M_\eps\ge6$ and a constant $\eta_\eps>0$ such
that, for every integer $K\ge2$, there are two probability laws $\nu_0,\nu_1$ on
$[0,B_K]$, each supported on at most $K+2$ points, with
\begin{equation}
 B_K\le2(K+1)^{M_\eps},\qquad
 \int x\,d\nu_0(x)=\int x\,d\nu_1(x)=1,                  \tag{4.P1}\label{eq:4.P1}
\end{equation}
\begin{equation}
 \int x^q\,d\nu_0(x)=\int x^q\,d\nu_1(x)
 \qquad(0\le q\le K),                                    \tag{4.P2}\label{eq:4.P2}
\end{equation}
and, writing $h_b=\int\sqrt{x}\,d\nu_b(x)$,
\begin{equation}
 \eta_\eps=\frac{1-\eps}{4},\qquad
 \frac{h_1-h_0}{h_1+h_0}
 \ge\frac{1+\eps}{2}=\eps+2\eta_\eps,
 \qquad h_1\ge\frac1{2\sqrt2}.                          \tag{4.P3}\label{eq:4.P3}
\end{equation}
The proof gives an explicit choice of $M_\eps$ in \eqref{eq:4.P4}.  In particular, $M_\eps$ grows only
logarithmically in $1/(1-\eps)$ as $\eps$ approaches one.
\end{proposition}

\paragraph{Proof sketch.}
Choose \(K+2\) ordered nodes \(x_j=j^{M_\eps}\).  The standard Lagrange-interpolation
weights on these nodes annihilate every polynomial of degree at most \(K\), and their signs
alternate.  Normalizing the positive and negative weights separately therefore gives two
probability laws with matching moments.  When \(M_\eps\) is large, one law is concentrated near
the node \(0\), whereas the other is concentrated near the node \(1\).  Their square-root moments
are consequently separated.  A common rescaling makes their first moment equal to one without
affecting moment matching or the atom at zero.  Appendix~\ref{app:scalar-path} proves all weight
estimates, gives the explicit choice of \(M_\eps\), and verifies \eqref{eq:4.P1}--\eqref{eq:4.P3}.

\paragraph{From probability laws to finite spectra.}
Proposition~\ref{prop:scalar-pair} gives two abstract probability laws $\nu_0,\nu_1$.  They have
the desired moment matching and square-root-moment gap, but they are not yet spectra of
$n\times n$ matrices.  Indeed, if
\[
 D=\operatorname{diag}(\sqrt{x_1},\ldots,\sqrt{x_n}),
\]
then the empirical law of its squared diagonal entries is
\[
 \frac1n\sum_{i=1}^n\delta_{x_i}.
\]
Every probability in this law is therefore an integer multiple of $1/n$, whereas the atom
probabilities in $\nu_b$ need not have that form.  Rounding those probabilities to multiples of
$1/n$ would generally destroy their exact moment matching.

A further requirement arises from the continuous path used to compare the two observed
distributions.  If $p_t$ denotes the density at path time $t$, the elementary
Fisher--Rao bound is
\begin{equation}
 \frac12\int\abs{p_1(y)-p_0(y)}\,dy
 \le \frac12\int_0^1\sqrt{I_t}\,dt,
 \qquad
 I_t=\int \frac{\dot p_t(y)^2}{p_t(y)}\,dy.              \tag{4.8a}\label{eq:4.8a}
\end{equation}
Appendix~\ref{app:fisher-tv} derives this inequality using only the fundamental theorem of
calculus and Cauchy--Schwarz.  The integral charges every change made along the path.  Therefore
it is not enough for the two endpoints to have matching moments: a low moment could vary in the
middle and then return to its starting value, producing a large and easily observed change.

Lemma~\ref{lem:path} prevents this problem by keeping the first $K$ moments exactly constant at
every time.  It also replaces arbitrary atom probabilities by exactly $n$ equal weights, as a
matrix spectrum requires, and bounds the ordinary length of the resulting particle path.  More
precisely, it produces
\begin{enumerate}[label=(\roman*),leftmargin=2.2em]
 \item actual $n$-point empirical laws, hence actual diagonal $n\times n$ matrices;
 \item exact moment matching throughout the path, as stated in \eqref{eq:4.9b};
 \item endpoint square-root moments close to those of $\nu_0,\nu_1$, so the nuclear-norm gap is
       retained, as stated in \eqref{eq:4.10a}; and
 \item the path-length bound \eqref{eq:4.10}, which controls the integral in \eqref{eq:4.8a}.
\end{enumerate}
The precise conversion from these four properties to small total variation is deferred until
Lemma~\ref{lem:path-tv}, after the hard matrix distributions have been defined.  The Steinitz
lemma below is one ingredient in the path construction.

\paragraph{Ordering the particle moves.}
In the path construction, spectral entries are changed one at a time.  Each change creates a
vector of $K$ moment changes.  The sum of the base-particle vectors is the small rounding error
(denoted $e^{\mathrm{end}}$ in the proof), because the control particles supply the opposite change.  Subtracting the
average vector produces a zero-sum family to which the Steinitz lemma applies.  It supplies an
order in which every centered partial sum remains bounded; adding back the small average drift
keeps the original partial sums bounded as well.  The movable controls can then cancel the current
partial sum without moving far.

\begin{lemma}[Steinitz rearrangement lemma {\cite[Theorem~1]{Steinitz}}]\label{ext:steinitz}
Let $d,N\ge1$, let $\norm{\cdot}$ be any norm on $\R^d$, and let
$v_1,\ldots,v_N\in\R^d$ satisfy
\[
 \norm{v_i}\le1\quad(1\le i\le N),
 \qquad
 \sum_{i=1}^Nv_i=0.
\]
Then there exists a permutation $\pi$ of $\{1,\ldots,N\}$ such that every partial sum satisfies
\begin{equation}
 \norm{\sum_{i=1}^m v_{\pi(i)}}\le d
 \qquad(1\le m\le N).                                     \tag{4.S}\label{eq:4.S}
\end{equation}
\end{lemma}

Lemma~\ref{lem:path} is a general conversion result.  In the lower-bound construction, its
input laws $\nu_0,\nu_1$ are specifically those supplied by
Proposition~\ref{prop:scalar-pair}.

\begin{lemma}[Equal-weight packet and moment-preserving path]\label{lem:path}
There is a fixed absolute constant $C_{\mathrm{path}}$ with the following property.  Let
$B\ge1$ and $K\ge2$, and let $\nu_0,\nu_1$ be supported on at most $K+2$ points in $[0,B]$ and have equal moments through
degree $K$.  Fix $0<\eta<1/10$.  If
\begin{equation}
 n\ge C_{\mathrm{path}}\eta^{-1}BK^5\log^2(K+1),          \tag{4.9}\label{eq:4.9}
\end{equation}
there exist $n$ particle paths
\[
 x_i:[0,1]\longrightarrow[0,B],\qquad 1\le i\le n,
\]
which are piecewise continuously differentiable after taking square roots.  Define the equal-weight
empirical law and the associated $n\times n$ diagonal matrix by
\begin{equation}
 \widehat\nu_t=\frac1n\sum_{i=1}^n\delta_{x_i(t)},
 \qquad
 D_t=\operatorname{diag}\bigl(\sqrt{x_1(t)},\ldots,\sqrt{x_n(t)}\bigr).
                                                               \tag{4.9a}\label{eq:4.9a}
\end{equation}
For every $0\le q\le K$, the moment
\begin{equation}
 \int x^q\,d\widehat\nu_t(x)
 =\frac1n\sum_{i=1}^nx_i(t)^q
 =\frac1n\tr D_t^{2q}                                      \tag{4.9b}\label{eq:4.9b}
\end{equation}
is independent of $t$.  Thus the endpoint empirical laws
$\widehat\nu_0$ and $\widehat\nu_1$ have exactly equal moments through degree $K$.
Furthermore,
\begin{equation}
 \mathcal L_K\defeq\int_0^1
 \left(\frac1n\tr\dot D_t^2\right)^{1/2}dt
 \le C_{\mathrm{path}}\sqrt{nB}
      +\frac{C_{\mathrm{path}}BK^2\log(K+1)}{\sqrt\eta}. \tag{4.10}\label{eq:4.10}
\end{equation}
Writing $h(\mu)=\int\sqrt{x}\,d\mu(x)$, the endpoints also satisfy
\begin{equation}
 \left|h(\widehat\nu_b)-h(\nu_b)\right|
 \le C_{\mathrm{path}}\eta
    +\frac{C_{\mathrm{path}}BK^2\log(K+1)}{n},
 \qquad b\in\{0,1\}.                                    \tag{4.10a}\label{eq:4.10a}
\end{equation}
Thus the endpoint square-root moments change by $o(1)$ whenever $\eta=\eta_n\to0$ and \eqref{eq:4.9}
holds.
\end{lemma}

\paragraph{Proof sketch.}
First approximate each input law by \(n\) equal-weight particles.  Reserve \(K\) small groups of
particles as adjustable controls, one for each nonconstant moment constraint.  Integer rounding
of the remaining atom counts creates a small \(K\)-dimensional moment error; moving the \(K\)
control groups cancels it exactly.  Next move the ordinary particles from the first endpoint to
the second one at a time.  The Steinitz lemma orders these moves so that the accumulated moment
error stays small, and the controls continuously cancel that error throughout the path.  The
resulting base-particle motion gives the first term in \eqref{eq:4.10}, while the control motion gives the
second.  Appendix~\ref{app:scalar-path} supplies the complete construction and verifies the
endpoint, regularity, and length bounds.

Apply Lemma~\ref{lem:path} to the laws furnished by Proposition~\ref{prop:scalar-pair}, with
\begin{equation}
 c_{\mathrm{low}}\defeq\frac1{100},\qquad
 K=\left\lfloor c_{\mathrm{low}}\frac{\log n}{\log\log n}\right\rfloor,
 \qquad B\le2(K+2)^{M_\eps},\qquad
 \eta_{\mathrm{path}}(n)
 =\min\left\{\frac1{20},\frac1{\log(en)}\right\}.        \tag{4.17}\label{eq:4.17}
\end{equation}
Since $M_\eps$ is fixed once $\eps$ is fixed, the right side of \eqref{eq:4.9} is only a fixed power of
$\log n$; hence \eqref{eq:4.9} holds for all sufficiently large $n$.  The tensor order
$m=2K+1$ satisfies
$4m\log(4m)\le\tfrac14\log n$ for all sufficiently large $n$ because
$c_{\mathrm{low}}=1/100$.  Thus
$(4m)^{4m}\le n^{1/4}\le n/2$, as required by
Lemma~\ref{lem:twirl}.  Moreover,
$\eta_{\mathrm{path}}(n)\to0$, so \eqref{eq:4.10a} shows that the endpoint square-root moments differ
from those in Proposition~\ref{prop:scalar-pair} by $o(1)$.  Write
$x_i^{(b)}=x_i(b)$ and $D_b=D_t|_{t=b}$ for the two endpoints $b\in\{0,1\}$ of the path
defined in \eqref{eq:4.9a}.  Write
\[
 \mathcal O_n=\{Q\in\R^{n\times n}:Q^TQ=I_n\},
\]
and take $U,V$ independently from the Haar-uniform probability law on $\mathcal O_n$.  Define
\begin{equation}
 F_b=\sqrt n\,UD_bV^T.                                     \tag{4.18}\label{eq:4.18}
\end{equation}
Then
\[
 \norm{F_b}_{\Sone}=\sqrt n\sum_i\sqrt{x_i^{(b)}}
\]
is deterministic.  By \eqref{eq:4.P3}, the scalar relative gap is at least
$(1+\eps)/2=\eps+2\eta_\eps$; consequently these two values have relative gap at least
$\eps+\eta_\eps$ for all sufficiently large $n$.
The identities $n^{-1}\sum_i(x_i^{(0)})^q=n^{-1}\sum_i(x_i^{(1)})^q$ for $q\le K$ imply that
the tensor moments of $F_0,F_1$ agree through total degree $2K$.

\subsection{Step 3: add a small Gaussian channel}

The orbit laws in \eqref{eq:4.18} are singular.  Add a common Gaussian perturbation:
\begin{equation}
 X_b=aF_b+\sigma Z,                                        \tag{4.19}\label{eq:4.19}
\end{equation}
where $Z$ has independent $N(0,1)$ entries, $Z$ is independent of $(U,V)$, and
$a,\sigma>0$ are constants depending only on $\eps$.  Let
$\Pi_b=\operatorname{Law}(X_b)$ for $b\in\{0,1\}$; these are the two input laws
to which Lemma~\ref{lem:yao} will be applied.

\paragraph{Gaussian regularization.}
The matrices $F_b=\sqrt n\,UD_bV^T$ lie on lower-dimensional rotation orbits, so their laws do
not have ordinary densities on the full matrix space.  The later score and Fisher-information
calculation differentiates densities, so we add Gaussian noise.  Lemma~\ref{lem:smoothing}
verifies both properties needed after this regularization: for sufficiently small $\sigma/a$,
the nuclear-norm gap survives with high probability, and the measured noise is a standard
Gaussian vector $G_k$ because $L$ has orthonormal rows.

\begin{lemma}[Smoothing preserves the norm gap]\label{lem:smoothing}
Let $m_b=n^{-1}\tr D_b$ and suppose
$(m_1-m_0)/(m_1+m_0)\ge\eps+\eta$.  For $0<\zeta<1$, if
\begin{equation}
 \frac\sigma a\le\frac{\eta(m_0+m_1)\sqrt\zeta}{8},        \tag{4.20}\label{eq:4.20}
\end{equation}
then, with probability at least $1-\zeta$ under either label, the nuclear norms of $X_0$ and
$X_1$ lie in deterministic intervals whose $(1\pm\eps)$ enlargements are disjoint.  Moreover, for
every row coisometry $L$,
\begin{equation}
 LX_b=aLF_b+\sigma G_k,\qquad G_k\sim N(0,I_k).            \tag{4.21}\label{eq:4.21}
\end{equation}
Here $G_k$ is independent of $(U,V)$.
\end{lemma}

The lemma follows from the standard comparison
$\norm Z_{\Sone}\le\sqrt n\norm Z_F$, Markov's inequality, and the fact that a row coisometry
maps a standard Gaussian vector to a standard Gaussian vector.  For completeness,
Appendix~\ref{app:smoothing-proof} gives the details with the constants in \eqref{eq:4.20}.

The smoothing parameters can be chosen explicitly.  Take $\zeta=1/100$ and $a=1$.
Proposition~\ref{prop:scalar-pair} gives $h_1\ge1/(2\sqrt2)$, while \eqref{eq:4.10a} changes each endpoint
square-root moment by $o(1)$.  Hence, for all sufficiently large $n$, the empirical quantities
$m_b=n^{-1}\tr D_b$ satisfy $m_0+m_1\ge1/(4\sqrt2)$.  The fixed choice
\begin{equation}
 \sigma\defeq\frac{\eta_\eps\sqrt\zeta}{64\sqrt2}       \tag{4.20a}\label{eq:4.20a}
\end{equation}
therefore satisfies \eqref{eq:4.20}, with margin parameter $\eta=\eta_\eps$.  In particular,
$a/\sigma$ depends only on $\eps$ and not on $n$.

\subsection{Step 4: compare every low-dimensional observation}

The main proof uses the following random-rotation inequality.  Appendix~\ref{app:tensors}
derives it from the exact orthogonal Weingarten formula, an elementary trace decomposition, and
explicit linear-algebra estimates.
Appendix~\ref{app:score} combines the inequality with the score calculation to obtain the
total-variation bound used below.

Let $e_1,\ldots,e_n$ be the standard basis of $\R^n$.  We identify a matrix with a vector having
one row coordinate and one column coordinate:
\begin{equation}
 \Phi(A)=\sum_{i,j=1}^n A_{ij}\,e_i\otimes e_j
 \in\R^n\otimes\R^n.                                    \tag{4.21a}\label{eq:4.21a}
\end{equation}
Under this convention, the two-sided rotation $A\mapsto UAV^T$ is represented by $U\otimes V$.
Indeed,
\begin{align*}
 (U\otimes V)\Phi(A)
 &=\sum_{i,j}A_{ij}(Ue_i)\otimes(Ve_j)\\
 &=\sum_{a,b}\left(\sum_{i,j}U_{ai}A_{ij}V_{bj}\right)e_a\otimes e_b
 =\Phi(UAV^T).                                           \tag{4.21b}\label{eq:4.21b}
\end{align*}
Thus $U$ rotates the row index and $V$ rotates the column index.

Now let $m$ be a positive integer, called the \emph{tensor order}, and consider $m$ matrix
factors.  Throughout this argument, $m$ denotes this integer order.  On a simple tensor
$\Phi(A_1)\otimes\cdots\otimes\Phi(A_m)$, applying the same two-sided rotation to every matrix
factor gives
\begin{align*}
 &\Phi(UA_1V^T)\otimes\cdots\otimes\Phi(UA_mV^T)\\
 &\qquad=(U\otimes V)^{\otimes m}
       \bigl(\Phi(A_1)\otimes\cdots\otimes\Phi(A_m)\bigr).       \tag{4.21c}\label{eq:4.21c}
\end{align*}
The identity extends to every element of $(\R^{n\times n})^{\otimes m}$ by linearity.  Here
$(U\otimes V)^{\otimes m}$ is the tensor product of $m$ copies of the $n^2\times n^2$ operator
$U\otimes V$; it is not an ordinary matrix power.  The following estimate bounds the expected
squared norm of the projection of a randomly rotated tensor onto the same rank-$k$ subspace in
all $m$ tensor positions.

\paragraph{Measurement contraction at the first unmatched order.}
After moment matching cancels all tensor degrees through $2K$, the score calculation first sees
the odd tensor order $m=2K+1$.  The required estimate is that an arbitrary rank-$k$ measurement
captures little of this tensor after the independent random rotations $U,V$.  The factor
$(k/n^2)^{(m+1)/2}$ in the following lemma supplies this decay.  Appendix~\ref{app:tensors}
proves the lemma by decomposing the row and column indices into trace-free blocks and evaluating
the surviving Haar pairings directly.

\begin{lemma}[Odd two-sided rotation bound]\label{lem:twirl}
Let $\mathcal H=\R^{n\times n}$ and equip $\mathcal H^{\otimes m}$ with the Frobenius norm,
denoted by $\norm{\cdot}$.  There is an absolute constant $C_{\mathrm{ten}}\ge1$ with the
following property.  Let $m$ be odd and suppose $(4m)^{4m}\le n/2$.  For every tensor
$T\in\mathcal H^{\otimes m}$ and every rank-$k$ orthogonal projection $P$ on $\mathcal H$,
\begin{equation}
 \E_{U,V}\norm{P^{\otimes m}(U\otimes V)^{\otimes m}T}^2
 \le (C_{\mathrm{ten}}m)^{C_{\mathrm{ten}}m}
      \left(\frac{k}{n^2}\right)^{(m+1)/2}\norm{T}^2.     \tag{4.22}\label{eq:4.22}
\end{equation}
\end{lemma}

\begin{proof}
We give the main steps and leave the complete indexed calculation to
Appendix~\ref{app:tensors}.

For $d\ge0$, let $\mathcal T_d\subseteq(\R^n)^{\otimes d}$ be the trace-free subspace: setting
any two tensor indices equal and summing over their common value gives zero.  A partial matching
is a collection of disjoint pairs of positions.  A block label
$\alpha=(\pi_R^\alpha,\pi_C^\alpha)$ consists of one partial matching of the $m$ row positions
and one of the $m$ column positions.  If
\[
 p_\alpha=m-2\abs{\pi_R^\alpha},\qquad
 q_\alpha=m-2\abs{\pi_C^\alpha},\qquad
 W_\alpha=\mathcal T_{p_\alpha}\otimes\mathcal T_{q_\alpha},
\]
then $p_\alpha$ and $q_\alpha$ are the numbers of unmatched row and column positions.  The map
$S_\alpha:W_\alpha\to\mathcal H^{\otimes m}$ retains the entries of its input in these unmatched
positions and, for every matched row pair $\{r,s\}$, adds the coordinate factor
$n^{-1/2}\mathbf1_{\{i_r=i_s\}}$; matched column pairs are treated identically.  Thus
$S_\alpha$ embeds a tensor with $p_\alpha$ row and $q_\alpha$ column indices into one with $m$ of
each.

Let $\alpha=1,\ldots,D$ run over all pairs of partial matchings.  The stable pairing-frame
proposition in Appendix~\ref{app:tensors} (Proposition~\ref{ext:pairing-frame}) proves
\[
 D\le(4m)^{4m},\qquad
 \frac12\sum_\alpha\norm{z_\alpha}^2
 \le\norm{\sum_\alpha S_\alpha z_\alpha}^2
 \le\frac32\sum_\alpha\norm{z_\alpha}^2.
\]
Indeed, the elementary trace decomposition shows that the block images span; each $S_\alpha$ is
an isometry; and distinct blocks satisfy
$\norm{S_\alpha^*S_\beta}_{\op}\le1/n$.  Here $S_\alpha^*$ is the adjoint of $S_\alpha$.
Since $D\le n/2$, the sum of all off-diagonal Gram terms is at most one half of the sum of the
diagonal terms, which gives the displayed inequalities.  Consequently every $T$ can be written
as
\[
 T=\sum_{\alpha=1}^D S_\alpha z_\alpha,\qquad
 \sum_{\alpha=1}^D\norm{z_\alpha}^2\le2\norm T^2.
\]

We next estimate one block.  Fix $\alpha$ and $h\in W_\alpha$, abbreviate
$p=p_\alpha$, $q=q_\alpha$, and set $R_{U,V}=U^{\otimes p}\otimes V^{\otimes q}$.  The exact
orthogonal Weingarten expansion, together with trace-freeness, gives the following operator
inequality.  Here $hh^*$ is the rank-one operator $x\mapsto h\langle h,x\rangle$, and
$A\preceq B$ means that $B-A$ is positive semidefinite:
\[
 \E_{U,V}R_{U,V}hh^*R_{U,V}^*
 \preceq c_{p,q}\norm h^2 I_{W_\alpha},\qquad
 c_{p,q}=4(p!q!)^2n^{-(p+q)}.
\]
The reason for the constant is simple: trace-freeness eliminates every pairing that joins two
indices within the same copy of the tensor.  In each of the two row pairing slots only the $p!$
cross pairings remain, and in each column slot only the $q!$ cross pairings remain.  Their
coordinate permutations have norm one, while the two Weingarten coefficients contribute at
most $(2n^{-p})(2n^{-q})$.  Proposition~\ref{prop:tracefree-haar} proves this statement for an
arbitrary tensor $h$, including one entangled between its row and column factors.  That proof is
given over the complex numbers; since all operators have real coefficients, its conclusion
restricts to the real space $W_\alpha$ used here.

The matched-pair factors are unchanged by orthogonal rotations, as proved in
\eqref{eq:B.block-intertwining}; hence
$(U\otimes V)^{\otimes m}S_\alpha=S_\alpha R_{U,V}$.  To use this Haar estimate, Appendix
\ref{app:measurement-graph} first enlarges $S_\alpha$ to a map $\widetilde J_\alpha$ with the
same output formula but with no trace-free restriction on its unmatched input.  Since
$S_\alpha S_\alpha^*\preceq\widetilde J_\alpha\widetilde J_\alpha^*$, it obtains
\[
 \E_{U,V}\norm{P^{\otimes m}(U\otimes V)^{\otimes m}S_\alpha h}^2
 \le c_{p,q}\norm h^2
 \tr\!\left(P^{\otimes m}\widetilde J_\alpha\widetilde J_\alpha^*\right).
\]

The trace in the last display is evaluated by an elementary graph calculation.  There is one
vertex for each of the $m$ copies of $P$; matched row and column pairs form alternating paths and
cycles.  The complete calculation in \eqref{eq:B.13a}--\eqref{eq:B.14} proves the following
three component bounds:
\begin{align*}
 \text{cycle with $2\ell$ vertices}&\le k^\ell,\\
 \text{row--column path with $2\ell-1$ vertices}&\le k^\ell,\\
 \text{row--row or column--column path with $2\ell$ vertices}&\le nk^\ell.
\end{align*}
Let $o$ be the number of row--column paths, and let $u_R,u_C$ be the numbers of row--row and
column--column paths.  Endpoint counting gives
\[
 2u_R+o=p,\qquad 2u_C+o=q,
\]
and multiplication over all graph components gives
$n^{u_R+u_C}k^{(m+o)/2}$.  There are $(m-p)/2+(m-q)/2$ matched pairs, whose squared
$n^{-1/2}$ normalizations contribute $n^{-m+(p+q)/2}$.  Combining this with the factor
$n^{-(p+q)}$ in the trace-free Haar estimate, and using
$u_R+u_C=(p+q-2o)/2$, gives the one-block bound
\[
 \E_{U,V}\norm{P^{\otimes m}(U\otimes V)^{\otimes m}S_\alpha h}^2
 \le4(p!q!)^2\left(\frac{k}{n^2}\right)^{(m+o)/2}\norm h^2.
\]
Because $m$ is odd, $p$ and $q$ are positive odd integers.  The endpoint equations force $o$ to
be a positive odd integer, so $o\ge1$.  Since $k\le n^2$, the last display is at most
\[
 \gamma_m\norm h^2,\qquad
 \gamma_m=4(m!)^4\left(\frac{k}{n^2}\right)^{(m+1)/2}.
\]

It remains to recombine the blocks.  For the coefficient vector
$\mathbf z=(z_1,\ldots,z_D)$, let $G_P$ be the positive-semidefinite block operator defined by
\[
 \langle\mathbf z,G_P\mathbf z\rangle
 =\E_{U,V}\norm{P^{\otimes m}(U\otimes V)^{\otimes m}
                 \sum_\alpha S_\alpha z_\alpha}^2.
\]
The one-block bound above shows that every diagonal block of $G_P$ is at most $\gamma_m I$.
Cauchy--Schwarz for the positive-semidefinite form defined by $G_P$ then bounds every
off-diagonal block in operator norm by $\gamma_m$.  Therefore, using the coefficient bound above,
\begin{align*}
 \langle\mathbf z,G_P\mathbf z\rangle
 &\le\gamma_m\left(\sum_\alpha\norm{z_\alpha}\right)^2
 \le\gamma_mD\sum_\alpha\norm{z_\alpha}^2
 \le2\gamma_mD\norm T^2.
\end{align*}
Finally $D\le(4m)^{4m}$, so one absolute choice of $C_{\mathrm{ten}}$ absorbs
$2D\cdot4(m!)^4$ into $(C_{\mathrm{ten}}m)^{C_{\mathrm{ten}}m}$.  This proves
\eqref{eq:4.22}; Appendix~\ref{app:block-recombination} gives the complete recombination.
\end{proof}

\paragraph{From measurement contraction to Fisher--Rao length.}
We record explicitly how Lemma~\ref{lem:twirl} enters the statistical comparison.  Along the
spectral path, set
\[
 F_t=\sqrt n\,UD_tV^T,\qquad
 Y_t=aLF_t+\sigma G_k,\qquad
 \mathbb P_t^L=\operatorname{Law}(Y_t),
\]
where $G_k\sim N(0,I_k)$ is independent of $U,V$.  Let $p_t$ be the density of $Y_t$ and let
$I_t=\int \dot p_t(y)^2/p_t(y)\,dy$ be its Fisher information along the path.  After the
moments through order $2K$ cancel, the score calculation constructs an order-$m$ tensor
$Q_m(t)\in\mathcal H^{\otimes m}$, with $m=2K+1$, and gives
\begin{equation}
 I_t\le\left(\frac a\sigma\right)^{2m}m!\,
       \E\norm{L^{\otimes m}Q_m(t)}^2.                   \tag{4.22a}\label{eq:4.22a}
\end{equation}
The recursion defining $Q_m(t)$ and the derivation of \eqref{eq:4.22a} are given in
Appendix~\ref{app:score}.  Thus the measurement map appears in the Fisher information through
the Frobenius norm of the measured tensor $L^{\otimes m}Q_m(t)$.

To apply Lemma~\ref{lem:twirl}, put $P=L^*L$.  Since $L$ has orthonormal rows, $P$ is the
rank-$k$ orthogonal projection onto the row space of $L$, and, pointwise in $Q_m(t)$,
\begin{equation}
 \norm{L^{\otimes m}Q_m(t)}^2
 =\left\langle Q_m(t),P^{\otimes m}Q_m(t)\right\rangle
 =\norm{P^{\otimes m}Q_m(t)}^2.                          \tag{4.22b}\label{eq:4.22b}
\end{equation}
For each fixed $t$, the tensor has the rotation form
$Q_m(t)=(U\otimes V)^{\otimes m}q_m(t)$ for a deterministic tensor $q_m(t)$.
Consequently, Lemma~\ref{lem:twirl} and \eqref{eq:4.22b} give
\begin{align}
 \E\norm{L^{\otimes m}Q_m(t)}^2
 &\le (C_{\mathrm{ten}}m)^{C_{\mathrm{ten}}m}
      \left(\frac{k}{n^2}\right)^{(m+1)/2}
      \E\norm{Q_m(t)}^2.                                \tag{4.22c}\label{eq:4.22c}
\end{align}
The energy estimate in Appendix~\ref{app:score} supplies the remaining link to the spectral path:
\[
 \E\norm{Q_m(t)}^2
 \le (C_{\mathrm{en}}B)^{m-1}n^2
      \left(\frac1n\tr\dot D_t^2\right).
\]
Taking square roots in these bounds therefore controls $\sqrt{I_t}$ by the ordinary path speed,
multiplied by $(k/n^2)^{(m+1)/4}$.  The Fisher--Rao formula then integrates this bound:
\[
 \TV(\mathbb P_0^L,\mathbb P_1^L)
 \le\frac12\int_0^1\sqrt{I_t}\,dt.
\]
Since $m=2K+1$, the measurement factor after taking the square root is
$(k/n^2)^{(K+1)/2}$.  The next lemma records the resulting endpoint estimate with all constants;
the rotation derivatives, Hermite tensors, and spectral-gap estimate underlying it are developed
in Appendix~\ref{app:score}.

\begin{lemma}[Observation comparison along a moment-preserving path]\label{lem:path-tv}
Let $D_t$ be a path supplied by Lemma~\ref{lem:path}, let
$F_t=\sqrt n\,UD_tV^T$, and let $L:\R^{n^2}\to\R^k$ be a row coisometry.  For
$G_k\sim N(0,I_k)$ independent of $U,V$, define
\[
 Y_t=aLF_t+\sigma G_k,
 \qquad \mathbb P_t^L=\operatorname{Law}(Y_t).
\]
Set $m=2K+1$ and suppose $(4m)^{4m}\le n/2$.  Then
\begin{equation}
 \TV(\mathbb P_0^L,\mathbb P_1^L)
 \le \frac12\mathcal A_m n
 \left(\frac{k}{n^2}\right)^{(K+1)/2}\mathcal L_K,       \tag{4.23}\label{eq:4.23}
\end{equation}
where $\mathcal L_K$ is the path length in \eqref{eq:4.10} and
\[
 \mathcal A_m=
 \left(\frac a\sigma\right)^m\sqrt{m!}\,
 (C_{\mathrm{ten}}m)^{C_{\mathrm{ten}}m/2}
 (C_{\mathrm{en}}B)^{(m-1)/2}.
\]
Here $C_{\mathrm{ten}}$ is the constant in Lemma~\ref{lem:twirl}, and
$C_{\mathrm{en}}$ is the fixed absolute constant from the energy estimate in
Appendix~\ref{app:score}.  At the endpoints,
$\mathbb P_b^L=\operatorname{Law}(LX_b)=L_\#\Pi_b$.
\end{lemma}

\begin{proof}
Gaussian smoothing gives $Y_t$ a positive density $p_t$.  The Fisher--Rao inequality proved in
Appendix~\ref{app:fisher-tv} bounds endpoint total variation by one half of the integral of the
statistical speed of $p_t$.  The first $K$ power moments of $D_t^2$ are constant along the path,
so all entry moments of $F_t$ through degree $2K$ are constant as well.  Rotation integration by
parts therefore cancels the first $2K$ terms in the density derivative.  The first remaining term
has the odd tensor order $m=2K+1$.

Appendix~\ref{app:score} carries out this density differentiation exactly and bounds the resulting
unmeasured tensor energy by
$(C_{\mathrm{en}}B)^{m-1}n^2(n^{-1}\tr\dot D_t^2)$.  Lemma~\ref{lem:twirl} shows that the
rank-$k$ measurement retains at most the squared-energy fraction
$(C_{\mathrm{ten}}m)^{C_{\mathrm{ten}}m}
(k/n^2)^{K+1}$.  Taking square roots, integrating in $t$, and using the definition of
$\mathcal L_K$ gives \eqref{eq:4.23}.  These calculations hold on every continuously differentiable
segment of the path; the finitely many joining times have measure zero in the integral.
Finally, $LZ\sim N(0,I_k)$ because $LL^*=I_k$, and $LZ$ is independent of $(U,V)$ because
$Z$ is.  Thus the endpoint law is $\operatorname{Law}(LX_b)$.
\end{proof}

\subsection{Step 5: optimize the lower-bound parameters}

Squaring \eqref{eq:4.23} first gives the exact bound
\[
 \TV(\mathbb P_0^L,\mathbb P_1^L)^2
 \le \frac14\mathcal A_m^2n^2
       \left(\frac{k}{n^2}\right)^{K+1}\mathcal L_K^2.  \tag{4.23a}\label{eq:4.23a}
\]
For the choices in \eqref{eq:4.17}, $\eta_{\mathrm{path}}(n)=1/\log(en)$ once $n$ is large.  Squaring
\eqref{eq:4.10} and using $(u+v)^2\le2u^2+2v^2$ gives
\[
 \mathcal L_K^2
 \le2C_{\mathrm{path}}^2
 \{nB+B^2K^4\log^2(K+1)\log(en)\}.
\]
Here $B$ and $K$ are bounded by fixed powers of $\log n$, whereas the first term contains $n$.
Consequently, for all sufficiently large $n$,
\begin{equation}
 \mathcal L_K^2\le4C_{\mathrm{path}}^2nB.                \tag{4.23b}\label{eq:4.23b}
\end{equation}

Insert the definition of $\mathcal A_m$ and \eqref{eq:4.23b} into \eqref{eq:4.23a}.  Since $m=2K+1$ and
$\log B\le\log2+M_\eps\log(K+2)$, the logarithms of
$(a/\sigma)^{2m}$, $m!$, $(C_{\mathrm{ten}}m)^{C_{\mathrm{ten}}m}$,
$(C_{\mathrm{en}}B)^{m-1}$, and the additional factor $B$ sum to at most
$C_{\mathrm{low},1}(\eps)K\log(K+1)$ for one fixed constant
$C_{\mathrm{low},1}(\eps)$.  Fix this constant once at this point.  Finally,
\[
 n^3\left(\frac{k}{n^2}\right)^{K+1}
 =n\frac{k^{K+1}}{n^{2K}}.
\]
These substitutions give
\begin{equation}
 \TV(\mathbb P_0^L,\mathbb P_1^L)^2
 \le n\exp\{C_{\mathrm{low},1}(\eps)K\log(K+1)\}
       \frac{k^{K+1}}{n^{2K}}.                              \tag{4.24}\label{eq:4.24}
\end{equation}
The same inequality holds for every coisometry of rank $r\le k$: apply the displayed calculation
with $r$ and then use $r^{K+1}\le k^{K+1}$.

Define
\begin{equation}
 A_\eps\defeq1+
 \frac{3+C_{\mathrm{low},1}(\eps)c_{\mathrm{low}}}
      {c_{\mathrm{low}}},                                 \tag{4.25}\label{eq:4.25}
\end{equation}
and set $k_0=\lfloor n^2/(\log n)^{A_\eps}\rfloor$.  For every integer $k\le k_0$, the
logarithm of the right side of \eqref{eq:4.24} is at most
\begin{equation}
 3\log n+C_{\mathrm{low},1}(\eps)K\log(K+1)
 -A_\eps(K+1)\log\log n.                                 \tag{4.26}\label{eq:4.26}
\end{equation}
With $K=\lfloor c_{\mathrm{low}}\log n/\log\log n\rfloor$, the extra $1$ in \eqref{eq:4.25}
leaves a negative term of order $c_{\mathrm{low}}\log n$.  Hence \eqref{eq:4.26} tends to
$-\infty$, so the two
observed laws have total variation below $1/4$ for all large $n$.  Take $\zeta=1/100$ in
Lemma~\ref{lem:smoothing}; then $1/4+2/100<1/3$.  Apply Lemma~\ref{lem:yao} with
$\Phi(A)=\norm A_{\Sone}$ to obtain
\[
 k_\eps(n)\ge k_0+1>\frac{n^2}{(\log n)^{A_\eps}}.
\]

\begin{remark}[Limitation of the lower-bound argument]
At $k=\gamma n^2$ for a fixed $0<\gamma<1$, the power in \eqref{eq:4.24} contributes approximately
$n^3\gamma^{K+1}$, which requires
$K$ of order $\log n$.  The available tensor comparison also pays
$\exp\{C_{\mathrm{low},1}(\eps)K\log(K+1)\}$ and a worst-level spectral factor.  Those losses
dominate at that choice.
This is a limitation of the proved comparison, not evidence for a quadratic lower bound.
\end{remark}

\section{Upper bound}\label{sec:upper}

It suffices to prove the upper bound for $0<\eps<1/2$.  If $1/2\le\eps<1$, run the
construction with accuracy $1/3$; its $(1\pm1/3)$ guarantee implies the weaker
$(1\pm\eps)$ guarantee, and its measurement constant can be used as $C_\eps$.

\paragraph{Algorithmic idea.}
The algorithm divides the singular values into a short \emph{head}, containing the largest
singular directions, and a \emph{tail}, containing everything else.  The cutoff between them is
not known when the sketch is formed.  We therefore prepare for a short list of possible cutoffs
and let the decoder choose among them.  A useful cutoff has one of two properties: either the tail
already contributes very little to the nuclear norm, or its Frobenius energy is well spread, meaning
that no single remaining singular value carries too much of that energy.  In the first case the
tail may be discarded.  In the second case its nuclear norm can be recovered from a moderate
number of estimated spectral moments.

The sketch stores four independent Gaussian matrix products, all sampled before $A$ is seen.
The row sketch $GA$ lets the decoder identify candidate head subspaces.  Together with the
singular-value information in $GA$, the second row sketch $WA$ estimates the Frobenius norm of
each resulting residual and certifies whether that residual is sufficiently well spread.  After a
cutoff is selected, the product $AR_0$ estimates the nuclear
norm of the head by a small regression problem, while $YA$ supplies independent Gaussian samples
from which the decoder estimates the nuclear norm of the tail.  Thus the data-dependent choice of
cutoff occurs only during decoding; every stored number remains a fixed linear measurement of
$A$.

\paragraph{Road map.}
Step~1 proves that the list of candidate cutoffs always contains either a negligible tail or a
well-spread tail.  Step~2 establishes one covariance event for $G$ that holds simultaneously at
all candidate cutoffs, and Step~3 uses it to show that the head chosen from $GA$ leaves almost the
same tail mass as the corresponding exact singular-value cutoff.  Step~4 turns the information in
$GA$ and $WA$ into a test that selects a well-spread residual whenever one is needed; if no
candidate passes, the residual at the last cutoff is negligible.  Step~5 estimates the selected
head from $AR_0$.  Step~6 estimates a certified tail from $YA$ by approximating the square-root
function with a polynomial and estimating the required spectral moments.  Finally, Step~7 states
the complete sketch and decoder and balances their parameters to obtain the measurement bound in
Theorem~\ref{thm:main}.

\subsection{Step 1: some cutoff has a small tail or a well-spread tail}

\begin{lemma}[Multiscale head--tail dichotomy]\label{lem:dichotomy}
Fix $1\le R\le n$ and $0<\eta<1$.  Let
\[
 h_j=\min\{jR,n\},\qquad
 0\le j\le J=\left\lceil\log_2\left(\frac{\sqrt{n/R}}\eta\right)\right\rceil.
\]
Here $T_h(A)=\sum_{i>h}\sigma_i(A)$ is the tail nuclear norm defined in \eqref{eq:3.2}.
For at least one inspected cutoff $h_j$, either
\begin{equation}
 T_{h_j}(A)\le\eta\norm A_{\Sone},                         \tag{5.1}\label{eq:5.1}
\end{equation}
or the exact residual $A-A_{h_j}$ has stable rank at least $R/4$.
\end{lemma}

\begin{proof}
If any inspected tail satisfies $T_{h_j}(A)=0$, then \eqref{eq:5.1} holds immediately.  Hence
assume that every inspected tail is nonzero.  If some inspected residual $A-A_{h_j}$ has stable
rank at least $R/4$, the second alternative holds, so assume that all inspected residuals have
stable rank below $R/4$.

Since $T_n(A)=0$, no inspected cutoff can now equal $n$.  Thus $h_j=jR<n$ for every
$0\le j\le J$.  In particular, the quantities
\begin{equation*}
 b_j:=\sigma_{jR+1}(A)
\end{equation*}
are defined and positive.  The stable-rank assumption gives, for every $0\le j\le J$,
\[
 \sum_{i>jR}\sigma_i(A)^2<\frac R4b_j^2.
\]
For $0\le j<J$, the block from $jR+1$ through $(j+1)R$ contains $R$ singular values, each at
least $b_{j+1}$.  Consequently
\begin{equation*}
 Rb_{j+1}^2
 \le\sum_{i=jR+1}^{(j+1)R}\sigma_i(A)^2
 <\frac R4b_j^2,
\end{equation*}
and hence $b_{j+1}<b_j/2$.  Iterating gives
$b_j\le2^{-(j-1)}b_1$ for $1\le j\le J$.

Since $R\le n$ and $0<\eta<1$, the definition of $J$ gives $J\ge1$.  Cauchy--Schwarz and the
stable-rank inequality give, for $1\le j\le J$,
\[
 T_{jR}(A)
 \le\sqrt n\left(\sum_{i>jR}\sigma_i(A)^2\right)^{1/2}
 <\frac{\sqrt{nR}}2b_j.
\]
Moreover, the first $R$ singular values are at least $b_1=\sigma_{R+1}(A)$, so
$\norm A_{\Sone}\ge Rb_1$.  Therefore
\begin{equation*}
 \frac{T_{jR}(A)}{\norm A_{\Sone}}
 <\frac12\sqrt{\frac nR}\frac{b_j}{b_1}
 \le\sqrt{\frac nR}\,2^{-j}.
\end{equation*}
At $j=J$, the definition of $J$ gives
$\sqrt{n/R}\,2^{-J}\le\eta$, proving \eqref{eq:5.1}.
\end{proof}

Later we set $R=\lceil n/b\rceil$.  In the stable-rank alternative, a residual $C$ has normalized
squared singular values
\begin{equation}
 \lambda_i=\frac{n\sigma_i(C)^2}{\norm C_F^2},\qquad
 \frac1n\sum_i\lambda_i=1,\qquad \max_i\lambda_i\le4b. \tag{5.2}\label{eq:5.2}
\end{equation}
This bounded support is what makes high-order moment estimation possible.

\subsection{Step 2: one Gaussian row sketch approximates every candidate head space}

\begin{theorem}[Effective-rank Gaussian covariance bound]\label{ext:covariance}
Let $G\in\R^{s\times n}$ have independent $N(0,1/s)$ entries.  For every deterministic
$D\in\R^{n\times r}$ and
$0<\delta<1/2$, there are fixed absolute constants
$C_{\mathrm{cov},1},C_{\mathrm{cov},2}>0$ such that, with probability at least $1-\delta$,
\begin{align}
 \norm{D^T(G^TG-I)D}_{\op}
 &\le C_{\mathrm{cov},1}\left(\frac{\norm D_{\op}\norm D_F}{\sqrt s}
             +\frac{\norm D_F^2}{s}\right) \notag\\
 &\quad+C_{\mathrm{cov},1}\norm D_{\op}^2
      \left(\sqrt{\frac{\log(1/\delta)}s}+\frac{\log(1/\delta)}s\right). \tag{5.3}\label{eq:5.3}
\end{align}
In particular, if $\norm D_{\op}\le1$ and
$s\ge C_{\mathrm{cov},2}\alpha^{-2}(\norm D_F^2+\log(1/\delta))$, the right side is at most $\alpha$.
This is the finite-dimensional Gaussian specialization of Theorems~4 and~5 of Koltchinskii and
Lounici~\cite{KoltchinskiiLounici}.  Indeed, the sampled vectors are $D^Tg$, where
$g\sim N(0,I_n)$, and therefore have covariance $\Sigma=D^TD$.  For $D\ne0$, their effective-rank
parameter satisfies
\[
 \mathbf r(\Sigma)=\frac{(\E\norm{D^Tg})^2}{\norm\Sigma_{\op}}
 \le\frac{\E\norm{D^Tg}^2}{\norm\Sigma_{\op}}
 =\frac{\tr(\Sigma)}{\norm\Sigma_{\op}}
 =\frac{\norm D_F^2}{\norm D_{\op}^2}.
\]
Combining their expectation and concentration bounds with this inequality yields
\eqref{eq:5.3}, after enlarging the absolute constant.  The case $D=0$ is immediate.
\end{theorem}

Fix $0<\alpha<1/10$ and set $c_0=1$.  For $h\ge1$, put
\begin{equation}
 \beta_h=\frac{c_0\alpha^4}{h},\qquad
 \delta_h=\frac{\beta_h}{\alpha}F_h(A)^2.                  \tag{5.4}\label{eq:5.4}
\end{equation}
The desired simultaneous event is
\begin{equation}
 (1-\alpha)A^TA-\beta_hF_h^2I
 \preceq A^TG^TGA
 \preceq(1+\alpha)A^TA+\beta_hF_h^2I.                     \tag{5.5}\label{eq:5.5}
\end{equation}
Although $F_h$ appears in the analysis, it does not affect the distribution of $G$.

To prove \eqref{eq:5.5}, when $F_h>0$ set
\[
 D_h=A(A^TA+\delta_hI)^{-1/2}.
\]
Then $\norm{D_h}_{\op}\le1$ and
\begin{equation}
 \norm{D_h}_F^2=\sum_i\frac{\sigma_i^2}{\sigma_i^2+\delta_h}
 \le h+\frac{F_h^2}{\delta_h}=h+\frac\alpha{\beta_h}
 =C_{\mathrm{dim}}(\alpha)h,                              \tag{5.6}\label{eq:5.6}
\end{equation}
where $C_{\mathrm{dim}}(\alpha)\defeq1+(c_0\alpha^3)^{-1}$.
Apply Theorem~\ref{ext:covariance} and conjugate the resulting inequality by
$(A^TA+\delta_hI)^{1/2}$.  Since $\alpha\delta_h=\beta_hF_h^2$, this is exactly \eqref{eq:5.5}.
If $F_h=0$, restrict to the rank-$h$ support and use the same bound there.

One $G$ works for every predetermined $h_j$ by a union bound.  Fix
$\delta_{\mathrm{up}}=1/100$.  With
\begin{equation}
 H=\min\{JR,n\},\qquad
 s_G=\left\lceil C_{\mathrm{rows}}(\alpha)
          \left\{H+\log\frac{J+1}{\delta_{\mathrm{up}}}\right\}\right\rceil. \tag{5.7}\label{eq:5.7}
\end{equation}
all events hold simultaneously.  At $h=0$, take
$\beta_0=c_0\alpha^4/R$ and $\delta_0=(\beta_0/\alpha)F_0^2$; the same calculation has effective
dimension at most $C_{\mathrm{dim}}(\alpha)R$.  Here
$C_{\mathrm{rows}}(\alpha)$ is one fixed multiple of
$C_{\mathrm{cov},2}\alpha^{-2}C_{\mathrm{dim}}(\alpha)$.

\subsection{Step 3: the head found from the row sketch leaves nearly the correct tail}

For $h\ge1$, let $Q_h\in\R^{n\times h}$ contain the top $h$ right singular vectors of $GA$,
let $P_h=Q_hQ_h^T$, and set $C_h=A(I-P_h)$.  Let $\norm B_{(r)}=\sum_{i=1}^r\sigma_i(B)$ be
the Ky--Fan $r$-norm.  Throughout this step, every singular-value sequence is extended by zeros
beyond index $n$; in particular, $\norm B_{(r)}=\norm B_{\Sone}$ when $r\ge n$.  The min--max
principle applied to \eqref{eq:5.5} gives, for every projector $P$,
\begin{align}
 \sigma_i(GA(I-P))&\le\sqrt{1+\alpha}\,\sigma_i(A(I-P))+\sqrt{\beta_h}F_h, \tag{5.8}\label{eq:5.8}\\
 \sigma_i(A(I-P))&\le\frac{\sigma_i(GA(I-P))+\sqrt{\beta_h}F_h}{\sqrt{1-\alpha}}. \tag{5.9}\label{eq:5.9}
\end{align}
Let $P_*$ be the true top-$h$ right projector and set $r=\lceil h/\alpha\rceil$.  Apply \eqref{eq:5.9}
to $P_h$, use the optimality of $P_h$ for every unitarily invariant norm of the residual of
$GA$, and then apply \eqref{eq:5.8} to $P_*$.  Since
$\sqrt{(1+\alpha)/(1-\alpha)}\le1+2\alpha$ and
$2(1-\alpha)^{-1/2}<3$ for $0<\alpha<1/10$, this gives
\begin{align*}
 \norm{C_h}_{(r)}
 &\le\frac{\norm{GA(I-P_h)}_{(r)}+r\sqrt{\beta_h}F_h}{\sqrt{1-\alpha}}\\
 &\le\frac{\norm{GA(I-P_*)}_{(r)}+r\sqrt{\beta_h}F_h}{\sqrt{1-\alpha}}\\
 &\le\sqrt{\frac{1+\alpha}{1-\alpha}}
       \sum_{i=h+1}^{h+r}\sigma_i(A)
       +2(1-\alpha)^{-1/2}r\sqrt{\beta_h}F_h.
\end{align*}
The middle inequality is the defining optimality of the truncated singular-value decomposition
of $GA$.  Therefore
\begin{equation}
 \norm{C_h}_{(r)}\le(1+2\alpha)
                  \sum_{i=h+1}^{h+r}\sigma_i(A)
                  +3r\sqrt{\beta_h}F_h.                 \tag{5.10}\label{eq:5.10}
\end{equation}
Right multiplication by a projection cannot increase singular values, so
\begin{equation}
 \sum_{i>r}\sigma_i(C_h)\le\sum_{i>r}\sigma_i(A).          \tag{5.11}\label{eq:5.11}
\end{equation}
The two right sides overlap on only $h$ entries, and monotonicity gives
\begin{equation}
 \sum_{i=r+1}^{r+h}\sigma_i(A)
 \le\frac hr\sum_{i=h+1}^{h+r}\sigma_i(A)\le\alpha T_h(A). \tag{5.12}\label{eq:5.12}
\end{equation}
Therefore, using $\sum_{i=h+1}^{h+r}\sigma_i(A)\le T_h(A)$,
\begin{equation}
 \norm{C_h}_{\Sone}\le T_h(A)+3\alpha\norm A_{\Sone}
       +3r\sqrt{\beta_h}F_h.                            \tag{5.13}\label{eq:5.13}
\end{equation}
If $a_h>0$, every tail singular value $x\le a_h$ satisfies $x\ge x^2/a_h$, so
\begin{equation}
 \norm A_{\Sone}\ge ha_h+\frac{F_h^2}{a_h}\ge2\sqrt hF_h. \tag{5.14}\label{eq:5.14}
\end{equation}
If $a_h=0$, then $F_h=0$.  Since $r\le2h/\alpha$ and $\beta_h=c_0\alpha^4/h$, \eqref{eq:5.14} gives
\[
 3r\sqrt{\beta_h}F_h
 \le6\sqrt{c_0}\,\alpha\sqrt hF_h
 \le3\sqrt{c_0}\,\alpha\norm A_{\Sone}.
\]
Thus
\begin{equation}
 \boxed{\norm{A(I-P_h)}_{\Sone}
 \le T_h(A)+6\alpha\norm A_{\Sone}.}                    \tag{5.15}\label{eq:5.15}
\end{equation}
The total-relative form of this error bound is essential: a tail-relative error would require
many more rows.

We also need approximate additivity.  Since $A=AP_h+C_h$, the triangle inequality gives the
lower bound below.  For the upper bound, $AP_h$ has rank at most $h$ and right multiplication by
an orthogonal projection cannot increase any singular value; hence
\begin{equation}
 \norm{AP_h}_{\Sone}\le\sum_{i\le h}\sigma_i(A).         \tag{5.16}\label{eq:5.16}
\end{equation}
Combining \eqref{eq:5.15} and \eqref{eq:5.16} gives
\begin{equation}
 0\le\norm{AP_h}_{\Sone}+\norm{C_h}_{\Sone}-\norm A_{\Sone}
 \le6\alpha\norm A_{\Sone}.                             \tag{5.17}\label{eq:5.17}
\end{equation}
No stable-rank assertion about $C_h$ is needed here: the data-dependent certificate in the next
step proves that property directly.
For $h=0$, set $P_0=0$ and $C_0=A$.  Then \eqref{eq:5.15} and \eqref{eq:5.17} hold
trivially: the former reads $\norm A_{\Sone}\le\norm A_{\Sone}+6\alpha\norm A_{\Sone}$,
and the latter has both sides equal to zero.

\subsection{Step 4: certify a usable cutoff from stored data}

Independently draw $W\in\R^{d\times n}$ with $N(0,1/d)$ entries and store $WA$.  Conditional on
$G$, all candidate residuals are fixed.  We use the following elementary estimate.  If $B$ is
deterministic and $0<u<1$, then
\begin{equation}
 \Prb\!\left\{\left|\norm{WB}_F^2-\norm B_F^2\right|>u\norm B_F^2\right\}
 \le2e^{-du^2/16}.                                      \tag{5.18a}\label{eq:5.18a}
\end{equation}
To prove it, diagonalize $BB^T$ with eigenvalues $\lambda_i$.  For a standard Gaussian vector
$g$ and $|\theta|\le(4\max_i\lambda_i)^{-1}$,
\[
 \E e^{\theta(g^TBB^Tg-\tr BB^T)}
 =\prod_i e^{-\theta\lambda_i}(1-2\theta\lambda_i)^{-1/2}
 \le e^{2\theta^2\sum_i\lambda_i^2}.
\]
Apply this inequality to the $d$ independent rows, use
$\sum_i\lambda_i^2\le(\sum_i\lambda_i)^2$, and take
$\theta=u/\{4\sum_i\lambda_i\}$.  This value obeys the displayed restriction because
$\max_i\lambda_i\le\sum_i\lambda_i$.  Chernoff's argument gives at most
$e^{-du^2/8}$ for each tail, and hence the weaker common bound \eqref{eq:5.18a}, with denominator $16$.
If $B=0$, the assertion is deterministic.  This proves \eqref{eq:5.18a}.  Choose
\begin{equation}
 d=\left\lceil16\alpha^{-2}\log\{2(J+1)/\delta_{\mathrm{up}}\}\right\rceil. \tag{5.18b}\label{eq:5.18b}
\end{equation}
Then \eqref{eq:5.18a} and a union bound give, simultaneously for all $J+1$ candidates,
\begin{equation}
 \widehat F_j^2\defeq\norm{WA(I-P_j)}_F^2
 =(1\pm\alpha)\norm{A(I-P_j)}_F^2.                        \tag{5.19}\label{eq:5.19}
\end{equation}
For candidate index $j$, write $\beta_j=\beta_{h_j}$ when $h_j>0$, and use the separately
defined value $\beta_0=c_0\alpha^4/R$ when $h_j=0$.  Let
$\gamma_j=\sigma_{h_j+1}(GA)$, with $\gamma_j=0$ if $h_j=n$.  Declare $j$ certified if
\begin{equation}
 \frac{\gamma_j^2+2\beta_j\widehat F_j^2}{1-\alpha}
 \le\frac{D_{\mathrm{cert}}\widehat F_j^2}{R},             \tag{5.20}\label{eq:5.20}
\end{equation}
where we fix the numerical constant $D_{\mathrm{cert}}=12$.

Equations \eqref{eq:5.5}, \eqref{eq:5.19}, and \eqref{eq:5.20} imply that every nonzero
certified $C_j$ has stable rank at least $R/\{D_{\mathrm{cert}}(1+\alpha)\}$.  If $C_j=0$,
its tail nuclear norm is exactly zero and the decoder bypasses tail estimation.
Appendix~\ref{app:decoder} verifies both cases and the exact constant.
Conversely, if the exact tail at $h_j\ge R$ has stable rank at least $R/4$, then
$a_{h_j}^2\le4F_{h_j}^2/R$, and the same equations make \eqref{eq:5.20} pass.  The $h=0$ case uses
$P_0=0$, $C_0=A$, and $\gamma_0=\norm{GA}_{\op}$ and follows identically.

If no cutoff is certified, the stable-rank alternative in Lemma~\ref{lem:dichotomy} cannot
occur, because every such cutoff would pass the certificate.  The lemma therefore supplies an
inspected $h_j$ with $T_{h_j}(A)\le\eta\norm A_{\Sone}$.  Since $H=h_J$ is the largest
inspected cutoff and $T_h(A)$ is nonincreasing in $h$, $T_H(A)\le\eta\norm A_{\Sone}$.
Equation \eqref{eq:5.15} then makes the actual last residual
$(\eta+6\alpha)$-negligible.  The decoder can
skip tail estimation in this branch.

\subsection{Step 5: estimate the implicit head without reconstructing it}

The projector $P=QQ^T$ chosen above depends on $GA$, but a fresh Gaussian block can estimate
$AQ$.  Draw $R_0\in\R^{n\times t}$ with $N(0,1/t)$ entries, store $AR_0$, and form
\begin{equation}
 X=(AR_0)(Q^TR_0)^\dagger,                                 \tag{5.21}\label{eq:5.21}
\end{equation}
where ${}^\dagger$ is the Moore--Penrose inverse.  With $C=A(I-QQ^T)$,
\begin{equation}
 X-AQ=CR_0(Q^TR_0)^\dagger.                                \tag{5.22}\label{eq:5.22}
\end{equation}
Each right singular direction $v$ of $C$ is orthogonal to $Q$, so the Gaussian vector
$v^TR_0$ is independent of $Q^TR_0$.  Conditional on $Q$, and provided $t>h+1$, the
inverse-Wishart identity gives
\begin{equation}
 \E[\norm{X-AQ}_{\Sone}\mid Q]
 \le\norm C_{\Sone}\sqrt{\frac{h}{t-h-1}}.                \tag{5.23}\label{eq:5.23}
\end{equation}
For completeness, if $Z=Q^TR_0$ then $tZZ^T$ is Wishart with $t$ samples and identity
covariance.  Rotational invariance makes its inverse expectation a scalar multiple of the identity;
integration by parts in the Wishart density gives
\[
 \E\norm{Z^\dagger}_F^2=\E\tr(ZZ^T)^{-1}=\frac{th}{t-h-1}.
\]
The inverse-Wishart expectation used in the last display follows from
Muirhead's Theorem~3.2.12; see~\cite[pp.~96--97]{Muirhead}.
Expanding $C$ in singular directions, applying the nuclear-norm triangle inequality, and then
Jensen gives \eqref{eq:5.23}.  Fix
\[
 C_{\mathrm{reg}}(\alpha)\defeq 2+10^4\alpha^{-2},
 \qquad t=\lceil C_{\mathrm{reg}}(\alpha)H\rceil.
\]
Since $h\le H$,
$t-h-1\ge10^4\alpha^{-2}h$ and hence the square root in \eqref{eq:5.23} is at most
$\alpha/100$.  Also $\norm C_{\Sone}\le\norm A_{\Sone}$.  Markov's inequality therefore gives,
with probability at least $0.99$,
\begin{equation}
 \abs{\norm X_{\Sone}-\norm{AQ}_{\Sone}}
 \le\alpha\norm A_{\Sone}.                                \tag{5.24}\label{eq:5.24}
\end{equation}
The matrix $Q^TR_0$ has full row rank almost surely.  Moreover, $AQ$ and
$AP=AQQ^T$ have the same nonzero singular values because their products with their transposes
are both $AQQ^TA^T$; hence \eqref{eq:5.24} estimates the head norm used in \eqref{eq:5.17}.  Only the selected
$Q$ matters, so no union bound over cutoffs is needed.

\subsection{Step 6: estimate a certified, well-spread tail}

Fix a target tail accuracy $0<\rho<1$.
If a certified residual is zero, then its tail nuclear norm is zero and the decoder bypasses this
step.  Hence let $C\ne0$ be a certified residual.  Define its normalized squared singular values as
in \eqref{eq:5.2}:
\[
 \lambda_i=\frac{n\sigma_i(C)^2}{\norm C_F^2},
 \qquad \frac1n\sum_i\lambda_i=1.
\]
By \eqref{eq:D.3}, they satisfy the known support bound
\[
 0\le\lambda_i\le B_\lambda:=\frac{D_0n}{R},
 \qquad D_0=D_{\mathrm{cert}}(1+\alpha).
\]
Define
\begin{equation}
 q=\frac1n\sum_i\sqrt{\lambda_i}
   =\frac{\norm C_{\Sone}}{\sqrt n\norm C_F}.              \tag{5.25}\label{eq:5.25}
\end{equation}
Thus it suffices to estimate $q$ and $\norm C_F$.

\begin{theorem}[Polynomial approximation to the square root]\label{ext:poly}
There is a fixed absolute constant $C_{\mathrm{poly}}$ such that, for every $K\ge1$, there is a degree-$K$ polynomial
$P_K^*(u)=\sum_{j=0}^Ka_ju^j$ such that
\begin{equation}
 \max_j\abs{a_j}\le2^{3K},\qquad
 \sup_{0\le u\le1}\abs{P_K^*(u)-\sqrt u}
 \le\frac{C_{\mathrm{poly}}}{K}.                         \tag{5.26}\label{eq:5.26}
\end{equation}
The stated bounded-coefficient construction is Lemma~4.5 of Th\'epaut and
Verzelen~\cite{ThepautVerzelen}; their explicit approximation error is
$2/\{\pi(2K+1)\}$, which is bounded by \eqref{eq:5.26} for a fixed $C_{\mathrm{poly}}$.
\end{theorem}

After scaling $u=\lambda/B_\lambda$, the approximation bias is at most a $\rho$ fraction of $q$ when
$K\ge4C_{\mathrm{poly}}B_\lambda/\rho$.  If
\begin{equation}
 \mu_j=\frac1n\sum_i\lambda_i^j,                           \tag{5.27}\label{eq:5.27}
\end{equation}
then it is enough to estimate every $\mu_j$ with relative error
\begin{equation}
 \tau=\frac{c_{\mathrm{mom}}\rho}{K2^{3K}},                \tag{5.28}\label{eq:5.28}
\end{equation}
where we fix $c_{\mathrm{mom}}=1/4$.

Draw an independent $Y\in\R^{s\times n}$ with $N(0,1/s)$ entries and store $YA$.  Once $Q$ is
selected, the decoder forms $YC=YA(I-QQ^T)$.  Write
$z_a=\sqrt{s}(YC)_{a,*}$; its rows are independent $N(0,\Sigma)$ vectors with
$\Sigma=C^TC$.  For $j\ge2$, define the injective cycle statistic
\begin{equation}
 \widehat p_j=\frac1{(s)_j}
 \sum_{a_1,\ldots,a_j\ \mathrm{distinct}}
 \ip{z_{a_1}}{z_{a_2}}\ip{z_{a_2}}{z_{a_3}}\cdots
 \ip{z_{a_j}}{z_{a_1}},                                   \tag{5.29}\label{eq:5.29}
\end{equation}
where $(s)_j=s(s-1)\cdots(s-j+1)$, and put
$\widehat p_1=s^{-1}\sum_a\norm{z_a}^2$.  Successively integrating one Gaussian row at a time
gives
\begin{equation}
 \E\widehat p_j=\tr(\Sigma^j)=\sum_i\sigma_i(C)^{2j}.      \tag{5.30}\label{eq:5.30}
\end{equation}

\begin{theorem}[Gaussian injective-cycle variance bound]\label{ext:moment}
Let $2\le j\le s$ and $0<\zeta<1$.  For the statistic in \eqref{eq:5.29}, Chebyshev's
inequality at failure probability $\zeta$ gives relative error at most
\begin{equation}
 \frac{f(j)}{\sqrt\zeta}
 \max\left\{
   \frac{n^{j/2-1}}{s^{j/2}},
   \frac{n^{1/4-1/(2j)}}{\sqrt s},
   \frac1{\sqrt s}
 \right\},\qquad
 f(j)=2^{6j}j^{3j}3^{j/2}.                                \tag{5.31}\label{eq:5.31}
\end{equation}
This is Theorem~1 and Proposition~4 of Kong and Valiant~\cite{KongValiant}, with ambient
dimension $d$ there replaced by $n$, sample size $n$ there replaced by $s$, and Gaussian fourth
moment $\beta=3$.  Their increasing-cycle statistic has the displayed bound.  The ordered
statistic \eqref{eq:5.29} is its symmetrization over sample labels, which cannot increase variance; this
normalization is verified in Appendix~\ref{app:decoder}.
\end{theorem}

Estimate
\begin{equation}
 \widehat\mu_j=\frac{n^{j-1}\widehat p_j}{\widehat p_1^j},
 \qquad
 \widehat q=\sqrt{B_\lambda}\sum_{j=0}^Ka_j\frac{\widehat\mu_j}{B_\lambda^j}. \tag{5.32}\label{eq:5.32}
\end{equation}
Here $\widehat\mu_0=1$ by definition.  The actual nonnegative tail estimate returned by the
decoder is
\begin{equation}
 \widehat T(C)=\max\{0,\sqrt n\,\sqrt{\widehat p_1}\,\widehat q\}. \tag{5.32a}\label{eq:5.32a}
\end{equation}
On the simultaneous relative-moment event, \eqref{eq:5.26}--\eqref{eq:5.28} bound stochastic error by $\rho q/2$
and approximation bias by $\rho q/2$.  Because
$\norm C_{\Sone}=\sqrt n\,\norm C_F\,q$ and $\widehat p_1$ estimates $\norm C_F^2$, this gives a
relative-error-at-most-$\rho$ tail estimate.  Appendix~\ref{app:decoder} proves this ratio calculation and
verifies the zero-residual convention used in Algorithm~\ref{alg:upper-decode}.

Take $\zeta=(100K)^{-1}$.  Solving the first constraint in \eqref{eq:5.31}, uniformly for the precision
\eqref{eq:5.28}, gives fixed constants
$C_{\mathrm{row},1},C_{\mathrm{row},2}>0$ and a fixed $c_{\mathrm{row}}>0$ for which the
following degree-$j$ row count is sufficient for $2\le j\le K$:
\begin{equation}
 s_j=\left\lceil n^{1-2/j}\exp\!\left\{C_{\mathrm{row},1}\log(j+1)
                         +\frac{C_{\mathrm{row},2}K}{j}\right\}\right\rceil. \tag{5.33}\label{eq:5.33}
\end{equation}
If
\begin{equation}
 K\le c_{\mathrm{row}}\frac{\log n}{\log\log(e n)},       \tag{5.34}\label{eq:5.34}
\end{equation}
the exponent $(-2\log n+C_{\mathrm{row},2}K)/j$ is largest at $j=K$.  Hence
\begin{equation}
 s_T=\max\left\{K,\left\lceil
 C_{\mathrm{tail}}(\rho)K^{C_{\mathrm{tail}}(\rho)}n^{1-2/K}
 \right\rceil\right\}.                                    \tag{5.35}\label{eq:5.35}
\end{equation}
controls the first terms for every $j\le K$.  At this row count, the logarithm of the second term's
ratio to the target precision is at most a fixed multiple of $K\log(K+1)$ minus
$(1/2-1/K)\log n$; the third term has the larger saving $(1-2/K)\log n$.
The degree-one estimate $\widehat p_1$ requires only $\exp\{O_\rho(K)\}$ rows.  Appendix~\ref{app:decoder} derives
these three comparisons explicitly in \eqref{eq:D.10a}--\eqref{eq:D.12}.  A sufficiently small fixed constant in
\eqref{eq:5.34} controls all of them.  The maximum with $K$ ensures that every injective cycle is defined;
it is redundant in the eventual regime $K\ge4$.  The tail block therefore uses
\begin{equation}
 ns_T\le2C_{\mathrm{tail}}(\rho)K^{C_{\mathrm{tail}}(\rho)}n^{2-2/K}. \tag{5.36}\label{eq:5.36}
\end{equation}
scalar linear measurements.
Here $C_{\mathrm{tail}}(\rho)$ is chosen once, after the four row-count estimates, as their
maximum (including the fixed factors depending on the requested precision $\rho$).

\subsection{Step 7: assemble and optimize the fixed sketch}

We collect the construction and the decoder in one place.  All constants and row counts referenced
below were defined in Steps~2--6 and depend only on $n$ and $\eps$.
Fix a deterministic integer $n_0(\eps)$ large enough that, for every $n\ge n_0(\eps)$, the
parameter choice below satisfies all large-$n$ conditions used in the proof, including $b\ge1$
(and hence $R=\lceil n/b\rceil\le n$), condition \eqref{eq:5.34}, and the estimates in the
optimization following the algorithms.  Such an integer exists by the estimates proved there.

\begin{algorithmblock}{Oblivious Gaussian sketch}\label{alg:upper-sketch}
\textbf{Input and output.}
The input is $A\in\R^{n\times n}$ and an accuracy $0<\eps<1/2$.  The output is a collection of
linear measurements of $A$.  If $n<n_0(\eps)$, store all $n^2$ entries of $A$, tag this as the
exact branch, and stop.  The following is the nontrivial branch.

\smallskip
\noindent
\textbf{Parameters.}
Set $\alpha=\rho=\eta=\eps/20$; recall that $\delta_{\mathrm{up}}=1/100$.  For the sufficiently small fixed constant
$c_{\mathrm{bal}}(\eps)$ chosen below, set
\begin{equation*}
 b=c_{\mathrm{bal}}(\eps)\frac{\log(e n)}{\log\log(e^e n)},
\end{equation*}
and then set
\begin{equation*}
 R=\lceil n/b\rceil,\qquad
 J=\left\lceil\log_2\!\left(\frac{\sqrt{n/R}}\eta\right)\right\rceil,
 \qquad H=\min\{JR,n\},
\end{equation*}
\begin{equation*}
 D_0=D_{\mathrm{cert}}(1+\alpha),\qquad
 K=\left\lceil\frac{4C_{\mathrm{poly}}D_0b}{\rho}\right\rceil.
\end{equation*}
Choose $s_G,d,s_T$ by \eqref{eq:5.7}, \eqref{eq:5.18b}, and \eqref{eq:5.35}, respectively,
and set $t=\lceil C_{\mathrm{reg}}(\alpha)H\rceil$.

\smallskip
\noindent\textbf{Sketch construction.}
Before seeing $A$, independently draw
\begin{equation*}
 G\in\R^{s_G\times n},\quad W\in\R^{d\times n},\quad
 R_0\in\R^{n\times t},\quad Y\in\R^{s_T\times n},
\end{equation*}
with independent centered Gaussian entries of variances $1/s_G$, $1/d$, $1/t$, and $1/s_T$,
respectively.  Store
\begin{equation}
 \boxed{GA,\qquad WA,\qquad AR_0,\qquad YA.}               \tag{5.37}\label{eq:5.37}
\end{equation}
Return the four stored products in \eqref{eq:5.37}.  The decoder knows the four sampled matrices
as public randomness forming part of the sketch description.
\end{algorithmblock}

\begin{algorithmblock}{Nuclear-norm decoder}\label{alg:upper-decode}
\textbf{Input and output.}
The input consists of the sketch produced by Algorithm~\ref{alg:upper-sketch}, its public Gaussian
matrices and parameters, and the accuracy $\eps$.  The output is an estimate $\widehat N$ of
$\norm A_{\Sone}$.

\smallskip
\noindent\textbf{Decoding steps.}
\begin{enumerate}
 \item If the sketch is tagged as the exact branch, compute the singular values of the stored
       matrix $A$ and return their sum.  Otherwise, if all four stored products in
       \eqref{eq:5.37} are zero, return $0$.

 \item For $0\le j\le J$, let $h_j=\min\{jR,n\}$.  From a singular value decomposition of
       $GA$, let $Q_j\in\R^{n\times h_j}$ contain its top $h_j$ right singular vectors and put
       $P_j=Q_jQ_j^T$.  Use one fixed measurable singular-value-decomposition convention,
       including deterministic tie breaking and orthonormal completion when $h_j$ exceeds the
       rank of $GA$.  Use the empty matrix and $P_0=0$ when $h_j=0$.  Compute
       \begin{equation*}
        \gamma_j=\begin{cases}
          \sigma_{h_j+1}(GA),&h_j<n,\\
          0,&h_j=n,
        \end{cases}
        \qquad
        \widehat F_j^2=\norm{WA(I-P_j)}_F^2.
       \end{equation*}

 \item Set $\beta_j=c_0\alpha^4/h_j$ when $h_j>0$ and
       $\beta_0=c_0\alpha^4/R$.  Test each candidate, in increasing order of $j$, using
       \begin{equation*}
        \frac{\gamma_j^2+2\beta_j\widehat F_j^2}{1-\alpha}
        \le\frac{D_{\mathrm{cert}}\widehat F_j^2}{R}.
       \end{equation*}
       If a candidate passes, select the first such index $\widehat j$ and mark the branch
       \emph{certified}.  If none passes, set $\widehat j=J$ and mark the branch
       \emph{negligible}.  Write $Q=Q_{\widehat j}$, $P=QQ^T$, and $h=h_{\widehat j}$.

 \item Estimate the head.  If $h=0$, set $\widehat H=0$.  Otherwise form
       \begin{equation*}
        X=(AR_0)(Q^TR_0)^\dagger
       \end{equation*}
       and set $\widehat H=\norm X_{\Sone}$.

 \item If the branch is negligible, return $\widehat N=\widehat H$.

 \item In the certified branch, form the stored tail sketch
       \begin{equation*}
        YC=YA(I-P).
       \end{equation*}
       If $YC=0$, return $\widehat N=\widehat H$.  Otherwise set $B_\lambda=D_0n/R$, put
       $z_a=\sqrt{s_T}(YC)_{a,*}$, and compute the injective cycle statistics
       $\widehat p_1,\ldots,\widehat p_K$ from \eqref{eq:5.29}.  Using the fixed polynomial
       $P_K^*(u)=\sum_{r=0}^Ka_ru^r$ from Theorem~\ref{ext:poly}, compute
       $\widehat\mu_r$, $\widehat q$, and $\widehat T(C)$ from
       \eqref{eq:5.32}--\eqref{eq:5.32a}.  Return
       \begin{equation*}
        \widehat N=\widehat H+\widehat T(C).
       \end{equation*}
\end{enumerate}
\end{algorithmblock}

The stored state contains exactly $n(s_G+d+t+s_T)$ real linear measurements.  Every operation in
the decoding phase is postprocessing of \eqref{eq:5.37}, so the sketch is oblivious and linear.
Equations \eqref{eq:5.17}, \eqref{eq:5.24}, and the tail guarantee give error at most
$(7\alpha+\rho)\norm A_{\Sone}$ in the certified branch.  Equation \eqref{eq:5.15} gives error at
most $(7\alpha+\eta)\norm A_{\Sone}$ in the negligible branch; Appendix~\ref{app:decoder}
verifies the edge cases and both triangle inequalities.

For $R=\lceil n/b\rceil$, the definition of $J$ gives
$J\le C_{\mathrm{grid}}(\eps)\log(2+b)$ for one fixed $C_{\mathrm{grid}}(\eps)$, and hence
\begin{equation}
 H\le C_{\mathrm{grid}}(\eps)\frac{n\log(2+b)}b.          \tag{5.38}\label{eq:5.38}
\end{equation}
The covariance, certificate, and regression blocks use
\begin{equation}
 C_{\mathrm{blocks}}(\eps)nH
 \le C_{\mathrm{blocks}}(\eps)C_{\mathrm{grid}}(\eps)
       \frac{n^2\log(2+b)}b                               \tag{5.39}\label{eq:5.39}
\end{equation}
scalar measurements.  A nonzero certified tail has $B_\lambda\le D_0b$.  For the choice
$K=\lceil4C_{\mathrm{poly}}D_0b/\rho\rceil$ used in Algorithm~\ref{alg:upper-sketch},
equation \eqref{eq:5.36} has the form
\begin{equation}
 C_{\mathrm{opt},1}(\eps)b^{C_{\mathrm{opt},2}(\eps)}
 n^{2-c_{\mathrm{opt},1}(\eps)/b}.                       \tag{5.40}\label{eq:5.40}
\end{equation}
Here $C_{\mathrm{opt},1}(\eps)$, $C_{\mathrm{opt},2}(\eps)$, and
$c_{\mathrm{opt},1}(\eps)$ are three fixed positive constants.
The algorithm uses the choice
\begin{equation}
 b=c_{\mathrm{bal}}(\eps)\frac{\log(e n)}{\log\log(e^e n)} \tag{5.41}\label{eq:5.41}
\end{equation}
with one fixed sufficiently small $c_{\mathrm{bal}}(\eps)$.  Then the negative factor
$n^{-c_{\mathrm{opt},1}(\eps)/b}$ is a sufficiently large negative power of $\log n$ to absorb
the fixed polynomial $b^{C_{\mathrm{opt},2}(\eps)}$.  Meanwhile,
\eqref{eq:5.39} becomes
\begin{equation}
 C_{\mathrm{opt},3}(\eps)
 \frac{n^2\{\log\log(e^e n)\}^2}{\log(e n)}
 .                                                         \tag{5.42}\label{eq:5.42}
\end{equation}
Assign failure probability $1/100$ to each of the simultaneous covariance event, simultaneous
Frobenius event, selected-head regression, and selected-tail estimate.  Freshness of $R_0$ and $Y$
after the selection based on $G,W$ justifies conditioning.  The joint success probability is at
least $0.96>2/3$.  Choose
$\alpha=\rho=\eta=\eps/20$, so the total relative error is at most $2\eps/5<\eps$.
After \eqref{eq:5.40} is absorbed into \eqref{eq:5.42}, let $C_{\mathrm{asym}}(\eps)$ be the sum of the fixed
coefficients of the four stored blocks.  For the finitely many $n<n_0(\eps)$, store all $n^2$
entries and let $C_{\mathrm{pre}}(\eps)$ be the maximum of the resulting finite set of ratios to
the right side of \eqref{eq:1.2}.  The single theorem constant
$C_\eps\defeq\max\{C_{\mathrm{asym}}(\eps),C_{\mathrm{pre}}(\eps)\}$ then proves the upper
half of \eqref{eq:1.2} for every $n$ in the stated range.
This proves the upper half of Theorem~\ref{thm:main}.

\section{Discussion}

Theorem~\ref{thm:main} resolves, up to polylogarithmic factors, the general-linear-sketch question
for Schatten--1 left open by the previous $\Omega(n)$ lower bound and $O(n^2)$ upper
bound~\cite{LiNguyenWoodruff}.  It determines the polynomial order of the number of measurements:
for every fixed accuracy, this number is $n^{2-o(1)}$.  Equivalently, no
$O(n^{2-c})$-measurement sketch exists for any fixed $c>0$.  The explicit upper bound nevertheless
saves a factor of order $\log n/(\log\log n)^2$ over storing all $n^2$ entries.

The broader consequence is the classification in \eqref{eq:classification}.
Corollary~\ref{cor:noneven} shows that every fixed finite non-even Schatten exponent has nearly
quadratic sketch dimension, while Corollary~\ref{cor:noneven-upper} supplies a nontrivial
polylogarithmic saving below storage for every such exponent.  Positive even exponents behave differently:
$p=2q$ exposes the polynomial trace moment $\tr(A^TA)^q$, and the tight complexity is
$\Theta_{p,\eps}(n^{2-4/p})$~\cite{LiWoodruff}.  At the limiting exponent $p=\infty$, the
operator norm again requires $\Theta_\eps(n^2)$ measurements~\cite{LiWoodruff}.  Thus the new
non-even lower bound supplies the missing regime and, together with the earlier even-$p$ and
operator-norm results, resolves the polynomial-order sketching question for all Schatten norms.
The separation between an even exponent and every neighboring fixed non-even exponent is
polynomial in $n$.

For Schatten--1, this is markedly different from the vector $\ell_1$ problem, which has sketches
whose dimension is independent of the ambient dimension for fixed accuracy, and from the
positive-semidefinite matrix problem, where the nuclear norm is simply the trace.  The obstruction
is therefore not the sum of the singular values by itself; it is the need to infer that sum when
both singular directions are unknown.

The lower-bound construction explains one reason this inference is difficult.  Its two input
distributions have separated nuclear norms but matching low even spectral moments.  Random left
and right singular vectors then hide the remaining difference from any fixed low-dimensional
linear observation.  Gaussian smoothing makes the observed distributions differentiable, and
the moment-preserving path converts the indistinguishability question into an integral of Fisher
information.  The argument is decoder-independent: once the two observed laws are close in total
variation, no measurable postprocessing of the sketch can reliably distinguish them.  More
broadly, this combination of moment matching, random rotations, and a smooth path may be useful
for lower bounds for other unitarily invariant quantities.  Corollary~\ref{cor:noneven}
makes one such extension explicit: replacing the square-root moment by $x^{p/2}$ gives the same
polynomial-order lower bound for every fixed finite $p>0$ that is not a positive even integer.
The distinction is sharp for this construction.  At $p=2q$, the endpoint functional is the
matched polynomial moment $x^q$; away from those exponents, a single anchor uniquely dominates
the noninteger moment.  Above $p=2$, controlling the Gaussian perturbation in operator norm
rather than Frobenius norm keeps its Schatten-$p$ size on the same scale as the hard instance.

The upper bounds give a complementary picture.  A small number of very large singular values
must be treated individually, because low-degree moment estimates are unreliable when the
spectrum is highly concentrated.  After removing such a head, however, a residual with large
stable rank has a sufficiently spread spectrum that polynomial approximation and unbiased moment
estimators become effective.  The multiscale cutoff chooses between these two regimes using only
stored sketches.  Balancing the cost of recovering the head against the degree and variance
needed for the tail produces the saving
$\log n/(\log\log n)^2$ over storing the whole matrix at $p=1$.  For general non-even $p$, the
same architecture estimates the exact leading power sum from the row sketch and replaces the
square-root polynomial by one for $x^{p/2}$, giving a conservative fixed polylogarithmic saving.
This head--tail principle is conceptually
similar to heavy-hitter decompositions in vector streaming, but the head here is a pair of
unknown singular subspaces rather than a set of coordinates.

The remaining quantitative gaps are polylogarithmic rather than polynomial.  For Schatten--1,
the lower bound is $n^2/(\log n)^{A_\eps}$, with an exponent depending on the fixed accuracy,
whereas the new upper bound is
$O_\eps(n^2(\log\log n)^2/\log n)$.  For every other fixed finite non-even exponent,
Corollary~\ref{cor:noneven-upper} gives $n^2/(\log n)^{c_p}$ measurements with an explicit,
conservative $c_p>0$.  The paper does not determine the optimal logarithmic savings in those
regimes.  Closing these gaps would require either a harder pair that
remains indistinguishable after more measurements or new sketches that exploit structure
specific to the target exponent.

Dependence on the approximation parameter is another open issue.  The theorem fixes
$0<\eps<1$ before letting $n$ grow, and the constants in both proofs are not optimized.  It would
be useful to determine the joint dependence on $n$ and $\eps$, particularly when $\eps$ tends to
zero with $n$.  Such a result would need to track the number of matched moments, the smoothing
scale, the polynomial degree, and the Gaussian row counts more sharply than is needed for
Theorem~\ref{thm:main}.  In the non-even extension, the constants are also nonuniform as $p$
approaches a positive even integer.  Understanding a regime in which $p$ itself approaches an
even integer with $n$ would require a quantitative crossover analysis not supplied here.

The finite-memory consequence is now nearly quadratic as well.  The older turnstile
linearization theorem~\cite{LiNguyenWoodruffLinearization} does not by itself turn a real row
lower bound into a bit lower bound, because it controls the number of attainable sketch states.
For smooth approximation problems, however, Jiang, Liu, and Yu obtain a bounded-entry integer
sketch with $O(S/\log R)$ rows from an $S$-bit polynomial-length turnstile algorithm
\cite{JiangLiuYuLinearization}.  The lifting framework of Gribelyuk, Lin, Woodruff, Yu, and Zhou
then converts such an integer sketch into a real sketch with only a constant-factor increase in
dimension~\cite{GribelyukEtAlLifting}.  Appendix~\ref{app:streaming} verifies that the present
Gaussian-plus-orbit hard pair survives the required scaling, entrywise rounding, support
truncation, and discrete-Gaussian mollification for every fixed finite non-even $p>0$.
Consequently Corollary~\ref{cor:streaming} gives
\[
 \Omega_{p,\eps}\!\left(\frac{n^2\log W}{(\log n)^{B_{p,\eps}}}\right)
\]
bits for unit-update turnstile streams of length at most $W$, whenever $W$ is a sufficiently
large polynomial in $n$.  This strengthens the earlier nearly linear finite-bit lower bound for
non-even Schatten powers on sparse matrices~\cite{LiWoodruffSingularValues}, but in a different
input regime: the new hard matrices are dense, unrestricted integer matrices and the required
polynomial exponent in the stream length is not optimized.

The upper bound is also not automatically a Turing-machine algorithm.  Its Gaussian maps contain
real coefficients, the stored measurements are exact, and the sketch-dimension model does not
charge the decoder's working memory.  A finite-precision implementation would have to discretize the random maps,
bound accumulated rounding error in the cutoff, regression, and moment computations, and count
the random seed and working memory.  If such an implementation used $b$ bits per measurement,
the stored sketch would cost $O(kb)$ bits plus its seed and workspace; Theorem~\ref{thm:main}
identifies the target value of $k$ but supplies no bound on $b$.  Establishing a numerically stable
discretization with polylogarithmic $b$ is a natural open problem.

Likewise, row-order streams, sparse matrices, strict turnstile streams, bilinear sketches, update
time, decoding time, and numerical stability impose different requirements, as the results
reviewed in Section~1.2 illustrate.  Corollary~\ref{cor:streaming} concerns general turnstile
streams and does not assert the same bound under a nonnegativity promise on every intermediate
matrix.  Rectangular or complex matrices are natural further settings, but their analogous bounds
are not proved in this paper.

\section*{Acknowledgments}

The author thanks Yinchen Liu for pointing out, through personal communication, the recent
mollified turnstile-to-sketch transfer of Jiang, Liu, and
Yu~\cite{JiangLiuYuLinearization}, and for suggesting its relevance, together with the
integer-to-real lifting framework, to the finite-memory streaming consequence developed in
Corollary~\ref{cor:streaming} and Appendix~\ref{app:streaming}.

\appendix

\section{Testing, scalar moments, and path constructions}\label{app:scalar-path}

\subsection{Proof of the Yao--Le Cam reduction}\label{app:yao-proof}

\begin{proof}[Proof of Lemma~\ref{lem:yao}]
Assume, for contradiction, that a randomized sketch satisfies the functional-estimation
guarantee in the lemma statement.  Introduce a hidden binary variable $B$, uniform on
$\{0,1\}$.  Conditional on $B=i$, sample the input matrix $A\sim\Pi_i$.  The value $B$ is the
\emph{label}: it records which of the two input distributions generated $A$.

Let $\Omega$ denote all randomness used to choose the sketch map and decoder.  For every fixed
matrix $A$, the assumed guarantee says that the decoder output is a valid $(1\pm\eps)$ estimate
of $\Phi(A)$ with probability at least $2/3$ over $\Omega$.  Averaging this inequality over the
random pair $(B,A)$ gives
\[
 \Prb_{B,A,\Omega}\{\text{the estimate is valid for }A\}\ge\frac23.
\]
Therefore there is at least one fixed realization $\omega_0$ of the sketch randomness for which
\begin{equation}
 \Prb_{B,A}\{\text{the fixed sketch at }\omega_0
                  \text{ estimates }\Phi(A)\text{ validly}\}
 \ge\frac23.                                               \tag{4.0a}\label{eq:4.0a}
\end{equation}
Fix that realization.  Its measurement map and decoder are now deterministic.  Factor the map as
$S=TL$ as in Section~\ref{sec:lower}.  From $LA$ one can compute $SA=T(LA)$ and hence run the fixed
decoder, so its output, denoted $\widehat\Phi(LA)$, is a deterministic function of $LA$.

Let $J_i=(1\pm\eps)I_i$.  These two sets are disjoint by hypothesis.  Define a test
$\widehat B$ from the observed vector $LA$ as follows:
\[
 \widehat B=
 \begin{cases}
 0,&\widehat\Phi(LA)\in J_0,\\
 1,&\widehat\Phi(LA)\in J_1,\\
 0,&\widehat\Phi(LA)\notin J_0\cup J_1.
 \end{cases}
\]
The value assigned outside $J_0\cup J_1$ is immaterial.  Consider the two events
\[
 E=\{\widehat\Phi(LA)\text{ is a valid }(1\pm\eps)
                 \text{ estimate of }\Phi(A)\},
 \qquad
 G=\{\Phi(A)\in I_B\}.
\]
On $E\cap G$, the estimate belongs to $J_B$, and therefore $\widehat B=B$.  Equation \eqref{eq:4.0a}
gives $\Prb(E)\ge2/3$.  Moreover,
\[
 \Prb(G^c)
 =\frac12\sum_{i=0}^1
   \Prb_{A\sim\Pi_i}\{\Phi(A)\notin I_i\}
 \le\zeta.
\]
The union bound now gives
\begin{equation}
 \Prb\{\widehat B=B\}
 \ge\Prb(E\cap G)
 \ge\frac23-\zeta.                                        \tag{4.0b}\label{eq:4.0b}
\end{equation}

On the other hand, conditional on $B=i$, the observation $LA$ has law $L_\#\Pi_i$.
For a deterministic test $\psi$, put $\mathcal A=\{x:\psi(x)=0\}$.  Its success probability is
\[
\begin{aligned}
 \frac12(L_\#\Pi_0)(\mathcal A)+\frac12(L_\#\Pi_1)(\mathcal A^c)
 &=\frac12\left[1+(L_\#\Pi_0)(\mathcal A)-(L_\#\Pi_1)(\mathcal A)\right]\\
 &\le\frac{1+\TV(L_\#\Pi_0,L_\#\Pi_1)}2.
\end{aligned}
\]
The supremum over $\mathcal A$ attains the corresponding equality; randomized tests cannot do
better because one may condition on their internal randomness.  This is Le Cam's binary testing
identity; see~\cite[Proposition~2.3.1]{Duchi}.  In particular, every test based on $LA$ has
success at most
\[
 \frac{1+\TV(L_\#\Pi_0,L_\#\Pi_1)}2
 \le\frac{1+\delta}{2}.
\]
But $\delta+2\zeta<1/3$ is equivalent to
$2/3-\zeta>(1+\delta)/2$, contradicting \eqref{eq:4.0b}.
\end{proof}

\subsection{Proof that Gaussian smoothing preserves the norm gap}\label{app:smoothing-proof}

\begin{proof}[Proof of Lemma~\ref{lem:smoothing}]
Since $\E\norm Z_F^2=n^2$, Markov's inequality gives
$\Prb\{\norm Z_F>n/\sqrt\zeta\}\le\zeta$.  On the complementary event,
\[
 \norm{\sigma Z}_{\Sone}\le\sqrt n\,\sigma\norm Z_F
 \le\frac{\sigma n^{3/2}}{\sqrt\zeta}
 =:r_{\mathrm{noise}}.
\]
Define the deterministic noiseless nuclear-norm level under label $b$ by
\[
 \mathcal N_b\defeq\norm{aF_b}_{\Sone}
 =a\sqrt n\tr D_b=an^{3/2}m_b.
\]
The triangle inequality gives
$\abs{\norm{X_b}_{\Sone}-\mathcal N_b}\le r_{\mathrm{noise}}$ on the event above.  Accordingly,
consider arbitrary values $y_b$ satisfying
$\abs{y_b-\mathcal N_b}\le r_{\mathrm{noise}}$.  Then
\begin{align*}
 &(y_1-y_0)-\eps(y_1+y_0)\\
 &\quad\ge an^{3/2}\{(m_1-m_0)-\eps(m_1+m_0)\}
          -2(1+\eps)r_{\mathrm{noise}}\\
 &\quad\ge an^{3/2}\eta(m_1+m_0)-4r_{\mathrm{noise}}
 \ge\frac12an^{3/2}\eta(m_1+m_0)>0
\end{align*}
by \eqref{eq:4.20}.  Thus one may take the deterministic intervals
$I_b=[\max\{0,\mathcal N_b-r_{\mathrm{noise}}\},
          \mathcal N_b+r_{\mathrm{noise}}]$; the preceding inequality says exactly that their
$(1\pm\eps)$ enlargements are disjoint.  Finally, vectorized $Z$ is standard Gaussian in
$\R^{n^2}$; $LL^*=I_k$ therefore implies $LZ\sim N(0,I_k)$.  Moreover, $LZ$ is independent of
$(U,V)$ because $Z$ is independent of $(U,V)$ and $L$ is deterministic.
\end{proof}

\subsection{The moment-matched scalar pair}

\begin{proof}[Proof of Proposition~\ref{prop:scalar-pair}]
Fix $K\ge2$.  Begin with an integer $M_\star\ge6$; its value will be chosen as
$M_\star=M_\eps$ at the end.  Let
$x_j=j^{M_\star}$ for $0\le j\le K+1$, and set
\begin{equation}
 w_j=\left(\prod_{\ell\ne j}(x_j-x_\ell)\right)^{-1}.        \tag{4.1}\label{eq:4.1}
\end{equation}
Lagrange interpolation implies
\begin{equation}
 \sum_{j=0}^{K+1}w_jP(x_j)=0
 \qquad\text{for every polynomial }P\text{ of degree at most }K. \tag{4.2}\label{eq:4.2}
\end{equation}
Indeed, interpolate $P$ at the $K+2$ points; the left side is the coefficient of $z^{K+1}$,
which is zero because $\deg P\le K$.

The signs of the $w_j$ alternate.  Indeed,
$x_0<x_1<\cdots<x_{K+1}$.  In the denominator
\[
 \prod_{\ell\ne j}(x_j-x_\ell),
\]
the $j$ factors with $\ell<j$ are positive, while the $K+1-j$ factors with $\ell>j$ are
negative.  Therefore
\begin{equation}
 \operatorname{sign}(w_j)=(-1)^{K+1-j},                   \tag{4.2a}\label{eq:4.2a}
\end{equation}
which changes sign whenever $j$ increases by one.  Separate the signed weights into two positive
measures:
\[
 \mu_+=\sum_{j:w_j>0}w_j\,\delta_{x_j},
 \qquad
 \mu_-=\sum_{j:w_j<0}(-w_j)\,\delta_{x_j},
\]
where $\delta_x$ denotes one unit of probability mass at the point $x$.  Taking $P(x)=1$ in
\eqref{eq:4.2} shows that these measures have the same total mass:
\[
 Z\defeq\mu_+(\R)=\sum_{j:w_j>0}w_j
 =\sum_{j:w_j<0}(-w_j)=\mu_-(\R).
\]
Hence $\mu_+/Z$ and $\mu_-/Z$ are probability laws.  The already defined point
$x_0=0^{M_\star}=0$ belongs to exactly one of the two sign classes because $w_0\ne0$.
Name the probability law containing this atom $\widetilde\nu_0$, and name the other one
$\widetilde\nu_1$.  Equation \eqref{eq:4.2} then gives
\begin{equation}
 \int x^q\,d\widetilde\nu_0(x)=\int x^q\,d\widetilde\nu_1(x),
 \qquad 0\le q\le K.                                      \tag{4.3}\label{eq:4.3}
\end{equation}

The following estimates quantify the separation.  Let $R_j=\abs{w_j}/\abs{w_0}$.  Each symbol
$C_{\mathrm{pair},i}$ introduced below denotes a fixed positive absolute constant whose value
does not change later.  Direct cancellation in
\eqref{eq:4.1} gives
\begin{align}
 R_1&=\prod_{\ell=2}^{K+1}(1-\ell^{-M_\star})^{-1},
 &\abs{R_1-1}&\le C_{\mathrm{pair},1}2^{-M_\star},        \tag{4.4}\label{eq:4.4}\\
 2^{M_\star}R_2&=(1-2^{-M_\star})^{-1}
      \prod_{\ell=3}^{K+1}(1-(2/\ell)^{M_\star})^{-1},
 &\abs{2^{M_\star}R_2-1}&\le C_{\mathrm{pair},2}(2/3)^{M_\star}. \tag{4.5}\label{eq:4.5}
\end{align}
To prove the bound in \eqref{eq:4.4}, observe that every factor in the product for $R_1$ is at least one,
so $R_1\ge1$.  Since $M_\star\ge6$, we have $\ell^{-M_\star}\le2^{-6}<1/2$, and the elementary
inequality $-\log(1-u)\le2u$ for $0\le u\le1/2$ gives
\[
\begin{aligned}
 0\le\log R_1
 &\le2\sum_{\ell=2}^{K+1}\ell^{-M_\star}
 \le2\left(2^{-M_\star}+
           \int_2^\infty x^{-M_\star}\,dx\right)\\
 &=2^{1-M_\star}\left(1+\frac{2}{M_\star-1}\right)
 \le 3\,2^{-M_\star},
\end{aligned}
\]
and the exact expression on the preceding line is at most $7/160$ when $M_\star\ge6$.
Set $u=\log R_1$.  Although $R_1>1$, the preceding calculation shows that
$0\le u\le7/160$.  By the mean-value theorem,
$e^u-1=e^\xi u\le e^{7/160}u\le2u$ for some $\xi\in[0,u]$.  Hence
\[
 |R_1-1|=R_1-1=e^{\log R_1}-1\le6\,2^{-M_\star}.
\]
Thus \eqref{eq:4.4} holds with the fixed choice $C_{\mathrm{pair},1}=6$.
We prove \eqref{eq:4.5} in the same way.  First, the cancellation can be seen explicitly from
$x_0=0$, $x_1=1$, and $x_j=j^{M_\star}$:
\[
\begin{aligned}
 R_2
 &=\frac{\prod_{\ell=1}^{K+1}x_\ell}
         {x_2(x_2-x_1)\prod_{\ell=3}^{K+1}(x_\ell-x_2)}\\
 &=\frac1{2^{M_\star}-1}
   \prod_{\ell=3}^{K+1}
   \left(1-(2/\ell)^{M_\star}\right)^{-1}.
\end{aligned}
\]
Multiplication by $2^{M_\star}$ gives the product displayed in \eqref{eq:4.5}.  Put
$S=2^{M_\star}R_2$.  Every factor in this product is at least one, so $S\ge1$.  Moreover,
$2^{-M_\star}\le1/64$ and $(2/\ell)^{M_\star}\le(2/3)^6<1/2$.  Hence
\[
\begin{aligned}
 0\le\log S
 &\le2\left\{2^{-M_\star}
       +\sum_{\ell=3}^{K+1}(2/\ell)^{M_\star}\right\}\\
 &\le2\left\{2^{-M_\star}+(2/3)^{M_\star}
       +\int_3^\infty(2/x)^{M_\star}\,dx\right\}\\
 &=2\left\{2^{-M_\star}
       +(2/3)^{M_\star}\left(1+\frac3{M_\star-1}\right)\right\}\\
 &\le4(2/3)^{M_\star},
\end{aligned}
\]
and the preceding explicit expression is at most $1/3$ for every $M_\star\ge6$.
Set $v=\log S$.  Then $S=e^v$ and $0\le v\le1/3$.  Apply the mean-value theorem to
$f(t)=e^t$ on the interval $[0,v]$: there is a number $\xi\in[0,v]$ such that
\[
 S-1=e^v-e^0=f'(\xi)(v-0)=e^\xi v.
\]
Since $0\le\xi\le v\le1/3$, we have $e^\xi\le e^{1/3}<2$, and consequently
$S-1\le2v=2\log S$.  Therefore
\[
 \left|2^{M_\star}R_2-1\right|
 =S-1\le8(2/3)^{M_\star}.
\]
Thus \eqref{eq:4.5} holds with the fixed choice $C_{\mathrm{pair},2}=8$.

For integers $0\le\ell<j$, the difference-of-powers identity gives
\[
 j^{M_\star}-\ell^{M_\star}
 =(j-\ell)\left(
 j^{M_\star-1}+j^{M_\star-2}\ell+\cdots+
 j\ell^{M_\star-2}+\ell^{M_\star-1}\right).
\]
Every term inside the parentheses is nonnegative, so retaining only its first term yields
\begin{equation}
 j^{M_\star}-\ell^{M_\star}
 \ge(j-\ell)j^{M_\star-1}.                                \tag{4.5a}\label{eq:4.5a}
\end{equation}
For $j\ge2$, canceling the common factor $x_j=j^{M_\star}$ in the definition of $R_j$ and
splitting the remaining factors at $j$ gives
\begin{equation}
\begin{aligned}
 R_j&=A_jB_{j,K},\\
 A_j&=\prod_{\ell=1}^{j-1}
       \frac{\ell^{M_\star}}{j^{M_\star}-\ell^{M_\star}},\\
 B_{j,K}&=\prod_{\ell=j+1}^{K+1}
       \frac{\ell^{M_\star}}{\ell^{M_\star}-j^{M_\star}}
 =\prod_{\ell=j+1}^{K+1}
       \left(1-(j/\ell)^{M_\star}\right)^{-1}.
\end{aligned}                                             \tag{4.5b}\label{eq:4.5b}
\end{equation}
For the factors below $j$, inequality \eqref{eq:4.5a} gives
\begin{equation}
\begin{aligned}
 A_j
 &\le\prod_{\ell=1}^{j-1}
       \frac{\ell^{M_\star}}{(j-\ell)j^{M_\star-1}}\\
 &=\left\{\prod_{\ell=1}^{j-1}\frac{\ell}{j-\ell}\right\}
   \left\{\prod_{\ell=1}^{j-1}\frac{\ell}{j}\right\}^{M_\star-1}\\
 &=\left(\frac{(j-1)!}{j^{j-1}}\right)^{M_\star-1}.
\end{aligned}
\tag{4.5c}\label{eq:4.5c}
\end{equation}
The first product in braces equals one because both its numerator and denominator are
$(j-1)!$.

For the factors above $j$, expand $-\log(1-u)=\sum_{r\ge1}u^r/r$:
\begin{equation}
\begin{aligned}
 \log B_{j,K}
 &=\sum_{\ell=j+1}^{K+1}\sum_{r=1}^\infty
       \frac1r\left(\frac j\ell\right)^{M_\star r}\\
 &\le\sum_{r=1}^\infty\frac1r
       \int_j^\infty\left(\frac jx\right)^{M_\star r}dx\\
 &=\sum_{r=1}^\infty\frac1r\frac{j}{M_\star r-1}\\
 &\le\frac{2j}{M_\star}\sum_{r=1}^\infty\frac1{r^2}
 =\frac{\pi^2j}{3M_\star}.
\end{aligned}
\tag{4.5d}\label{eq:4.5d}
\end{equation}
The integral bounds the sum because $x\mapsto(j/x)^{M_\star r}$ is decreasing, and
$M_\star r-1\ge M_\star r/2$.  Exponentiating \eqref{eq:4.5d}, multiplying by the bound \eqref{eq:4.5c}, and using
\[
 j^{M_\star}
 \left(\frac{(j-1)!}{j^{j-1}}\right)^{M_\star-1}
 =
 j\left(\frac{(j-1)!}{j^{j-2}}\right)^{M_\star-1}
\]
gives
\begin{equation}
 j^{M_\star}R_j
 \le j\left(\frac{(j-1)!}{j^{j-2}}\right)^{M_\star-1}
       \exp\!\left(\frac{\pi^2j}{3M_\star}\right).        \tag{4.6}\label{eq:4.6}
\end{equation}
The right side decreases geometrically after $j=3$, because
\[
 \frac{(j!)/(j+1)^{j-1}}{(j-1)!/j^{j-2}}
 =\left(\frac{j}{j+1}\right)^{j-1}\le\frac9{16}.
\]
Indeed, the ratio of two consecutive right sides in \eqref{eq:4.6} is at most
\[
 q_\star\defeq\frac43\left(\frac9{16}\right)^5e^{\pi^2/18}<1
 \qquad(M_\star\ge6,\ j\ge3).
\]
The $j=3$ term is at most
$\frac92e^{\pi^2/6}(2/3)^{M_\star}$.  Consequently
\begin{equation}
 \sum_{j\ge3}j^{M_\star}R_j
 \le C_{\mathrm{pair},3}(2/3)^{M_\star}.                 \tag{4.7}\label{eq:4.7}
\end{equation}
Here the fixed absolute constant is
\[
 C_{\mathrm{pair},3}
 \defeq \frac{(9/2)e^{\pi^2/6}}{1-q_\star}.
\]
These weight estimates imply the required properties of the probability laws.  In units of
$|w_0|$, the atom at $x_0=0$ has weight $1$.  Let
\[
 E_0=\sum_{\substack{j\ge2\\w_jw_0>0}}R_j
\]
be the relative weight of all other atoms in the same sign class.  Equation \eqref{eq:4.5} gives
\[
 R_2\le (1+C_{\mathrm{pair},2})2^{-M_\star}
 =C_{\mathrm{pair},4}2^{-M_\star}
 \le C_{\mathrm{pair},4}(2/3)^{M_\star},
 \qquad C_{\mathrm{pair},4}\defeq1+C_{\mathrm{pair},2}.
\]
Equation \eqref{eq:4.7} and $j^{M_\star}\ge1$ give
\[
 \sum_{j\ge3}R_j
 \le\sum_{j\ge3}j^{M_\star}R_j
 \le C_{\mathrm{pair},3}(2/3)^{M_\star}.
\]
Consequently
\[
 0\le E_0\le(C_{\mathrm{pair},4}+C_{\mathrm{pair},3})(2/3)^{M_\star}
 =C_{\mathrm{pair},5}(2/3)^{M_\star},
\]
where $C_{\mathrm{pair},5}\defeq C_{\mathrm{pair},4}+C_{\mathrm{pair},3}$.  The normalizing mass of this sign class is
$1+E_0$, and hence
\begin{equation}
 \widetilde\nu_0(\{0\})=\frac1{1+E_0},\qquad
 0\le1-\widetilde\nu_0(\{0\})
 =\frac{E_0}{1+E_0}\le C_{\mathrm{pair},5}(2/3)^{M_\star}. \tag{4.7b}\label{eq:4.7b}
\end{equation}
The other sign class contains $j=1$ and, by alternation, excludes $j=2$.  Let
\[
 E_1=\sum_{\substack{j\ge3\\w_jw_1>0}}R_j.
\]
Equation \eqref{eq:4.7} gives $0\le E_1\le C_{\mathrm{pair},3}(2/3)^{M_\star}$, while \eqref{eq:4.4} gives $R_1\ge1$.  Therefore
\begin{equation}
 \widetilde\nu_1(\{1\})=\frac{R_1}{R_1+E_1},\qquad
 0\le1-\widetilde\nu_1(\{1\})
 =\frac{E_1}{R_1+E_1}\le C_{\mathrm{pair},3}(2/3)^{M_\star}. \tag{4.7c}\label{eq:4.7c}
\end{equation}
By \eqref{eq:4.3}, the two laws have a common first moment; call it $\mu$.  Computing it under
$\widetilde\nu_1$ gives the exact ratio
\[
 \mu=\frac{R_1+\sum_{j\ge3:\,w_jw_1>0}j^{M_\star}R_j}{R_1+E_1}.
\]
Subtracting one in the preceding exact ratio cancels $R_1$.  Equation \eqref{eq:4.7} bounds the remaining
extra numerator and $E_1$ separately by
$C_{\mathrm{pair},3}(2/3)^{M_\star}$.  Since $R_1+E_1\ge1$,
\begin{equation}
 |\mu-1|\le C_{\mathrm{pair},6}(2/3)^{M_\star},
 \qquad C_{\mathrm{pair},6}\defeq2C_{\mathrm{pair},3}.  \tag{4.7a}\label{eq:4.7a}
\end{equation}

Define $\nu_b$ as the law of $X/\mu$ when $X\sim\widetilde\nu_b$.  This common rescaling
preserves \eqref{eq:4.3}, gives both laws mean one, and, once $M_\star$ is large enough that
$1/2\le\mu\le2$, places their support in $[0,2(K+1)^{M_\star}]$.  It also leaves the atom of
$\nu_0$ at zero and moves the atom of $\nu_1$ from $1$ to $1/\mu$.  Consequently
the bounds \eqref{eq:4.7b}--\eqref{eq:4.7c} remain valid for
$\nu_0(\{0\})$ and $\nu_1(\{1/\mu\})$, respectively.

Since each $\nu_b$ has mean one, Cauchy--Schwarz and the single large atom of $\nu_1$ give
\begin{equation}
\begin{aligned}
 h_0
 &=\int_{\{x>0\}}\sqrt{x}\,d\nu_0(x)
 \le\sqrt{\nu_0(x>0)\int x\,d\nu_0(x)}
 \le\sqrt{C_{\mathrm{pair},5}}(2/3)^{M_\star/2},\\
 h_1
 &\ge \nu_1(\{1/\mu\})\sqrt{1/\mu}
 \ge\frac{1-C_{\mathrm{pair},3}(2/3)^{M_\star}}{\sqrt2}.
\end{aligned}                                             \tag{4.8}\label{eq:4.8}
\end{equation}
Define the envelope constant
\[
 C_{\mathrm{pair}}
 =\max\{C_{\mathrm{pair},1},\ldots,C_{\mathrm{pair},6},
          \sqrt{C_{\mathrm{pair},5}}\}.
\]
Thus every estimate in this construction is also valid with the one envelope constant
$C_{\mathrm{pair}}$, but the preceding proof records which local constant is used at each step.

An explicit admissible choice of $M_\eps$ is
\begin{equation}
 M_\eps=\max\left\{
 6,
 \left\lceil\frac{\log(2C_{\mathrm{pair},6})}{\log(3/2)}\right\rceil,
 \left\lceil
  \frac{2\log\!\left(
   \dfrac{2\sqrt{2C_{\mathrm{pair},5}}(3+\eps)}{1-\eps}
  \right)}{\log(3/2)}
 \right\rceil
 \right\}.                                               \tag{4.P4}\label{eq:4.P4}
\end{equation}
The second entry in the maximum ensures
\[
 C_{\mathrm{pair},6}(2/3)^{M_\eps}\le\frac12.
\]
Together with \eqref{eq:4.7a}, this gives $1/2\le\mu\le3/2$, so in particular the bound
$1/\sqrt\mu\ge1/\sqrt2$ used in \eqref{eq:4.8} is valid.  Since
$C_{\mathrm{pair},6}=2C_{\mathrm{pair},3}$, it also gives
\[
 \frac{1-C_{\mathrm{pair},3}(2/3)^{M_\eps}}{\sqrt2}
 \ge\frac1{2\sqrt2}.
\]
Thus \eqref{eq:4.8} gives $h_1\ge1/(2\sqrt2)$.  Moreover, \eqref{eq:4.8}, with
$a_\eps=(2/3)^{M_\eps/2}$, implies
\[
 \frac{h_0}{h_1}
 \le2\sqrt{2C_{\mathrm{pair},5}}\,a_\eps.
\]
The third entry in \eqref{eq:4.P4} makes the right side at most
$(1-\eps)/(3+\eps)$.  Therefore
\[
 \frac{h_1-h_0}{h_1+h_0}
 =\frac{1-h_0/h_1}{1+h_0/h_1}
 \ge
 \frac{1-(1-\eps)/(3+\eps)}{1+(1-\eps)/(3+\eps)}
 =\frac{1+\eps}{2}.
\]
Taking
\begin{equation}
 \eta_\eps=\frac{1-\eps}{4}                              \tag{4.P5}\label{eq:4.P5}
\end{equation}
gives $(1+\eps)/2=\eps+2\eta_\eps$, proving \eqref{eq:4.P3}.
Formula \eqref{eq:4.P4} also makes the dependence transparent: all constants inside it are absolute, and
the only unbounded term as $\eps\uparrow1$ is
$2\log(1/(1-\eps))/\log(3/2)$.  This proves \eqref{eq:4.P1}--\eqref{eq:4.P3}.
\end{proof}

\subsection{The equal-weight moment-preserving path}

\begin{proof}[Proof of Lemma~\ref{lem:path}]
The local constants in this proof are denoted by $C_{\mathrm{path},i}$ and
$c_{\mathrm{path},i}$; each has one fixed numerical value.  We combine them into
$C_{\mathrm{path}}$ only at the end.

\paragraph{Construction plan.}
The paths $x_i$ are constructed in four stages.
First, we replace each input law $\nu_b$ by $n$ equal-mass particles.  Most are
\emph{base particles}, which approximate the atoms of $\nu_b$; the remaining particles form
$K$ movable control groups.  Second, we use the $K$ control groups to correct the small errors
caused by rounding the base-particle counts to integers.  This produces two empirical endpoint
laws with exactly equal first $K$ moments.  Third, we move the base particles from endpoint
$b=0$ to endpoint $b=1$, one particle at a time.  During each such move, the control groups move
continuously so that all $K$ moments remain constant.  Finally, we concatenate these one-particle
segments and verify their regularity, length, and endpoint approximation.  The concatenated
trajectories are the desired $x_i$.

The numbers $t_j$ introduced next are locations in the scalar state space $[0,B]$; they are not
times along the path.  We use $\tau\in[0,1]$ for global path time below.

\paragraph{Stage 1: endpoint candidates.}
\emph{Endpoint lists and empirical laws.}
Write each input law as
\[
 \nu_b=\sum_{r=1}^{m_b}p_{b,r}\delta_{s_{b,r}},
 \qquad m_b\le K+2,\quad b\in\{0,1\}.
\]
Here $p_{b,r}>0$, $\sum_{r=1}^{m_b}p_{b,r}=1$, and $\nu_b$ assigns probability $p_{b,r}$ to the
scalar value $s_{b,r}\in[0,B]$.
For each $b$, the construction requires a list of exactly $n$ scalar values
\[
 \mathcal X_b=(x_1^{(b)},\ldots,x_n^{(b)})\in[0,B]^n
\]
whose empirical law
\[
 \frac1n\sum_{i=1}^n\delta_{x_i^{(b)}}
\]
approximates $\nu_b$.  Here $x_i^{(b)}=s_{b,r}$ means that the $i$th list entry has the
value $s_{b,r}$.  By contrast,
\[
 \#\{i:x_i^{(b)}=s_{b,r}\}
\]
is the \emph{number of copies} of that value.  The value $s_{b,r}$ and its number of copies are
different objects.

We also need the two empirical endpoint laws to have exactly equal moments through degree $K$.
Approximating $\nu_0$ and $\nu_1$ by integer copy counts will generally introduce a small error in
this equality.  Stage~1 constructs the approximate endpoint lists and measures that error;
Stage~2 will remove it using entries reserved for that purpose.

\medskip\noindent
\emph{Reserve $K$ adjustable groups.}
There are $K$ nonconstant moment equations, for degrees $1,\ldots,K$, so we reserve $K$ groups of
entries.  We use Chebyshev polynomials because they are bounded on $[0,B]$ and because the effect
of adjusting the groups can be inverted explicitly.  Set
\[
 T_q(\cos\theta)=\cos(q\theta),\qquad
 \theta_j=\frac{j\pi}{K+1},\qquad
 t_j=\frac B2(1+\cos\theta_j),\qquad 1\le j\le K,
\]
and
\[
 C_0(x)=1,\qquad C_q(x)=T_q(2x/B-1),\qquad 1\le q\le K.
\]
The numbers $t_j$ are scalar values in $[0,B]$, not path times.  Choose
\[
 N_c=\left\lfloor\frac{\eta n}{BK}\right\rfloor,
 \qquad w=\frac{N_c}{n},
 \qquad N_{\mathrm{base}}=n-KN_c.
\]
In each endpoint list, group $j$ consists of $N_c$ entries whose initial value is $t_j$:
\[
 \mathcal X_b^{\mathrm{ctrl}}
 =\bigl(\underbrace{t_1,\ldots,t_1}_{N_c\text{ copies}},
         \ldots,
         \underbrace{t_K,\ldots,t_K}_{N_c\text{ copies}}\bigr).
\]
Each group therefore contributes the measure $w\delta_{t_j}$.  The two endpoint lists contain
separate entries, but their control sublists are initially identical.  Consequently, the controls
initially contribute zero to the difference between the two endpoint moments.

The $K$ groups use $KN_c$ entries, leaving $N_{\mathrm{base}}$ entries to approximate $\nu_b$.
Their total empirical mass is small:
\[
 Kw=\frac{KN_c}{n}\le\frac{\eta}{B}\le\eta.
\]
Condition \eqref{eq:4.9}, with $C_{\mathrm{path}}$ sufficiently large, also ensures
$\eta n/(BK)\ge2$, and hence
\[
 \frac{\eta}{2BK}\le w\le\frac{\eta}{BK},
 \qquad N_{\mathrm{base}}\ge(1-\eta)n>0.
\]
The upper bound keeps the reserved mass small.  The lower bound gives each group enough empirical
mass to correct the later moment error without a large change in its common value; it is used in
\eqref{eq:4.13}.

\medskip\noindent
\emph{Use the remaining entries to approximate $\nu_b$.}
For a fixed support value $s_{b,r}$, the ideal number of copies among the
$N_{\mathrm{base}}$ remaining entries is
\[
 N_{\mathrm{base}}p_{b,r}.
\]
This is a desired \emph{copy count}; it is not the value assigned to an entry.  It need not be an
integer.  We choose integer copy counts $N_{b,r}$ by first taking all floors and then assigning the
remaining entries to the largest fractional remainders.  This largest-remainder rule gives
\begin{equation}
 \sum_{r=1}^{m_b}N_{b,r}=N_{\mathrm{base}},\qquad
 \left|N_{b,r}-N_{\mathrm{base}}p_{b,r}\right|<1.         \tag{4.10b}\label{eq:4.10b}
\end{equation}
Equivalently, the base portion of $\mathcal X_b$ contains exactly $N_{b,r}$ entries equal to
$s_{b,r}$.  The complete endpoint list may therefore be written, in any order, as
\[
 \mathcal X_b=\bigl(
  \underbrace{s_{b,1},\ldots,s_{b,1}}_{N_{b,1}\text{ copies}},\ldots,
  \underbrace{s_{b,m_b},\ldots,s_{b,m_b}}_{N_{b,m_b}\text{ copies}},
  \underbrace{t_1,\ldots,t_1}_{N_c\text{ copies}},\ldots,
  \underbrace{t_K,\ldots,t_K}_{N_c\text{ copies}}
 \bigr).
\]
Its length is exactly $N_{\mathrm{base}}+KN_c=n$.  Its empirical law, before any control values
are adjusted, is
\[
 \overline\nu_b
 =\frac1n\sum_{r=1}^{m_b}N_{b,r}\delta_{s_{b,r}}
  +w\sum_{j=1}^K\delta_{t_j}.
\]
Thus $\overline\nu_b$ is an equal-weight approximation to $\nu_b$, augmented by a small
common collection of adjustable entries.

\medskip\noindent
\emph{The moment error introduced by the integer copy counts.}
If fractional copy counts were allowed, the ideal endpoint measure would be
\[
 \overline\nu_b^{\mathrm{ideal}}
 =\frac{N_{\mathrm{base}}}{n}\nu_b+w\sum_{j=1}^K\delta_{t_j}.
\]
These two ideal measures have equal moments through degree $K$: the input laws have equal moments,
and the control term is identical for $b=0$ and $b=1$.  The actual law differs from the ideal one
only because $N_{\mathrm{base}}p_{b,r}$ was replaced by the integer $N_{b,r}$.

For $0\le q\le K$, define the actual Chebyshev-moment discrepancy by
\[
 e_q^{\mathrm{end}}
 =\int C_q\,d\overline\nu_1-\int C_q\,d\overline\nu_0.
\]
Writing
\[
 \Delta_{b,r}=N_{b,r}-N_{\mathrm{base}}p_{b,r},
\]
the preceding ideal cancellation gives the explicit formula
\[
 e_q^{\mathrm{end}}
 =\frac1n\left\{
   \sum_{r=1}^{m_1}\Delta_{1,r}C_q(s_{1,r})
   -\sum_{r=1}^{m_0}\Delta_{0,r}C_q(s_{0,r})
  \right\}.
\]
Because $|\Delta_{b,r}|<1$ and $|C_q(x)|\le1$ on $[0,B]$,
\[
 |e_q^{\mathrm{end}}|\le\frac{m_0+m_1}{n}\le\frac{2(K+2)}n
 \qquad(1\le q\le K).
\]
For $q=0$, the discrepancy is zero because both empirical laws have total mass one.  Write
$e^{\mathrm{end}}=(e_q^{\mathrm{end}})_{q=1}^K\in\R^K$.  Then
\begin{equation}
 \norm{e^{\mathrm{end}}}_\infty\le\frac{2(K+2)}n.        \tag{4.11}\label{eq:4.11}
\end{equation}

\paragraph{Stage 2: correct the rounded endpoint moments.}
To remove this rounding discrepancy, keep the $N_c$ entries of group $j$ in endpoint $b=0$ equal
to $t_j$, and change all $N_c$ corresponding entries in endpoint $b=1$ to
$t_j+\delta_j$, where
$\delta=(\delta_1,\ldots,\delta_K)$.  Define the vector-valued moment-change map
\[
 H(\delta)=(H_q(\delta))_{q=1}^K,
 \qquad
 H_q(\delta)=w\sum_{j=1}^K\{C_q(t_j+\delta_j)-C_q(t_j)\}.
\]
Here $DH(\delta)$ is the $K\times K$ Jacobian matrix of $H$, and
$\norm{\cdot}_{\infty\to\infty}$ is the matrix norm induced by the vector $\ell_\infty$ norm.
To compute $A=DH(0)$, note that $2t_j/B-1=\cos\theta_j$.  Differentiating
$T_q(\cos\theta)=\cos(q\theta)$ first with respect to $\theta$ and then with respect to its
argument gives
\[
 T_q'(\cos\theta)=q\frac{\sin(q\theta)}{\sin\theta},
 \qquad
 C_q'(t_j)=\frac{2q\sin(q\theta_j)}{B\sin\theta_j}.
\]
Since $\partial H_q/\partial\delta_j=wC_q'(t_j+\delta_j)$, this gives the displayed formula
for $A$ below.  The discrete sine orthogonality identity
\[
 \sum_{j=1}^K\sin(q\theta_j)\sin(r\theta_j)
 =\frac{K+1}{2}\,\mathbf 1_{\{q=r\}},
 \qquad 1\le q,r\le K,
\]
then yields the exact formulas
\begin{equation}
 A_{qj}=\frac{2wq\sin(q\theta_j)}{B\sin\theta_j},\qquad
 (A^{-1})_{jq}=\frac{B\sin\theta_j\sin(q\theta_j)}{w(K+1)q}, \tag{4.12}\label{eq:4.12}
\end{equation}
because direct multiplication gives
\[
 \sum_{j=1}^K A_{qj}(A^{-1})_{jr}
 =\frac{2q}{(K+1)r}\sum_{j=1}^K
   \sin(q\theta_j)\sin(r\theta_j)
 =\mathbf 1_{\{q=r\}}.
\]
Thus $A$ is invertible.  Summing the absolute values in a row of $A^{-1}$, using
$\sum_{q=1}^Kq^{-1}\le1+\log K$, and then using the lower bound on $w$ gives
\begin{equation}
 \norm{A^{-1}}_{\infty\to\infty}\le
 \frac{C_{\mathrm{path},1}B^2\log(K+1)}\eta.              \tag{4.13}\label{eq:4.13}
\end{equation}

The following estimates control the nonlinear remainder of $H$ near the origin.
Applying the Markov brothers' inequality twice gives
\begin{align*}
 \norm{T_q''}_{L^\infty[-1,1]}
 &\le(q-1)^2\norm{T_q'}_{L^\infty[-1,1]}\\
 &\le q^2(q-1)^2\norm{T_q}_{L^\infty[-1,1]}\\
 &=q^2(q-1)^2\\
 &\le q^4,
\end{align*}
where $\norm{T_q}_{L^\infty[-1,1]}=1$; see~\cite[equation~(3)]{KalmykovNagyTotik}.  Since
$C_q(x)=T_q(2x/B-1)$, the chain rule gives
\[
 \norm{C_q''}_{L^\infty[0,B]}\le\frac{4q^4}{B^2}.
\]
Whenever the segments from $t_j$ to $t_j+\delta_j$ lie in $[0,B]$, the mean-value theorem
therefore gives, row by row,
\begin{align*}
 \norm{DH(\delta)-A}_{\infty\to\infty}
 &=\max_{1\le q\le K}
   w\sum_{j=1}^K\abs{C_q'(t_j+\delta_j)-C_q'(t_j)}\\
 &\le\max_{1\le q\le K}
   w\sum_{j=1}^K\norm{C_q''}_{L^\infty[0,B]}\abs{\delta_j}\\
 &\le\frac{4wK^5}{B^2}\norm\delta_\infty.
\end{align*}
Combining this estimate with \eqref{eq:4.13} and $w\le\eta/(BK)$ gives
\[
 \norm{A^{-1}(DH(\delta)-A)}_{\infty\to\infty}
 \le C_{\mathrm{path},2}\frac{K^4\log(K+1)}B\norm\delta_\infty.
\]
Choose the fixed number $c_{\mathrm{path},1}>0$ small enough that, throughout the cube
\[
 \norm\delta_\infty\le\rho_0
 \defeq\frac{c_{\mathrm{path},1}B}{K^4\log(K+1)},
\]
both the following bound holds and every $t_j+\delta_j$ remains in $(0,B)$:
\begin{equation}
 \norm{A^{-1}(DH(\delta)-A)}_{\infty\to\infty}\le\frac12. \tag{4.14}\label{eq:4.14}
\end{equation}
For the second assertion, the distance from the closest node to either endpoint is
\[
 \min_j\min\{t_j,B-t_j\}
 =\frac B2\left(1-\cos\frac\pi{K+1}\right)
 \ge\frac{B}{(K+1)^2}\ge\frac{B}{4K^2},
\]
where we used $1-\cos u\ge2u^2/\pi^2$ for $0\le u\le\pi$.
The radius $\rho_0$ is smaller than half this distance when $c_{\mathrm{path},1}$ is sufficiently
small.  Fix $c_{\mathrm{path},2}=1/8$; then every point in the cube satisfies
\[
 t_j+\delta_j\ge\frac{c_{\mathrm{path},2}B}{K^2},
 \qquad
 B-(t_j+\delta_j)\ge\frac{c_{\mathrm{path},2}B}{K^2}.
\]

To solve $H(\delta)=-e^{\mathrm{end}}$, define the Newton map
\[
 \mathcal T(\delta)=-A^{-1}e^{\mathrm{end}}+\delta-A^{-1}H(\delta)
\]
on the cube above.  Its derivative is
$D\mathcal T(\delta)=-A^{-1}(DH(\delta)-A)$, so \eqref{eq:4.14} makes it a contraction with Lipschitz
constant at most $1/2$.  Moreover, \eqref{eq:4.9}, \eqref{eq:4.11}, and \eqref{eq:4.13} give
the required bound at the center of the cube.  Explicitly, since $K+2\le2K$,
set $C_{\mathrm{path},3}\defeq8C_{\mathrm{path},1}$.  Then
\[
 \norm{\mathcal T(0)}_\infty
 =\norm{A^{-1}e^{\mathrm{end}}}_\infty
 \le\frac{4C_{\mathrm{path},1}B^2K\log(K+1)}{\eta n}
 =\frac{C_{\mathrm{path},3}B^2K\log(K+1)}{2\eta n}.
\]
The last quantity is at most $\rho_0/2$ whenever
\[
 n\ge\frac{C_{\mathrm{path},3}}{c_{\mathrm{path},1}}
       \eta^{-1}BK^5\log^2(K+1),
\]
which is one of the fixed numerical requirements absorbed into
\eqref{eq:4.9}.
For every $\norm\delta_\infty\le\rho_0$, it follows that
\[
 \norm{\mathcal T(\delta)}_\infty
 \le\norm{\mathcal T(0)}_\infty+\frac12\norm\delta_\infty
 \le\rho_0.
\]
Hence $\mathcal T$ maps the cube into itself.  The contraction theorem supplies a unique fixed
point $\delta^*$, and the fixed-point identity gives
\begin{equation}
 \norm{\delta^*}_\infty
 \le\frac{C_{\mathrm{path},3}B^2K\log(K+1)}{\eta n}.     \tag{4.15}\label{eq:4.15}
\end{equation}
Indeed, $\delta^*=\mathcal T(\delta^*)$ and the Lipschitz estimate imply
\[
 \norm{\delta^*}_\infty
 \le\norm{\mathcal T(0)}_\infty+\frac12\norm{\delta^*}_\infty.
\]
Therefore $\norm{\delta^*}_\infty\le2\norm{A^{-1}e^{\mathrm{end}}}_\infty$, which yields
\eqref{eq:4.15}.
The fixed-point equation is exactly $H(\delta^*)=-e^{\mathrm{end}}$, and the choice of the cube
keeps every corrected control value in $[0,B]$.

The actual empirical endpoints are now
\begin{align*}
 \widehat\nu_0
 &=\frac1n\sum_{r=1}^{m_0}N_{0,r}\delta_{s_{0,r}}
   +w\sum_{j=1}^K\delta_{t_j},\\
 \widehat\nu_1
 &=\frac1n\sum_{r=1}^{m_1}N_{1,r}\delta_{s_{1,r}}
   +w\sum_{j=1}^K\delta_{t_j+\delta_j^*}.
\end{align*}
For $1\le q\le K$, their $C_q$-moment difference is
$e_q^{\mathrm{end}}+H_q(\delta^*)=0$.  Since
$C_0,C_1,\ldots,C_K$ form a basis for the polynomials of degree at most $K$, the ordinary
moments $\int x^q\,d\widehat\nu_b$ also agree for $0\le q\le K$.

\paragraph{Stage 3: construct one-particle path segments.}
For the path, list the $N_{\mathrm{base}}$ rounded base particles of endpoint $b=0$ as
$a_1^{(0)},\ldots,a_R^{(0)}$ and those of endpoint $b=1$ as
$a_1^{(1)},\ldots,a_R^{(1)}$, where $R=N_{\mathrm{base}}$.  Pair the two lists in any order.
The move $a_i^{(0)}\to a_i^{(1)}$ is assigned the vector
\[
 v_i=\frac1n\bigl(C_q(a_i^{(1)})-C_q(a_i^{(0)})\bigr)_{q=1}^K.
\]
The controls in $\overline\nu_0$ and $\overline\nu_1$ are still identical, so the total change of
the base-particle moments is precisely the endpoint discrepancy already defined in Stage~1:
\[
 S\defeq\sum_{i=1}^Rv_i=e^{\mathrm{end}}.
\]
The vectors $v_i$ do not sum to zero, so apply Steinitz after centering them.  Put
$u_i=v_i-S/R$.  Then $\sum_i u_i=0$.  Also $\norm{v_i}_\infty\le2/n$, while \eqref{eq:4.9}, \eqref{eq:4.11}, and
$R=N_{\mathrm{base}}\ge(1-\eta)n$ give
\[
 \norm{u_i}_\infty
 \le\frac2n+\frac{2(K+2)}{nR}
 \le\frac{C_{\mathrm{path},4}}n.
\]
For example, \eqref{eq:4.9} implies $K\le n$, while $\eta<1/10$ gives $R\ge9n/10$; hence one
may take $C_{\mathrm{path},4}\defeq7$ in the last display.
Apply
Lemma~\ref{ext:steinitz} in the $K$-dimensional space $(\R^K,\norm{\cdot}_\infty)$ to the scaled
vectors $(n/C_{\mathrm{path},4})u_i$.  It supplies a permutation $\pi$ of
$\{1,\ldots,R\}$.  For every $m$,
$\sum_{\ell=1}^mv_{\pi(\ell)}=\sum_{\ell=1}^mu_{\pi(\ell)}+mS/R$.
The Steinitz bound and $0\le m/R\le1$ give, respectively,
\[
 \left\|\sum_{\ell=1}^m u_{\pi(\ell)}\right\|_\infty
 \le\frac{C_{\mathrm{path},4}K}{n},
 \qquad
 \left\|\frac mR S\right\|_\infty
 \le\norm{e^{\mathrm{end}}}_\infty
 \le\frac{2(K+2)}n.
\]
The triangle inequality therefore gives
\begin{equation}
 \left\|\sum_{\ell=1}^m v_{\pi(\ell)}\right\|_\infty
 \le\frac{C_{\mathrm{path},5}K}{n}
 \qquad(0\le m\le R).                                   \tag{4.15a}\label{eq:4.15a}
\end{equation}
Here one may take $C_{\mathrm{path},5}\defeq C_{\mathrm{path},4}+4$, because
$2(K+2)\le4K$ for $K\ge2$.
Thus this ordering prevents the moment error from accumulating to order one while the base
particles are moved.

The $m$th path segment is defined as follows.  Put
\[
 \vartheta(x)=\arccos(2x/B-1)\in[0,\pi].
\]
During local segment time $s\in[0,1]$, move only base particle $\pi(m)$, along
\begin{equation}
 z_m(s)=\frac B2\left[1+\cos\left(
  (1-s)\vartheta(a_{\pi(m)}^{(0)})
  +s\vartheta(a_{\pi(m)}^{(1)})\right)\right].           \tag{4.15b}\label{eq:4.15b}
\end{equation}
All particles earlier in the order $\pi$ are already at their $b=1$ locations, and all later
particles remain at their $b=0$ locations.  Relative to endpoint $b=0$, the resulting base-particle
moment error is the vector $g_m(s)\in\R^K$ with coordinates
\begin{equation}
 g_{m,q}(s)
 =\sum_{\ell=1}^{m-1}v_{\pi(\ell),q}
  +\frac{C_q(z_m(s))-C_q(a_{\pi(m)}^{(0)})}{n},
 \qquad 1\le q\le K.                                    \tag{4.15c}\label{eq:4.15c}
\end{equation}
Equation \eqref{eq:4.15a} and $|C_q|\le1$ imply
\[
 \norm{g_m(s)}_\infty
 \le\frac{C_{\mathrm{path},5}K}{n}+\frac2n
 \le\frac{C_{\mathrm{path},6}K}{n}
\]
uniformly in $m$ and $s$, where
$C_{\mathrm{path},6}\defeq C_{\mathrm{path},5}+1$.

We use the controls to cancel this current error.  Replacing $e^{\mathrm{end}}$ by $g_m(s)$ in the
Stage~2 contraction argument is valid because the preceding bound has the same order as
\eqref{eq:4.11}.  More explicitly, define
\[
 \mathcal T_{m,s}(\delta)
 =-A^{-1}g_m(s)+\delta-A^{-1}H(\delta).
\]
Its derivative has norm at most $1/2$ throughout the same cube by \eqref{eq:4.14}.  Moreover,
\eqref{eq:4.9}, \eqref{eq:4.13}, and
$\norm{g_m(s)}_\infty\le C_{\mathrm{path},6}K/n$ ensure that
$\norm{\mathcal T_{m,s}(0)}_\infty\le\rho_0/2$.  Indeed,
\[
 \norm{\mathcal T_{m,s}(0)}_\infty
 \le\frac{C_{\mathrm{path},1}C_{\mathrm{path},6}
                 B^2K\log(K+1)}{\eta n},
\]
and comparison with the definition of $\rho_0$ again requires only
\[
 n\ge\frac{2C_{\mathrm{path},1}C_{\mathrm{path},6}}
             {c_{\mathrm{path},1}}
       \eta^{-1}BK^5\log^2(K+1),
\]
which is absorbed into
\eqref{eq:4.9}.  Thus the same calculation used for
$\mathcal T$ shows that $\mathcal T_{m,s}$ maps the cube into itself.  It has a unique fixed
point $\delta_m(s)$, and the fixed-point equation is
\begin{equation}
 H(\delta_m(s))=-g_m(s).                                 \tag{4.15d}\label{eq:4.15d}
\end{equation}
The fixed-point estimate also gives
\[
 \norm{\delta_m(s)}_\infty
 \le\frac{C_{\mathrm{path},7}B^2K\log(K+1)}{\eta n}.
\]
Here one may take
$C_{\mathrm{path},7}\defeq2C_{\mathrm{path},1}C_{\mathrm{path},6}$, by the same
$\norm{\delta}\le2\norm{\mathcal T_{m,s}(0)}$ argument used in Stage~2.
The control particles in group $j$ occupy $t_j+\delta_{m,j}(s)$.  At local time $s$, define
\begin{equation}
\begin{aligned}
 \widehat\nu_{m,s}
 &=\frac1n\left[
   \sum_{\ell=1}^{m-1}\delta_{a_{\pi(\ell)}^{(1)}}
   +\delta_{z_m(s)}
   +\sum_{\ell=m+1}^{R}\delta_{a_{\pi(\ell)}^{(0)}}
   \right]\\
 &\quad+w\sum_{j=1}^K\delta_{t_j+\delta_{m,j}(s)}.
\end{aligned}                                            \tag{4.15e}\label{eq:4.15e}
\end{equation}
For every $1\le q\le K$, its $C_q$ moment minus that of $\widehat\nu_0$ is
$g_{m,q}(s)+H_q(\delta_m(s))=0$.  The degree-zero moment is one throughout, and
$C_0,C_1,\ldots,C_K$ form a basis for the polynomials of degree at most $K$.  Hence all ordinary
moments through degree $K$ are constant on this segment.

\paragraph{Stage 4: concatenate the segments.}
For global time $\tau\in[(m-1)/R,m/R]$, set
$s=R\tau-(m-1)$ and $\widehat\nu_\tau=\widehat\nu_{m,s}$.  At a join,
$g_m(1)=g_{m+1}(0)$, so uniqueness in \eqref{eq:4.15d} gives
$\delta_m(1)=\delta_{m+1}(0)$.  Thus the segments join continuously.  Also
$g_1(0)=0$ and $g_R(1)=S=e^{\mathrm{end}}$.  Since $H(0)=0$ and
$H(\delta^*)=-e^{\mathrm{end}}$, uniqueness gives
$\delta_1(0)=0$ and $\delta_R(1)=\delta^*$.  Hence the global endpoints are exactly
$\widehat\nu_0$ and $\widehat\nu_1$ defined in Stage~2.  Labeling the $R$ base particles and the
$KN_c$ control particles consistently in \eqref{eq:4.15e} defines $n$ paths $x_i(\tau)$.  We rename the
global parameter $\tau$ as $t$ to match the lemma statement.

These paths have the stated regularity.  Indeed, along a base move,
\[
\sqrt{z_m(s)}=\sqrt B\cos\left(\frac{(1-s)\vartheta(a_{\pi(m)}^{(0)})
 +s\vartheta(a_{\pi(m)}^{(1)})}{2}\right),
\]
which is continuously differentiable.  The function $g_m(s)$ is continuously differentiable, and
$DH(\delta_m(s))$ is invertible by \eqref{eq:4.14}: indeed,
\[
 DH(\delta)=A\{I+A^{-1}(DH(\delta)-A)\},
\]
and the second factor is invertible by the Neumann series because its distance from $I$ is at
most $1/2$.  The implicit-function theorem applied to \eqref{eq:4.15d}
therefore gives a continuously differentiable solution locally around every $s$; these local
solutions agree on overlaps because the solution in the cube is unique.  Thus $\delta_m(s)$ is
continuously differentiable throughout each segment.  The control locations
stay a positive distance from $0$, so their square roots are also continuously differentiable.
Corners may occur only where consecutive segments are concatenated; this proves the asserted
piecewise continuous differentiability.

\medskip\noindent
\emph{Length estimate.}
By \eqref{eq:4.9a}, the normalized speed of the particle path is
\[
 \left(\frac1n\tr\dot D_t^2\right)^{1/2}
 =\left\{\frac1n\sum_{i=1}^n
       \left(\frac d{dt}\sqrt{x_i(t)}\right)^2\right\}^{1/2}.
\]
Its integral is invariant under the linear reparameterization from local time $s$ to global time
$t$.  Indeed, on segment $m$ we have $s=Rt-(m-1)$ and $ds/dt=R$, while the length of the
global time interval is $1/R$.  Therefore
\[
 \int_{(m-1)/R}^{m/R}
 \left\{\frac1n\sum_i
   \left(\frac d{dt}\sqrt{x_i(t)}\right)^2\right\}^{1/2}dt
 =\int_0^1
 \left\{\frac1n\sum_i
   \left(\frac d{ds}\sqrt{x_i(s)}\right)^2\right\}^{1/2}ds.
\]
We may consequently compute each segment in its local parameter $s$ and add the results.
The estimate follows the chain
\[
 \dot g_m(s)
 \ \xrightarrow{\ H(\delta_m)=-g_m\ }\
 \dot\delta_m(s)
 \ \xrightarrow{\ x=t_j+\delta_{m,j}\ }\
 \frac d{ds}\sqrt{x(s)}
 \ \xrightarrow{\text{sum over particle multiplicities}}\
 v_m(s).
\]
Each arrow is quantified below.

We first relate the velocity of the controls to the velocity of the moment error.  Differentiate
\eqref{eq:4.15d}; dots in the remainder of this calculation denote derivatives with respect to
$s$.  The ordinary multivariable chain rule gives
\[
 DH(\delta_m(s))\dot\delta_m(s)=-\dot g_m(s),
 \qquad
 \dot\delta_m(s)=-DH(\delta_m(s))^{-1}\dot g_m(s).
\]
Writing
\[
 DH(\delta)=A\{I+A^{-1}(DH(\delta)-A)\}
\]
and using \eqref{eq:4.14}, the Neumann-series bound
$\norm{(I+E)^{-1}}_{\infty\to\infty}\le(1-\norm E_{\infty\to\infty})^{-1}$ gives
\[
 \norm{DH(\delta_m(s))^{-1}}_{\infty\to\infty}
 \le2\norm{A^{-1}}_{\infty\to\infty}.
\]
Combining this with \eqref{eq:4.13} gives the first inequality below:
\begin{equation}
 \norm{\dot\delta_m(s)}_\infty
 \le\frac{C_{\mathrm{path},8}B^2\log(K+1)}\eta
       \norm{\dot g_m(s)}_\infty,\qquad
 \int_0^1\norm{\dot g_m(s)}_\infty ds
 \le\frac{C_{\mathrm{path},9}K}{n}.                      \tag{4.16}\label{eq:4.16}
\end{equation}
Here $C_{\mathrm{path},8}\defeq2C_{\mathrm{path},1}$.
For the last inequality, define the linearly interpolated angle
\[
 \alpha_m(s)=(1-s)\vartheta(a_{\pi(m)}^{(0)})
              +s\vartheta(a_{\pi(m)}^{(1)}).
\]
Then \eqref{eq:4.15b}--\eqref{eq:4.15c} give
$C_q(z_m(s))=\cos(q\alpha_m(s))$ and, for every $1\le q\le K$,
\[
 \dot g_{m,q}(s)
 =-\frac qn\sin(q\alpha_m(s))\dot\alpha_m(s).
\]
Consequently,
\[
 \norm{\dot g_m(s)}_\infty
 \le\frac Kn\abs{\dot\alpha_m(s)}.
\]
The function $\alpha_m$ is affine, hence monotone, and both endpoint angles lie in $[0,\pi]$.
It follows that
\[
 \int_0^1\norm{\dot g_m(s)}_\infty ds
 \le\frac Kn\int_0^1\abs{\dot\alpha_m(s)}ds
 =\frac Kn\abs{\alpha_m(1)-\alpha_m(0)}
 \le\frac{\pi K}{n},
\]
which is the second inequality in \eqref{eq:4.16}, with
$C_{\mathrm{path},9}\defeq\pi$.

We now convert these control velocities into the square-root particle velocities appearing in
the length.  On segment $m$, the only moving entries are the single base particle at $z_m(s)$
and the $N_c$ identical particles in each control group.  Since $w=N_c/n$, their exact normalized
speed is
\[
 v_m(s)=\left\{
  \frac1n\left(\frac d{ds}\sqrt{z_m(s)}\right)^2
  +w\sum_{j=1}^K
    \left(\frac d{ds}\sqrt{t_j+\delta_{m,j}(s)}\right)^2
 \right\}^{1/2}.
\]
The inequality $\sqrt{u+v}\le\sqrt u+\sqrt v$ separates this into a base contribution and a
control contribution.

For the base particle, \eqref{eq:4.15b} gives
\[
 \sqrt{z_m(s)}=\sqrt B\cos\frac{\alpha_m(s)}2,
 \qquad
 \frac d{ds}\sqrt{z_m(s)}
 =-\frac{\sqrt B}{2}\sin\frac{\alpha_m(s)}2\,\dot\alpha_m(s).
\]
Because $\alpha_m$ is monotone in $[0,\pi]$, $\sqrt{z_m(s)}$ is also monotone.  Hence its total
variation is the absolute difference of its endpoint values and is at most $\sqrt B$:
\[
 \int_0^1\frac1{\sqrt n}
   \left|\frac d{ds}\sqrt{z_m(s)}\right|ds
 =\frac1{\sqrt n}
   \abs{\sqrt{a_{\pi(m)}^{(1)}}-\sqrt{a_{\pi(m)}^{(0)}}}
 \le\frac{\sqrt B}{\sqrt n}.
\]
Since $R\le n$, summing over all base moves gives
\[
 \frac{R\sqrt B}{\sqrt n}\le\sqrt{nB}.                    \tag{4.16b}\label{eq:4.16b}
\]

For the controls, the ordinary one-variable chain rule expresses each square-root velocity in
terms of $\dot\delta_m$:
\[
 \frac d{ds}\sqrt{t_j+\delta_{m,j}(s)}
 =\frac{\dot\delta_{m,j}(s)}{2\sqrt{t_j+\delta_{m,j}(s)}}.
\]
Every control location is at least $c_{\mathrm{path},2}B/K^2$.  Therefore the control part of
$v_m(s)$ satisfies
\[
\begin{aligned}
 v_m^{\mathrm{ctrl}}(s)
 &\defeq\left\{w\sum_{j=1}^K
 \left(\frac{\dot\delta_{m,j}(s)}
 {2\sqrt{t_j+\delta_{m,j}(s)}}\right)^2\right\}^{1/2}\\
 &\qquad\le
 \frac{\sqrt w K}{2\sqrt{c_{\mathrm{path},2}B}}
 \left(\sum_{j=1}^K\dot\delta_{m,j}(s)^2\right)^{1/2}\\
 &\qquad\le C_{\mathrm{path},10}\sqrt w\frac{K^{3/2}}{\sqrt B}
       \norm{\dot\delta_m(s)}_\infty.
\end{aligned}
\]
Here the first inequality uses the lower bound on the control locations, and the second uses
$\norm{u}_2\le\sqrt K\norm{u}_\infty$; here
$C_{\mathrm{path},10}\defeq1/(2\sqrt{c_{\mathrm{path},2}})$.  Integrating, then substituting successively
$w\le\eta/(BK)$ and \eqref{eq:4.16}, gives
\begin{align*}
 \int_0^1v_m^{\mathrm{ctrl}}(s)\,ds
 &\le C_{\mathrm{path},10}\sqrt w\frac{K^{3/2}}{\sqrt B}
       \frac{C_{\mathrm{path},8}B^2\log(K+1)}\eta
       \int_0^1\norm{\dot g_m(s)}_\infty ds\\
 &\le
 \frac{C_{\mathrm{path},11}BK^2\log(K+1)}{n\sqrt\eta}.
\end{align*}
Here $C_{\mathrm{path},11}\defeq
C_{\mathrm{path},10}C_{\mathrm{path},8}C_{\mathrm{path},9}$.
For clarity, the powers in the last line come from
\[
 \frac{\sqrt\eta}{\sqrt B\sqrt K}\,
 \frac{K^{3/2}}{\sqrt B}\,
 \frac{B^2\log(K+1)}\eta\,
 \frac Kn
 =\frac{BK^2\log(K+1)}{n\sqrt\eta}.
\]
There are $R\le n$ segments, so the total control contribution is at most
$C_{\mathrm{path},11}BK^2\log(K+1)/\sqrt\eta$.  Together with \eqref{eq:4.16b}, this proves \eqref{eq:4.10}.

\medskip\noindent
\emph{Endpoint approximation.}
It remains to establish the endpoint estimate \eqref{eq:4.10a}.  For the preliminary endpoint law
$\overline\nu_b$ from Stage~1,
\eqref{eq:4.10b} gives
\[
 h(\overline\nu_b)
 =\frac{N_{\mathrm{base}}}{n}h(\nu_b)
   +w\sum_{j=1}^K\sqrt{t_j}+r_b,
 \qquad
 r_b=\frac1n\sum_{r=1}^{m_b}\Delta_{b,r}\sqrt{s_{b,r}},
 \qquad |r_b|\le\frac{2KB}{n}.
\]
Here we used $m_b\le K+2\le2K$, $\sqrt{x}\le\sqrt B\le B$, and the fact that each rounded
count has error below one.  Since $N_{\mathrm{base}}/n=1-Kw$, subtracting $h(\nu_b)$ from
the preceding identity gives
\[
 h(\overline\nu_b)-h(\nu_b)
 =-Kw\,h(\nu_b)+w\sum_{j=1}^K\sqrt{t_j}+r_b.
\]
Now $h(\nu_b)\le\sqrt B$, every $t_j\le B$, and $Kw\le\eta/B$.  Hence
\[
 Kw\,h(\nu_b)\le\eta,
 \qquad
 w\sum_{j=1}^K\sqrt{t_j}\le Kw\sqrt B\le\eta,
\]
where $B\ge1$ was used in both estimates.  Together with the bound on $r_b$, this proves
\begin{equation}
 |h(\overline\nu_b)-h(\nu_b)|
 \le2\eta+\frac{2KB}{n}.                                 \tag{4.16a}\label{eq:4.16a}
\end{equation}
The controls for endpoint $b=0$ do not move.  For endpoint $b=1$, \eqref{eq:4.15} and
$\abs{\sqrt{x+u}-\sqrt x}=\abs u/(\sqrt{x+u}+\sqrt x)$ show that the control displacement
changes the square-root moment by at most
\[
 w\sum_{j=1}^K\abs{\sqrt{t_j+\delta_j^*}-\sqrt{t_j}}
 \le \frac{wK^2}{2\sqrt{c_{\mathrm{path},2}B}}\norm{\delta^*}_\infty
 \le\frac{C_{\mathrm{path},12}BK^2\log(K+1)}n.
\]
Indeed, both $t_j$ and $t_j+\delta_j^*$ are at least
$c_{\mathrm{path},2}B/K^2$, so the denominator in the square-root identity is at least
$2\sqrt{c_{\mathrm{path},2}B}/K$.  Summing over $K$ groups gives the first inequality.  For the
second, substituting $w\le\eta/(BK)$ and \eqref{eq:4.15} gives
\[
 \frac{C_{\mathrm{path},3}}{2\sqrt{c_{\mathrm{path},2}}}
 \frac{\eta}{BK}\frac{K^2}{\sqrt B}
   \frac{B^2K\log(K+1)}{\eta n}
 =C_{\mathrm{path},12}\frac{\sqrt B\,K^2\log(K+1)}n
 \le C_{\mathrm{path},12}\frac{BK^2\log(K+1)}n,
\]
where $C_{\mathrm{path},12}\defeq
C_{\mathrm{path},3}/(2\sqrt{c_{\mathrm{path},2}})$ and the last inequality uses $B\ge1$.
Combining this with \eqref{eq:4.16a} proves \eqref{eq:4.10a}.

To unify the constants, let $C_{\mathrm{path},13}$ be the maximum of the finitely many fixed
numerical thresholds required of the coefficient in \eqref{eq:4.9}: those used to ensure
$\eta n/(BK)\ge2$ and to make the Stage~2 and Stage~3 Newton maps send the cube into itself.
This is an absolute constant because every local $C_{\mathrm{path},i}$ and
$c_{\mathrm{path},i}$ has already been fixed.  Finally, define
\[
 C_{\mathrm{path}}
 \defeq2\max\{1,C_{\mathrm{path},11},C_{\mathrm{path},12}+1,
                   C_{\mathrm{path},13}\}.
\]
Then \eqref{eq:4.9} implies every size condition used in the construction, while \eqref{eq:4.16b} and the control
estimate give \eqref{eq:4.10}, and \eqref{eq:4.16a} and the final control-displacement estimate give \eqref{eq:4.10a}.
In the last implication we also use
$2KB/n\le2BK^2\log(K+1)/n$ for $K\ge2$; the definition of
$C_{\mathrm{path}}$ absorbs this coefficient together with $C_{\mathrm{path},12}$.
\end{proof}

\subsection{Extension to every fixed non-even exponent}\label{app:noneven}

\begin{proof}[Proof of Corollary~\ref{cor:noneven}]
Fix $p>0$ with $p\notin\{2,4,6,\ldots\}$ and fix $0<\eps<1$.  Put
\[
 \alpha=\frac p2\notin\mathbb N,
 \qquad j_\alpha=\lceil\alpha\rceil,
 \qquad \rho_\eps=\frac{1-\eps}{3+\eps},
 \qquad \eta_\eps=\frac{1-\eps}{4}.
\]
We reuse the moment-matched scalar construction, the equal-weight path, and the observation
comparison from the nuclear-norm lower bound.  The two additional points are to separate an
arbitrary noninteger moment $x^\alpha$ and, when $p>2$, to control the Gaussian perturbation by
its operator norm rather than by its Frobenius norm.

\paragraph{Noninteger-moment separation.}
Repeat the construction in the proof of Proposition~\ref{prop:scalar-pair}, initially leaving
the integer $M_\star\ge6$ unspecified and taking $K\ge\max\{2,j_\alpha\}$.  Recall the nodes
$x_j=j^{M_\star}$ and the relative absolute weights
$R_j=|w_j|/|w_0|$.  The contribution of node $j\ge1$ to the unnormalized $\alpha$ moment, in
units of $|w_0|$, is
\[
 Q_j=x_j^\alpha R_j=j^{M_\star\alpha}R_j.
\]
Define
\[
 q_j=(j-1)!\,j^{\alpha-j+1}\qquad(j\ge1).
\]
The factorization \eqref{eq:4.5b} and the estimate leading to \eqref{eq:4.6} give
\begin{equation}
 q_j^{M_\star}\le Q_j
 \le \frac{j^{j-1}}{(j-1)!}
       \exp\!\left(\frac{\pi^2j}{3M_\star}\right)q_j^{M_\star}
 \qquad(j\ge2),                                                        \tag{A.p1}\label{eq:A.p1}
\end{equation}
while \eqref{eq:4.4} gives the analogous bound $Q_1\le Cq_1^{M_\star}$ with an absolute
constant.  For the lower bound in \eqref{eq:A.p1}, use
$j^{M_\star}-\ell^{M_\star}\le j^{M_\star}$ for $\ell<j$ and note that all factors in
\eqref{eq:4.5b} with $\ell>j$ are at least one.  The upper bound is \eqref{eq:4.6} multiplied by
$j^{M_\star(\alpha-1)}$ and then rearranged.

The decisive calculation is
\begin{equation}
 \frac{q_{j+1}}{q_j}
 =\left(1+\frac1j\right)^{\alpha-j}.                                  \tag{A.p2}\label{eq:A.p2}
\end{equation}
Because $\alpha$ is not an integer, $q_j$ has the unique maximum $q_{j_\alpha}$.  This maximum
dominates uniformly over the number of nodes.  To see this explicitly, the indices
$j<j_\alpha$ form a fixed finite set, so \eqref{eq:A.p1}, together with the preceding bound for
$Q_1$, shows that their contribution is exponentially smaller than
$q_{j_\alpha}^{M_\star}$.  For the upper tail, put
$J_\alpha=\max\{2,j_\alpha+1\}$ and let $E_j$ denote the right side of \eqref{eq:A.p1} for
$j\ge J_\alpha$.  Consecutive terms obey
\[
 \frac{E_{j+1}}{E_j}
 \le e\exp\!\left(\frac{\pi^2}{3M_\star}\right)
       \left(1+\frac1j\right)^{-M_\star(j-\alpha)}.
\]
The infimum of $(j-\alpha)\log(1+1/j)$ over integers $j\ge J_\alpha$ is positive.  Hence, after
increasing $M_\star$ by an amount depending only on $\alpha$, the last ratio is at most $1/2$;
the entire upper tail is then a geometric series.  Increasing $M_\star$ once more makes its first
term, and each of the finitely many lower-index contributions, at most the required fraction of
$Q_{j_\alpha}$.  It follows that, for every $\tau>0$, there is
a finite integer $M(\alpha,\tau)\ge6$, independent of $K$ and $n$, such that
\begin{equation}
 \sum_{\substack{1\le j\le K+1\\j\ne j_\alpha}}Q_j
 \le \tau Q_{j_\alpha}
 \qquad\text{for every }K\ge\max\{2,j_\alpha\}.                         \tag{A.p3}\label{eq:A.p3}
\end{equation}

Take $\tau=\rho_\eps^p$ and fix
$M_{p,\eps}=M(\alpha,\rho_\eps^p)$ large enough also to meet the common-rescaling requirements
following \eqref{eq:4.7a}.  Normalize the two sign classes and rescale them to common first
moment one exactly as in the proof of Proposition~\ref{prop:scalar-pair}.  These operations
multiply both unnormalized $\alpha$ moments by the same positive factor.  Label the class
containing $j_\alpha$ by $1$ and the other class by $0$, and write
\[
 H_b=\int x^\alpha\,d\nu_b(x),\qquad s_b=H_b^{1/p}.
\]
Equations \eqref{eq:4.7a}--\eqref{eq:4.7c} keep the common normalization and rescaling factors
bounded above and below.  Equation \eqref{eq:A.p1} when $j_\alpha\ge2$, and
\eqref{eq:4.4} when $j_\alpha=1$, give
$Q_{j_\alpha}\ge q_{j_\alpha}^{M_{p,\eps}}>0$.
Thus \eqref{eq:A.p3} implies
\begin{equation}
 \frac{H_0}{H_1}\le\rho_\eps^p,\qquad
 H_1\ge c_{p,\eps}^p>0,\qquad
 \frac{s_1-s_0}{s_1+s_0}
 \ge\frac{1-\rho_\eps}{1+\rho_\eps}
 =\frac{1+\eps}{2}=\eps+2\eta_\eps.                                  \tag{A.p4}\label{eq:A.p4}
\end{equation}
The two laws have matching moments through degree $K$ and support in
$[0,2(K+1)^{M_{p,\eps}}]$.  At a positive even $p$, by contrast, $\alpha=p/2$ is an integer and
$H_0=H_1$ is one of the exactly matched moments once $K\ge\alpha$.  The constants above become
nonuniform as $p$ approaches any positive even integer, exactly when the unique maximum in
\eqref{eq:A.p2} degenerates into a tie.

One could instead amplify any fixed ratio $H_1/H_0>1$ by taking a fixed number of independent
products: integer moments would remain matched and the $\alpha$-moment ratio would be raised to
that number of copies.  That route produces polynomially many support atoms and would require a
many-atom version of Lemma~\ref{lem:path}.  Increasing the fixed exponent $M_{p,\eps}$ as above
achieves the same amplification while retaining the present $K+2$-node construction.

\paragraph{The equal-weight endpoints retain the gap.}
Put
\[
 K=\left\lfloor c_{\mathrm{low}}\frac{\log n}{\log\log n}\right\rfloor,
 \qquad B=2(K+2)^{M_{p,\eps}},
 \qquad
 \eta_{\mathrm{path}}(n)
 =\min\left\{\frac1{20},\frac1{B^\alpha\log(en)}\right\}.
\]
For all sufficiently large $n$, we have $K\ge\max\{2,j_\alpha\}$ and the size condition
\eqref{eq:4.9} holds.  Indeed, $B$, $\eta_{\mathrm{path}}^{-1}$, and $K$ are all bounded by
fixed powers of $\log n$.  In addition to the conclusions of Lemma~\ref{lem:path}, its endpoints
satisfy
\begin{equation}
 \left|\int x^\alpha\,d\widehat\nu_b(x)-H_b\right|
 \le C_\alpha\left\{
   \eta_{\mathrm{path}}B^{(\alpha-1)_+}
   +\frac{B^\alpha K^3\log(K+1)}n\right\}=o(1)
 \qquad(b\in\{0,1\}),                                                       \tag{A.p5}\label{eq:A.p5}
\end{equation}
where $(u)_+=\max\{u,0\}$.  Here is the verification.  In the preliminary endpoint calculation
leading to \eqref{eq:4.16a}, use $x^\alpha\le B^\alpha$ on $[0,B]$.  The reserved mass and
rounding errors are then at most
$2\eta_{\mathrm{path}}B^{(\alpha-1)_+}+2KB^\alpha/n$.
The endpoint-zero controls do not move.  At endpoint one, the mean-value theorem and
\eqref{eq:4.15} give
\begin{align*}
 &w\sum_{j=1}^K
   \left|(t_j+\delta_j^*)^\alpha-t_j^\alpha\right|\\
 &\quad\le C_\alpha\frac{B^\alpha
          K^{1+2(1-\alpha)_+}\log(K+1)}n
 \le C_\alpha\frac{B^\alpha K^3\log(K+1)}n.
\end{align*}
For $0<\alpha<1$, use the lower bound
$t_j,t_j+\delta_j^*\ge c_{\mathrm{path},2}B/K^2$; for $\alpha\ge1$, use the upper bound $B$
instead.  This proves \eqref{eq:A.p5}.

Let
\[
 \widehat H_b=\frac1n\sum_{i=1}^nx_i(b)^\alpha,\qquad
 \widehat s_b=\widehat H_b^{1/p}.
\]
Equations \eqref{eq:A.p4}--\eqref{eq:A.p5} and continuity of $u\mapsto u^{1/p}$ show that,
for all sufficiently large $n$,
\begin{equation}
 \frac{\widehat s_1-\widehat s_0}{\widehat s_1+\widehat s_0}
 \ge\eps+\eta_\eps,
 \qquad \widehat s_1\ge c_{p,\eps}>0.                                      \tag{A.p6}\label{eq:A.p6}
\end{equation}

Take independent Haar orthogonal $U,V$ and set $F_b=\sqrt n\,UD_bV^T$.  Since the singular
values of $F_b$ are $\sqrt{nx_i(b)}$,
\begin{equation}
 \norm{F_b}_{\Sp{p}}
 =n^{1/2+1/p}\widehat s_b.                                                    \tag{A.p7}\label{eq:A.p7}
\end{equation}
The first $K$ integer moments of the squared singular values remain exactly constant along the
path.  Also, with $m=2K+1$, the verification following \eqref{eq:4.17} gives
$(4m)^{4m}\le n/2$ for all sufficiently large $n$.  Thus every hypothesis of
Lemma~\ref{lem:path-tv} holds.

\paragraph{Gaussian smoothing for the norm and quasi-norm ranges.}
Let $Z$ have independent $N(0,1)$ entries and define
\[
 X_b=aF_b+\sigma Z,\qquad b\in\{0,1\},
\]
with $a=1$ and with the fixed positive constant $\sigma=\sigma(p,\eps)$ chosen below.  Fix
$\zeta=1/100$ and put $\Lambda_p=n^{1/2+1/p}$.  There is a finite constant
$G_{p,\zeta}$, independent of $n$, such that
\begin{equation}
 \Prb\{\norm Z_{\Sp{p}}>G_{p,\zeta}\Lambda_p\}\le\zeta.                    \tag{A.p8}\label{eq:A.p8}
\end{equation}
For $0<p<2$, this follows from
$\norm Z_{\Sp{p}}\le n^{1/p-1/2}\norm Z_F$ and
$\Prb\{\norm Z_F>n/\sqrt\zeta\}\le\zeta$.
For $p\ge2$, apply the Gaussian covariance bound in Theorem~\ref{ext:covariance} to
$G=Z/\sqrt n$ with $D=I_n$ and fixed failure probability $\zeta$.  Its first estimate gives
$\norm Z_{\op}\le C_\zeta\sqrt n$, and hence
$\norm Z_{\Sp{p}}\le n^{1/p}\norm Z_{\op}\le C_\zeta\Lambda_p$.

First suppose $p\ge1$.  On the event complementary to \eqref{eq:A.p8}, put
$r_{\mathrm{sm}}=\sigma G_{p,\zeta}\Lambda_p$.  The triangle and reverse-triangle inequalities
place $\norm{X_b}_{\Sp{p}}$ in
\[
 I_b=\left[\max\{0,\Lambda_p\widehat s_b-r_{\mathrm{sm}}\},
                 \Lambda_p\widehat s_b+r_{\mathrm{sm}}\right].
\]
By \eqref{eq:A.p6},
\[
 \begin{aligned}
 (1-\eps)\inf I_1-(1+\eps)\sup I_0
 &\ge \Lambda_p\left\{
   (\widehat s_1-\widehat s_0)-\eps(\widehat s_1+\widehat s_0)
   -2\sigma G_{p,\zeta}\right\}\\
 &\ge \Lambda_p\left\{\eta_\eps(\widehat s_1+\widehat s_0)
   -2\sigma G_{p,\zeta}\right\}.
 \end{aligned}
\]
Choosing, for example,
$\sigma<\eta_\eps c_{p,\eps}/(4G_{p,\zeta})$ makes this quantity positive.  Thus the
$(1\pm\eps)$ enlargements of $I_0$ and $I_1$ are disjoint for every non-even $p\ge1$, including
the range above two.

Now suppose $0<p<1$.  The Schatten-$p$ quasi-norm satisfies the McCarthy $p$-triangle
inequality~\cite{McCarthyCp}
\[
 \norm{A+B}_{\Sp{p}}^p\le\norm A_{\Sp{p}}^p+\norm B_{\Sp{p}}^p.
\]
Applying it twice gives
\begin{equation}
 \left|\norm A_{\Sp{p}}^p-\norm B_{\Sp{p}}^p\right|
 \le\norm{A-B}_{\Sp{p}}^p.                                                   \tag{A.p9}\label{eq:A.p9}
\end{equation}
On the event in \eqref{eq:A.p8},
$\norm{\sigma Z}_{\Sp{p}}^p\le
 \sigma^pG_{p,\zeta}^p n^{1+p/2}$.
Set
\[
 \theta_{p,\eps}=\frac{1-\eps-\eta_\eps}{1+\eps+\eta_\eps}
 <\frac{1-\eps}{1+\eps}.
\]
Equation \eqref{eq:A.p6} implies
$\widehat H_0\le\theta_{p,\eps}^p\widehat H_1$ and
$\widehat H_1\ge c_{p,\eps}^p$.  Hence
\[
 (1-\eps)^p\widehat H_1-(1+\eps)^p\widehat H_0
 \ge \Delta_{p,\eps}>0,
\]
where one may take
\[
 \Delta_{p,\eps}=c_{p,\eps}^p
 \left\{(1-\eps)^p-(1+\eps)^p\theta_{p,\eps}^p\right\}.
\]
Choose the fixed $\sigma>0$ small enough that
\[
 R_p\defeq\sigma^pG_{p,\zeta}^p
 <\frac{\Delta_{p,\eps}}{(1+\eps)^p+(1-\eps)^p}.
\]
Then
\begin{equation}
 (1+\eps)^p\{\widehat H_0+R_p\}
 <(1-\eps)^p\{\widehat H_1-R_p\}.                         \tag{A.p10}\label{eq:A.p10}
\end{equation}
Equations \eqref{eq:A.p7}--\eqref{eq:A.p9} place $\norm{X_b}_{\Sp{p}}$ in
\[
 I_b=\Lambda_p
 \left[\max\{0,\widehat H_b-R_p\}^{1/p},
       (\widehat H_b+R_p)^{1/p}\right]
\]
with probability at least $1-\zeta$.  Raising the desired strict separation
$(1+\eps)\sup I_0<(1-\eps)\inf I_1$ to the positive power $p$ gives exactly
\eqref{eq:A.p10}.  Thus the required separation survives smoothing for every exponent in the
corollary.

\paragraph{Observation comparison and testing.}
The chosen ratio $a/\sigma$ and the support exponent $M_{p,\eps}$ are fixed independently of
$n$.  Moreover, \eqref{eq:4.10} and the present value of $\eta_{\mathrm{path}}$ give
$\mathcal L_K^2\le C_{p,\eps}nB$ for all sufficiently large $n$: the other term is bounded by a
fixed power of $\log n$.  Consequently, Lemma~\ref{lem:path-tv} and the calculation
\eqref{eq:4.23a}--\eqref{eq:4.24} give, for every row coisometry $L$ of rank at most $k$,
\begin{equation}
 \TV\bigl(\operatorname{Law}(LX_0),\operatorname{Law}(LX_1)\bigr)^2
 \le n\exp\{C_{p,\eps}K\log(K+1)\}\frac{k^{K+1}}{n^{2K}}                 \tag{A.p11}\label{eq:A.p11}
\end{equation}
with a finite constant $C_{p,\eps}$.  The optimization in
\eqref{eq:4.25}--\eqref{eq:4.26}, with $C_{\mathrm{low},1}(\eps)$ replaced by
$C_{p,\eps}$, furnishes a finite $A_{p,\eps}$ for which \eqref{eq:A.p11} is at most $1/16$
whenever
\[
 k\le\left\lfloor\frac{n^2}{(\log n)^{A_{p,\eps}}}\right\rfloor
\]
and $n$ is sufficiently large.

Finally, apply the functional form of Lemma~\ref{lem:yao} with
\[
 \Phi_p(A)=\norm{A}_{\Sp{p}}.
\]
The deterministic intervals constructed above satisfy its separation hypothesis and, by
\eqref{eq:A.p8}, contain $\Phi_p(X_b)$ with probability at least $1-\zeta$ under label $b$.
Equation \eqref{eq:A.p11} gives the lemma's observation hypothesis with $\delta=1/4$ for every
row coisometry of the stated rank.  Since $\zeta=1/100$ and $1/4+2/100<1/3$, the lemma rules out
every such Schatten-$p$ sketch and proves \eqref{eq:1.2a}.
\end{proof}

\section{Tensor and random-rotation verification}\label{app:tensors}

This appendix proves Lemma~\ref{lem:twirl}.  Its only external input is the exact orthogonal
Weingarten formula of Collins and P.~\'{S}niady~\cite[Corollary~3.4 and Lemma~3.5]{CollinsSniady}.  We derive the
coefficient bound needed from that formula's pairing Gram matrix and prove all trace
decompositions, contraction estimates, and block recombinations directly.

The constants $C_{\mathrm{ten},1},C_{\mathrm{ten},2},\ldots$ below are fixed absolute
constants attached to individual estimates.  At the end of the appendix they are combined into
the single envelope $C_{\mathrm{ten}}$ used in \eqref{eq:4.22}.

\subsection{Tensors, traces, and normalized metric insertions}\label{app:trace-decomposition}

Put $V=\R^n$ with standard basis $e_1,\ldots,e_n$.  Throughout this appendix, $m$ is the positive
integer specifying the tensor order.  An order-$m$ tensor is an array
$T_{i_1\cdots i_m}$ with the squared Euclidean norm
\[
 \norm T^2=\sum_{i_1,\ldots,i_m=1}^n T_{i_1\cdots i_m}^2.
\]
An orthogonal matrix acts in every index:
\[
 (U^{\otimes m}T)_{i_1\cdots i_m}
 =\sum_{a_1,\ldots,a_m}U_{i_1a_1}\cdots U_{i_ma_m}T_{a_1\cdots a_m}.
\]
This action preserves the norm.  The metric tensor and its normalized version are
\begin{equation}
 g=\sum_{i=1}^ne_i\otimes e_i,\qquad \bar g=n^{-1/2}g,\qquad \norm{\bar g}=1. \tag{B.1}\label{eq:B.1}
\end{equation}
Thus $\bar g$ is the normalized identity tensor; in coordinates,
\[
 \bar g_{ij}=n^{-1/2}\mathbf 1_{\{i=j\}}.
\]
For every $Q\in\mathcal O_n$, orthogonality gives
\begin{align}
 (Q\otimes Q)\bar g
 &=n^{-1/2}\sum_{i=1}^n(Qe_i)\otimes(Qe_i)\notag\\
 &=n^{-1/2}\sum_{a,b=1}^n
      \left(\sum_{i=1}^nQ_{ai}Q_{bi}\right)e_a\otimes e_b\notag\\
 &=n^{-1/2}\sum_{a,b=1}^n(QQ^T)_{ab}e_a\otimes e_b
 =\bar g.                                                \tag{B.1a}\label{eq:B.metric-invariance}
\end{align}
Thus a normalized identity tensor inserted in two positions is unchanged when the same
orthogonal transformation is applied in both positions.
We next define separately the operation that inserts this tensor and the adjoint operation that
removes it.

Fix an output order $d\ge2$ and two output positions $1\le a<b\le d$.  For an order-$(d-2)$
tensor $h$, write $j_{\widehat{a,b}}$ for the tuple obtained from
$(j_1,\ldots,j_d)$ by deleting $j_a,j_b$.  The \emph{normalized metric insertion}
is the linear map defined by
\[
\begin{gathered}
 I_{ab}^{(d)}:V^{\otimes(d-2)}\longrightarrow V^{\otimes d},\\
 \bigl(I_{ab}^{(d)}h\bigr)_{j_1\cdots j_d}
 =n^{-1/2}\mathbf 1_{\{j_a=j_b\}}h_{j_{\widehat{a,b}}}.
\end{gathered}
\]
In other words, $I_{ab}^{(d)}$ places one copy of $\bar g$ in positions $a,b$ and places the
indices of $h$, in their original order, in the remaining positions.  Such an insertion is also
called a \emph{cup}.

The adjoint of $I_{ab}^{(d)}$ for the tensor Euclidean inner products is the
\emph{normalized contraction}
\[
 C_{ab}^{(d)}=(I_{ab}^{(d)})^*:V^{\otimes d}\longrightarrow V^{\otimes(d-2)}.
\]
It sets the two selected indices equal, sums over their common value, and multiplies by
$n^{-1/2}$:
\begin{equation*}
 \bigl(C_{ab}^{(d)}T\bigr)_{j_{\widehat{a,b}}}
 =n^{-1/2}\sum_{\ell=1}^n
 T_{j_1\cdots j_{a-1}\,\ell\,j_{a+1}\cdots
   j_{b-1}\,\ell\,j_{b+1}\cdots j_d}.
\end{equation*}
The unnormalized trace over positions $a,b$ is $\sqrt n\,C_{ab}^{(d)}T$, so it has the same
kernel as the normalized contraction.

These contraction maps define the trace-free subspace.  An order-$p$ tensor $h$ is
\emph{trace-free} if
\[
 C_{ab}^{(p)}h=0\qquad\text{for every }1\le a<b\le p.
\]
For $p=0$ or $p=1$, there is no pair of positions, so every tensor is trace-free by convention.

A partial matching $\pi$ of $\{1,\ldots,m\}$ is a collection of disjoint unordered pairs.
If $\pi$ has $r$ pairs, then $p=m-2r$ positions remain unmatched.  Given a trace-free
$p$-tensor $h$, let $J_\pi h$ be the order-$m$ tensor obtained by placing one copy of $\bar g$ in
every matched pair and placing the indices of $h$, in their original order, into the unmatched
positions listed in increasing order.  Equivalently, $J_\pi$ performs one insertion of the form
above for each pair of $\pi$.  For example, when
$m=3$ and $\pi=\{\{1,3\}\}$, one has $J_\pi=I_{13}^{(3)}$ and
\[
 (J_{\{\{1,3\}\}}h)_{i_1i_2i_3}
 =n^{-1/2}\mathbf 1_{\{i_1=i_3\}}h_{i_2}.
\]
For a three-tensor $T$, its adjoint is the corresponding contraction:
\[
 (J_{\{\{1,3\}\}}^*T)_{i_2}=n^{-1/2}\sum_{i=1}^nT_{i,i_2,i}.
\]
The number of matchings with exactly $r$ pairs is
\begin{equation}
 N_{m,r}=\frac{m!}{2^r r!(m-2r)!}.                         \tag{B.2}\label{eq:B.2}
\end{equation}

The elementary trace decomposition says that $V^{\otimes m}$ is spanned by tensors $J_\pi h$
with $h$ trace-free.  A direct linear-algebra proof follows.  The joint trace-free space is
$\bigcap_{a<b}\ker C_{ab}^{(m)}$.  In finite dimensions,
\[
 \left(\bigcap_{a<b}\ker C_{ab}^{(m)}\right)^\perp
 =\sum_{a<b}\operatorname{range}\bigl((C_{ab}^{(m)})^*\bigr)
 =\sum_{a<b}\operatorname{range}(I_{ab}^{(m)}).
\]
Thus every tensor is a sum of a trace-free tensor and tensors containing one cup.  Apply the same
identity to the order-$(m-2)$ tensor behind each cup and continue by induction.  The resulting
terms are precisely $J_\pi h$ with $h$ trace-free and $\pi$ a partial matching.  Different
matchings can give dependent spanning blocks, necessitating the Gram estimate below.

\subsection{Stable trace-free pairing frame}

We use the elementary trace decomposition above over the complex numbers.  Put
$V_{\mathbb C}=\mathbb C^n$, with its standard Hermitian inner product.  We use the same notation
$C_{ab}^{(p)}$ for the complex-linear extension of the real contraction defined above; it sums
equal coordinates without conjugating them.  Define
\begin{equation}
 \mathcal H_p
 =\bigcap_{1\le a<b\le p}\ker C_{ab}^{(p)}
 \subseteq V_{\mathbb C}^{\otimes p}.                    \tag{B.2a}\label{eq:B.2a}
\end{equation}
Thus $\mathcal H_p$ consists of the trace-free $p$-tensors.  We use the conventions
$\mathcal H_0=\mathbb C$ and $V_{\mathbb C}^{\otimes0}=\mathbb C$.  Orthogonal transformations
preserve this space because
\[
 C_{ab}^{(p)}U^{\otimes p}
 =U^{\otimes(p-2)}C_{ab}^{(p)}.
\]

For a partial matching $\pi$ with $r$ pairs, set $p=m-2r$.  The insertion formula defining
$J_\pi$ extends complex-linearly to an isometry
\[
 \widetilde J_\pi:V_{\mathbb C}^{\otimes p}
 \longrightarrow V_{\mathbb C}^{\otimes m}.
\]
We write $J_\pi=\widetilde J_\pi|_{\mathcal H_p}$ for its trace-free restriction.  A
\emph{row--column block label} is an ordered pair $a=(\pi_R,\pi_C)$ of partial matchings.  If
$\pi_R$ and $\pi_C$ leave respectively $p$ and $q$ positions unmatched, define
\[
 W_a=\mathcal H_p\otimes\mathcal H_q,
 \qquad
 \widetilde J_a=\widetilde J_{\pi_R}\otimes\widetilde J_{\pi_C},
 \qquad
 S_a=\widetilde J_a|_{W_a}.
\]
Here row tensor positions and column tensor positions are regrouped using
\[
 (\mathbb C^{n\times n})^{\otimes m}
 \cong V_{\mathbb C,R}^{\otimes m}\otimes V_{\mathbb C,C}^{\otimes m}.
\]
Every $S_a$ is an isometry because each normalized cup has norm one.
Equation~\eqref{eq:B.metric-invariance}, applied to every row and column cup, shows that the cups
are unchanged by the rotations.  The unmatched row and column indices transform by
$U^{\otimes p}$ and $V^{\otimes q}$, respectively.  Therefore
\begin{equation}
 (U\otimes V)^{\otimes m}S_a
 =S_a\bigl(U^{\otimes p}\otimes V^{\otimes q}\bigr).     \tag{B.2b}\label{eq:B.block-intertwining}
\end{equation}
This is the insertion--rotation identity used in the main proof.

There are $N_{m,r}$ matchings with $r$ pairs by \eqref{eq:B.2}.  Consequently the number of
one-side matchings is bounded by
\[
 \sum_{r=0}^{\lfloor m/2\rfloor}N_{m,r}
 \le\sum_{r=0}^{\lfloor m/2\rfloor}\frac{m!}{2^rr!}
 \le e^{1/2}m!
 \le(2m)^m.
\]
The first inequality drops the factor $(m-2r)!\ge1$ from the denominator in
\eqref{eq:B.2}; the second extends the sum and uses
$\sum_{r\ge0}(2^rr!)^{-1}=e^{1/2}$; the last follows from $m!\le m^m$ and
$e^{1/2}<2\le2^m$.  Therefore the number $D_m$ of row--column block labels satisfies
\begin{equation}
 D_m\le(2m)^{2m}\le(4m)^{4m}.                            \tag{B.3}\label{eq:B.3}
\end{equation}

On the Hilbert direct sum $\mathcal W=\bigoplus_{a=1}^{D_m}W_a$, define
\[
 S(z_1,\ldots,z_{D_m})=\sum_{a=1}^{D_m}S_az_a.
\]
The next proposition shows that these partial-matching blocks form a well-conditioned coordinate
system for the full matrix-tensor space.

\begin{proposition}[Stable pairing frame]\label{ext:pairing-frame}
The maps $S_a$ span $(\mathbb C^{n\times n})^{\otimes m}$ and their number satisfies
$D_m\le(4m)^{4m}$.  If $(4m)^{4m}\le n/2$, then, for every $z_a\in W_a$,
\[
 \frac12\sum_{a=1}^{D_m}\norm{z_a}_F^2
 \le\norm{\sum_{a=1}^{D_m}S_az_a}_F^2
 \le\frac32\sum_{a=1}^{D_m}\norm{z_a}_F^2.
\]
\end{proposition}

\begin{proof}
The count is \eqref{eq:B.3}, and the elementary trace decomposition preceding
\eqref{eq:B.2a}, applied separately to row and column indices, proves that the images of the
$S_a$ span.  Since $S_a$ is an isometry, every diagonal block of the Gram operator
$G_I=S^*S$ is the identity.  We prove that
\begin{equation}
 \norm{S_a^*S_b}_{\op}\le\frac1n\qquad(a\ne b).           \tag{B.4}\label{eq:B.4}
\end{equation}

First consider only one index family.  Let $\pi\ne\pi'$ be two partial matchings.  Superimpose
their cups and cancel every cup common to both.  Each remaining connected component is an
alternating path or an alternating cycle.  A path whose endpoints lie in different trace-free
tensors and that uses $s$ cups from each matching equals $n^{-s}$ times an identity coordinate
wire, possibly followed by a permutation; its operator norm is $n^{-s}$.  A cycle using $s$
cups from each matching has one freely summed coordinate and coefficient
$n\,n^{-s}=n^{1-s}$; a noncommon cycle has $s\ge2$.  A path whose two endpoints lie in the same
trace-free tensor contracts two of its positions and is therefore zero.  These statements follow
directly by assigning an index to every vertex: each normalized cup contributes $n^{-1/2}$, a
path leaves one input-to-output Kronecker delta, and a cycle leaves one free sum.

If $\pi$ and $\pi'$ contain different numbers of cups, endpoint counting forces a path with both
endpoints in the matching having fewer cups, so the overlap is zero by trace-freeness.  If they
contain the same number of cups but differ, some nontrivial component is either a cross path with
$s\ge1$ or a cycle with $s\ge2$, and hence contributes at most $n^{-1}$; all other components
have norm at most one.  This proves the one-side overlap bound.  For row--column labels,
$S_a^*S_b$ is the tensor product of its row and column overlaps.  Distinct labels differ on at
least one side, which contributes at most $n^{-1}$, while the other side contributes at most one.
This proves \eqref{eq:B.4}.

For $z=(z_a)_a\in\mathcal W$, \eqref{eq:B.4} and $2uv\le u^2+v^2$ give
\begin{align}
 \left|\ip z{(G_I-I)z}\right|
 &\le\frac1n\sum_{a\ne b}\norm{z_a}\norm{z_b}
 \le\frac{D_m-1}{n}\sum_a\norm{z_a}^2.                 \tag{B.5}\label{eq:B.5}
\end{align}
The hypothesis and \eqref{eq:B.3} imply $D_m\le n/2$, and therefore
\begin{equation}
 \frac12I\preceq G_I\preceq\frac32I.                    \tag{B.6}\label{eq:B.6}
\end{equation}
This is equivalent to the displayed norm comparison.  In particular, $S$ is injective; together
with spanning, it is an isomorphism from $\mathcal W$ onto the matrix-tensor space.
\end{proof}

\subsection{Elementary Haar averaging on a trace-free block}

We next prove the only Haar estimate needed for one block.  It is deliberately stated for an
arbitrary tensor in $\mathcal H_p\otimes\mathcal H_q$, so the row and column factors may be
entangled.

\begin{proposition}[Trace-free Haar estimate]\label{prop:tracefree-haar}
Let $0\le p,q\le m$ and suppose $(4m)^{4m}\le n/2$.  For every
$z\in\mathcal H_p\otimes\mathcal H_q$,
\begin{align}
 &\E_{U,V}(U^{\otimes p}\otimes V^{\otimes q})zz^*
 (U^{\otimes p}\otimes V^{\otimes q})^* \notag\\
 &\qquad\preceq
 4(p!q!)^2n^{-(p+q)}\norm z^2
 I_{\mathcal H_p\otimes\mathcal H_q}.                   \tag{B.7}\label{eq:B.7}
\end{align}
Here $U,V$ are independent Haar matrices in $\mathcal O_n$, and $0!=1$.
\end{proposition}

\begin{proof}
We first record the exact orthogonal Weingarten formula and the elementary coefficient bound that
will be used.  Let $\mathcal P_2(2p)$ be the set of pairings of $2p$ positions.  If
$\pi\in\mathcal P_2(2p)$ and $\mathbf x=(x_1,\ldots,x_{2p})$, set
\[
 \delta_{\mathbf x}^{\pi}
 =\prod_{\{r,s\}\in\pi}\mathbf1_{\{x_r=x_s\}}.
\]
Collins and P.~\'{S}niady~\cite[Corollary~3.4 and Lemma~3.5]{CollinsSniady} give
\begin{equation}
 \E_U\prod_{r=1}^{2p}U_{x_ry_r}
 =\sum_{\pi,\sigma\in\mathcal P_2(2p)}
   \delta_{\mathbf x}^{\pi}\delta_{\mathbf y}^{\sigma}
   \operatorname{Wg}_{n,p}^{O}(\pi,\sigma).             \tag{B.8}\label{eq:B.8}
\end{equation}
For $p=0$, there is one empty pairing, the empty product is one, and
$\operatorname{Wg}_{n,0}^{O}=1$.

For $p\ge1$ and $\pi\in\mathcal P_2(2p)$, define the pairing tensor
\[
 v_\pi
 =\sum_{\mathbf x\in[n]^{2p}}
   \delta_{\mathbf x}^{\pi}\,
   e_{x_1}\otimes\cdots\otimes e_{x_{2p}}
 \in(\mathbb C^n)^{\otimes 2p}.
\]
The pairing Gram matrix $G^{(p)}$, indexed by $\mathcal P_2(2p)$, is therefore
\[
 G^{(p)}_{\pi,\sigma}
 =\langle v_\pi,v_\sigma\rangle
 =\sum_{\mathbf x\in[n]^{2p}}
   \delta_{\mathbf x}^{\pi}\delta_{\mathbf x}^{\sigma}
 =n^{\ell(\pi,\sigma)}.
\]
Here $\ell(\pi,\sigma)$ is the number of connected components in the graph obtained by
superimposing the two pairings.  All coordinates within one component must be equal, while
different components can be chosen independently from $[n]$; this gives the last equality.
The diagonal entries equal $n^p$.  If $\pi\ne\sigma$, their overlay contains a cycle
of length at least four, rather than two, so $\ell(\pi,\sigma)\le p-1$ and
$G^{(p)}_{\pi,\sigma}\le n^{p-1}$.  The number of pairings satisfies
\[
 M_p=(2p-1)!!\le(2p)^p\le(2m)^m\le n/2.
\]
Write $n^{-p}G^{(p)}=I+E$.  The matrix $E$ is real symmetric, and hence
\[
 \norm E_{\op}\le\norm E_{\infty}
 \le\frac{M_p-1}{n}\le\frac12.
\]
Thus $G^{(p)}$ is invertible.  In this range the Weingarten matrix in \eqref{eq:B.8} is its
ordinary inverse, so
\begin{equation}
 \norm{\operatorname{Wg}_{n,p}^{O}}_{\op}
 \le2n^{-p},
 \qquad
 \abs{\operatorname{Wg}_{n,p}^{O}(\pi,\sigma)}\le2n^{-p}. \tag{B.9}\label{eq:B.9}
\end{equation}
The same entrywise bound holds for $p=0$.

We now test the left side of \eqref{eq:B.7} against an arbitrary
$x\in\mathcal H_p\otimes\mathcal H_q$.  Use multi-indices
$\mathbf i,\mathbf a\in[n]^p$ and $\mathbf j,\mathbf c\in[n]^q$, and write
\[
 A(U,V)
 =\left\langle x,(U^{\otimes p}\otimes V^{\otimes q})z\right\rangle
 =\sum_{\mathbf i,\mathbf j,\mathbf a,\mathbf c}
   \overline{x_{\mathbf i,\mathbf j}}z_{\mathbf a,\mathbf c}
   \prod_{r=1}^pU_{i_ra_r}\prod_{s=1}^qV_{j_sc_s}.
\]
Applying \eqref{eq:B.8} separately to $U$ and $V$ in $\abs{A(U,V)}^2$ gives the exact identity
\begin{equation*}
 \E\abs A^2
 =\sum_{\pi,\sigma\in\mathcal P_2(2p)}
   \sum_{\rho,\tau\in\mathcal P_2(2q)}
   \operatorname{Wg}_{n,p}^{O}(\pi,\sigma)
   \operatorname{Wg}_{n,q}^{O}(\rho,\tau)
   X_{\pi,\rho}(x)Z_{\sigma,\tau}(z),
\end{equation*}
where, concatenating unprimed and primed multi-indices in that order,
\begin{align*}
 X_{\pi,\rho}(x)
 &=\sum_{\mathbf i,\mathbf i',\mathbf j,\mathbf j'}
   \overline{x_{\mathbf i,\mathbf j}}x_{\mathbf i',\mathbf j'}
   \delta_{(\mathbf i,\mathbf i')}^{\pi}
   \delta_{(\mathbf j,\mathbf j')}^{\rho},\\
 Z_{\sigma,\tau}(z)
 &=\sum_{\mathbf a,\mathbf a',\mathbf c,\mathbf c'}
   z_{\mathbf a,\mathbf c}\overline{z_{\mathbf a',\mathbf c'}}
   \delta_{(\mathbf a,\mathbf a')}^{\sigma}
   \delta_{(\mathbf c,\mathbf c')}^{\tau}.
\end{align*}

Split the $2p$ row positions into the unprimed and primed halves.  Call a pairing \emph{cross}
if every pair contains one position from each half.  If $\pi$ is not cross, it contains a pair
inside one half.  Summing the common index in that pair takes an unnormalized contraction of two
row positions of $x$ or $\overline x$, pointwise in all remaining row and column indices.  This
is zero because
\[
 (C_{rs}^{(p)}\otimes I_{\mathcal H_q})x=0
 \quad\text{and}\quad
 (C_{rs}^{(p)}\otimes I_{\mathcal H_q})\overline x
 =\overline{(C_{rs}^{(p)}\otimes I_{\mathcal H_q})x}=0.
\]
The deltas from the other pairs do not involve either contracted position and therefore do not
affect this cancellation.  Hence $X_{\pi,\rho}(x)=0$ unless $\pi$ is cross.  The identical
argument in the column positions, and then for $z$, shows that a term survives only if all four
pairings $\pi,\rho,\sigma,\tau$ are cross.

Every cross pairing in $\mathcal P_2(2p)$ has the unique form
\[
 \pi_\gamma=\{\{r,\gamma(r)'\}:1\le r\le p\},
 \qquad \gamma\in\mathfrak S_p,
\]
where $\mathfrak S_p$ is the symmetric group of all permutations of $\{1,\ldots,p\}$.
Likewise, for $\eta\in\mathfrak S_q$, write
\[
 \rho_\eta=\{\{s,\eta(s)'\}:1\le s\le q\}
 \in\mathcal P_2(2q).
\]
More generally, for $\theta\in\mathfrak S_d$, define the place-permutation operator on order-$d$ tensors
by
\[
 (\Pi_\theta h)_{i_1\cdots i_d}
 =h_{i_{\theta^{-1}(1)}\cdots i_{\theta^{-1}(d)}}.
\]
Direct substitution of the cross-pairing deltas gives, for
$\gamma,\alpha\in\mathfrak S_p$ and $\eta,\beta\in\mathfrak S_q$,
\[
 X_{\pi_\gamma,\rho_\eta}(x)
 =\left\langle x,(\Pi_\gamma\otimes\Pi_\eta)x\right\rangle,
 \qquad
 Z_{\pi_\alpha,\rho_\beta}(z)
 =\overline{\left\langle z,(\Pi_\alpha\otimes\Pi_\beta)z\right\rangle}.
\]
Permutation operators are unitary, even when $x$ and $z$ are entangled across their row and
column factors.  Thus every surviving $X$ has magnitude at most $\norm x^2$ and every surviving
$Z$ has magnitude at most $\norm z^2$.  There are $p!$ cross row pairings and $q!$ cross column
pairings.  Summing over the four pairing indices and using \eqref{eq:B.9} for the two
Weingarten coefficients yields
\[
 \E\abs A^2
 \le4(p!q!)^2n^{-(p+q)}\norm x^2\norm z^2.
\]
The averaged operator preserves and is supported on $\mathcal H_p\otimes\mathcal H_q$ because
that subspace is invariant under $U^{\otimes p}\otimes V^{\otimes q}$.  The last scalar
inequality for every $x$ in that subspace is therefore equivalent to \eqref{eq:B.7}.
\end{proof}

\subsection{The measurement graph for one frame block}\label{app:measurement-graph}

Fix a row--column block $a=(\pi_R,\pi_C)$, let $p,q$ be its numbers of unmatched row and column
positions, and take $z\in W_a=\mathcal H_p\otimes\mathcal H_q$.  The normalized cups commute
with every orthogonal transformation, so Proposition~\ref{prop:tracefree-haar} and
$S_a=\widetilde J_a|_{W_a}$ imply
\begin{align*}
 &\E_{U,V}(U\otimes V)^{\otimes m}S_azz^*S_a^*
       (U\otimes V)^{*\otimes m}\\
 &\qquad\preceq
 4(p!q!)^2n^{-(p+q)}\norm z^2S_aS_a^*\\
 &\qquad\preceq
 4(p!q!)^2n^{-(p+q)}\norm z^2
       \widetilde J_a\widetilde J_a^*.
\end{align*}
The second inequality uses
$S_aS_a^*=\widetilde J_a\Pi_{W_a}\widetilde J_a^*$, where $\Pi_{W_a}$ is the orthogonal
projection onto the trace-free unmatched tensor space, and $0\preceq\Pi_{W_a}\preceq I$.
Taking the trace against the positive operator $P^{\otimes m}$ gives
\begin{equation}
 \E_{U,V}\norm{P^{\otimes m}(U\otimes V)^{\otimes m}S_az}^2
 \le4(p!q!)^2n^{-(p+q)}\norm z^2
   \tr\!\left(P^{\otimes m}\widetilde J_a\widetilde J_a^*\right). \tag{B.10}\label{eq:B.10}
\end{equation}

Let $r_R=(m-p)/2$ and $r_C=(m-q)/2$ be the numbers of row and column cups.  In the trace on the
right side of \eqref{eq:B.10}, the squared normalizations of these cups contribute
\begin{equation}
 n^{-r_R-r_C}=n^{-m+(p+q)/2}.                             \tag{B.11}\label{eq:B.11}
\end{equation}
It remains to bound the index contraction left after removing the factor \eqref{eq:B.11}.  Draw one vertex
for each of the $m$ copies of $P$.  The row and column partial matchings give row and column
edges.  Their union has degree at most two and hence is a disjoint union of alternating paths and
even cycles.  A vertex unmatched on the row side is called a row endpoint, and similarly for a
column endpoint.  Let $o$ be the number of paths with unlike endpoints, and let $u_R,u_C$ be the
numbers with two row endpoints and two column endpoints.  There are $p$ row endpoints and $q$
column endpoints, so
\begin{equation}
 2u_R+o=p,\qquad 2u_C+o=q.                                \tag{B.12}\label{eq:B.12}
\end{equation}

We now give the complete indexed calculation for one component.  At vertex $v$, write the entry
of $P$ as
\[
 P_{(i_v,j_v),(i'_v,j'_v)}.
\]
A row edge $v$--$w$ imposes $i_v=i_w$ and $i'_v=i'_w$; a column edge imposes
$j_v=j_w$ and $j'_v=j'_w$.  At a row endpoint the identity on the unpaired position imposes
$i_v=i'_v$, and at a column endpoint it imposes $j_v=j'_v$.  These rules include every index in
the contraction and show directly that different connected components factor.

Define the linear map $\Phi:M_n(\mathbb C)\to M_n(\mathbb C)$ by
\begin{equation}
 [\Phi(X)]_{ii'}=\sum_{j,j'=1}^n
       P_{(i,j),(i',j')}X_{jj'}.                          \tag{B.13a}\label{eq:B.13a}
\end{equation}
Its matrix in the matrix-unit bases is the realignment
$R_{(i,i'),(j,j')}=P_{(i,j),(i',j')}$.  The projection $P$ in Lemma~\ref{lem:twirl} is real,
so $R^*=R^T$; we retain adjoint notation because it displays the positive operators in the
calculation.  Use the convention
$\operatorname{vec}(A)=\sum_{i,j}A_{ij}e_i\otimes e_j$.  Because $P$ is a rank-$k$ orthogonal
projection, choose matrices $A_1,\ldots,A_k\in M_n(\mathbb C)$ such that
$\operatorname{vec}(A_1),\ldots,\operatorname{vec}(A_k)$ form an orthonormal basis of
$\operatorname{range}(P)$.  Then
$P=\sum_{s=1}^k\operatorname{vec}(A_s)\operatorname{vec}(A_s)^*$, and hence
$\Phi(X)=\sum_sA_sXA_s^*$.  Thus $\Phi$ and its Hilbert--Schmidt adjoint $\Phi^*$ are completely
positive.  Moreover, $0\preceq P\preceq I$ implies, for all $X\succeq0$,
\begin{equation}
 0\preceq\Phi(X)\preceq(\tr X)I,
 \qquad
 0\preceq\Phi^*(X)\preceq(\tr X)I.                       \tag{B.13b}\label{eq:B.13b}
\end{equation}
Indeed, for vectors $x,y\in\mathbb C^n$,
\[
 y^*\Phi(xx^*)y
 =\langle y\otimes\overline{x},P(y\otimes\overline{x})\rangle
 \le\norm y^2\norm x^2;
\]
spectrally decompose $X$ to obtain the first inequality, and use
$y^*\Phi^*(xx^*)y=x^*\Phi(yy^*)x$ for the second.  Finally,
\begin{equation}
 \tr\Phi(I)=\tr\Phi^*(I)=\tr P=k,
 \qquad
 \operatorname{Tr}_{\mathrm{HS}}(\Phi^*\Phi)=\norm R_F^2=\norm P_F^2=k. \tag{B.13c}\label{eq:B.13c}
\end{equation}
Here $\operatorname{Tr}_{\mathrm{HS}}$ is the trace of an operator on the $n^2$-dimensional
Hilbert space $M_n(\mathbb C)$.  To spell out the two applications of
\eqref{eq:B.13b}, complete positivity first gives
$\Phi(I)\succeq0$ and $\Phi^*(I)\succeq0$.  Apply the bound for $\Phi^*$ to the positive
matrix $X=\Phi(I)$ and the bound for $\Phi$ to the positive matrix
$X=\Phi^*(I)$.  Using \eqref{eq:B.13c} to evaluate the two traces gives
\[
 \Phi^*(\Phi(I))
 \preceq\tr(\Phi(I))I=kI,
 \qquad
 \Phi(\Phi^*(I))
 \preceq\tr(\Phi^*(I))I=kI.
\]
Equivalently,
\begin{equation}
 \Phi^*\Phi(I)\preceq kI,
 \qquad \Phi\Phi^*(I)\preceq kI.                         \tag{B.13d}\label{eq:B.13d}
\end{equation}

To identify the component contractions, orient and relabel an alternating component.  For a
cycle of length $2\ell$, take row edges $(1,2),(3,4),\ldots,(2\ell-1,2\ell)$ and column edges
$(2,3),(4,5),\ldots,(2\ell,1)$.  Substitution of the preceding Kronecker-delta rules gives
\[
 \sum_{\substack{x_s,x'_s\\y_s,y'_s}}
 \prod_{s=1}^{\ell}
 P_{(x_s,y_{s-1}),(x'_s,y'_{s-1})}
 P_{(x_s,y_s),(x'_s,y'_s)}
 =\operatorname{Tr}_{\mathrm{HS}}\bigl((\Phi^*\Phi)^\ell\bigr),        \tag{B.13e}\label{eq:B.13e}
\]
where $y_0=y_\ell$ and $y'_0=y'_\ell$.  Reversing the starting color replaces
$\Phi^*\Phi$ by $\Phi\Phi^*$ and leaves the value unchanged.  If
$d=\operatorname{vec}(I)$, the identical index substitution for paths gives
\begin{align}
 \text{row--column path with $2\ell-1$ vertices:}\quad
 &d^*R(R^*R)^{\ell-1}d
   =\tr\!\left[\Phi((\Phi^*\Phi)^{\ell-1}(I))\right],       \tag{B.13f}\label{eq:B.13f}\\
 \text{row--row path with $2\ell$ vertices:}\quad
 &d^*(RR^*)^\ell d
   =\langle I,(\Phi\Phi^*)^\ell(I)\rangle_{\mathrm{HS}}.   \tag{B.13g}\label{eq:B.13g}
\end{align}
The column--row and column--column cases exchange $\Phi$ and $\Phi^*$.  Equations
\eqref{eq:B.13e}--\eqref{eq:B.13g} are simply matrix multiplication with
$R_{(i,i'),(j,j')}=P_{(i,j),(i',j')}$; the endpoint entries of $d$ are precisely the endpoint
conditions $i=i'$ or $j=j'$ stated above.

The bounds are now immediate and contain no diagrammatic peeling step.  The eigenvalues of the
positive operator $\Phi^*\Phi$ are nonnegative and sum to $k$ by \eqref{eq:B.13c}, so
$\operatorname{Tr}_{\mathrm{HS}}((\Phi^*\Phi)^\ell)\le k^\ell$.  The maps
$\Phi^*\Phi$ and $\Phi\Phi^*$ preserve the positive-semidefinite order; \eqref{eq:B.13d} therefore gives,
by induction,
\[
 (\Phi^*\Phi)^r(I)\preceq k^rI,
 \qquad (\Phi\Phi^*)^r(I)\preceq k^rI.
\]
Using these inequalities in \eqref{eq:B.13f}--\eqref{eq:B.13g}, together with $\tr\Phi(I)=k$, proves
\begin{align}
 \text{cycle with $2\ell$ $P$-vertices}&\le k^\ell, \notag\\
 \text{row--column path with $2\ell-1$ $P$-vertices}&\le k^\ell,\notag\\
 \text{row--row or column--column path with $2\ell$ $P$-vertices}&\le nk^\ell. \tag{B.14}\label{eq:B.14}
\end{align}

The graph contains $m$ copies of $P$.  Multiplying \eqref{eq:B.14} over its components gives
\begin{equation}
 n^{u_R+u_C}k^{(m+o)/2}.                                  \tag{B.15}\label{eq:B.15}
\end{equation}
Indeed, a path with unlike endpoints uses one fewer $P$-vertex than twice its power of $k$;
summing this deficit over the $o$ such paths gives the exponent $(m+o)/2$.

Multiply the factor in \eqref{eq:B.10}, the cup normalization \eqref{eq:B.11}, and the
measurement-graph estimate \eqref{eq:B.15}.  By \eqref{eq:B.12},
$u_R+u_C=(p+q-2o)/2$, so the total power of $n$ is
\[
 -(p+q)-m+\frac{p+q}{2}+\frac{p+q-2o}{2}=-m-o.
\]
Set $C_{\mathrm{ten},1}=8$.  Since
$4(p!q!)^2\le4(m!)^4\le4m^{4m}\le(C_{\mathrm{ten},1}m)^{C_{\mathrm{ten},1}m}$,
every unit vector $z\in W_a$---including vectors
entangled between its row and column factors---satisfies
\begin{equation}
 \E_{U,V}\norm{P^{\otimes m}(U\otimes V)^{\otimes m}S_az}^2
 \le(C_{\mathrm{ten},1}m)^{C_{\mathrm{ten},1}m}
 \left(\frac{k}{n^2}\right)^{(m+o)/2}.                    \tag{B.16}\label{eq:B.16}
\end{equation}
For odd $m$, both $p$ and $q$ are positive and odd.  Equations \eqref{eq:B.12} imply that $o$ is a
positive odd integer, so $o\ge1$.  Also $k\le n^2$ because $P$ acts on the $n^2$-dimensional
matrix space.  Hence $k/n^2\le1$, and the right side of \eqref{eq:B.16} is at most the same
prefactor times $(k/n^2)^{(m+1)/2}$.

\subsection{From one frame block to every tensor}\label{app:block-recombination}

Let $G_P$ be the block Gram operator obtained after Haar averaging and applying $P^{\otimes m}$.
It is positive semidefinite because
\[
 z^*G_Pz
 =\E_{U,V}\norm{P^{\otimes m}(U\otimes V)^{\otimes m}Sz}^2\ge0.
\]
For unit vectors $x\in W_a,y\in W_b$, positivity of the restriction to their two-dimensional
span gives
\[
 |\ip{x}{G_P(a,b)y}|^2
 \le\ip{x}{G_P(a,a)x}\ip{y}{G_P(b,b)y}.
\]
Indeed, the positive-semidefinite operator $G_P$ defines a positive-semidefinite sesquilinear
form $\langle u,v\rangle_{G_P}=\langle u,G_Pv\rangle$.  Cauchy--Schwarz for this form, applied
to the vectors supported in blocks $a$ and $b$, is exactly the displayed inequality.  Taking the
supremum over unit $x,y$ gives the corresponding block-operator inequality.
Equation \eqref{eq:B.16} and $o\ge1$ therefore give the block-operator bound
\[
 \norm{G_P(a,b)}_{\op}\le
 \gamma_m,\qquad
 \gamma_m=(C_{\mathrm{ten},1}m)^{C_{\mathrm{ten},1}m}
 \left(\frac{k}{n^2}\right)^{(m+1)/2}.
\]
For an arbitrary tensor $T\in(\mathbb C^{n\times n})^{\otimes m}$, write $T=Sz$ using the frame.  The
lower bound in \eqref{eq:B.6} makes $S$ injective, while the trace decomposition makes it surjective, so
this coefficient vector exists uniquely.  Then \eqref{eq:B.3},
Cauchy--Schwarz, and the lower Gram bound \eqref{eq:B.6} give
\[
\begin{aligned}
 \E\norm{P^{\otimes m}(U\otimes V)^{\otimes m}T}^2
 &=z^*G_Pz\\
 &\le\gamma_m\left(\sum_a\norm{z_a}\right)^2\\
 &\le\gamma_mD_m\sum_a\norm{z_a}^2\\
 &\le (C_{\mathrm{ten},2}m)^{C_{\mathrm{ten},2}m}
 \left(\frac{k}{n^2}\right)^{(m+1)/2}\norm{T}^2,
\end{aligned}
\]
where $C_{\mathrm{ten},2}$ is one fixed constant large enough to dominate the product of the
specific factors in \eqref{eq:B.3}, \eqref{eq:B.6}, and \eqref{eq:B.16}.  Finally set
\[
 C_{\mathrm{ten}}
 \defeq10\left(1+C_{\mathrm{ten},1}+C_{\mathrm{ten},2}\right).
\]
Substituting this envelope yields exactly Lemma~\ref{lem:twirl}; no constant is changed between
the displayed inequalities over the complexified tensor space.  The original $P$, $U$, and $V$
are real operators.  Their complex-linear extensions have the same operator actions and the
Hermitian norm of a real tensor equals its original Frobenius norm.  Restricting the proved
complex inequality to $T\in(\R^{n\times n})^{\otimes m}$ therefore gives precisely the real
statement of Lemma~\ref{lem:twirl}.

\section{Score identities and energy estimates}\label{app:score}

This appendix proves the analytic estimates used in Lemma~\ref{lem:path-tv}.  All random
expectations here are over independent Haar $U,V\in\mathcal O_n$ unless a Gaussian is explicitly
present.

\subsection{Rotation derivatives and integration by parts}\label{app:rotation-calculus}

For $1\le a<b\le n$, let $E^{ab}$ be the skew-symmetric matrix with entries $+1$ at
$(a,b)$ and $-1$ at $(b,a)$.  These matrices are orthonormal for the standard rotation
inner product $\ip AB_{\rm rot}=-\tfrac12\tr(AB)$.  Because
$(E^{ab})^T=-E^{ab}$,
\[
 (e^{sE^{ab}})^Te^{sE^{ab}}
 =e^{-sE^{ab}}e^{sE^{ab}}=I.
\]
Thus $e^{sE^{ab}}\in\mathcal O_n$ for every real $s$; it is the ordinary rotation through angle
$s$ in the $(a,b)$ coordinate plane and acts as the identity on the other coordinates.  Let
$f,h:\mathcal O_n\times\mathcal O_n\to\R$ be continuously differentiable functions.  Define
their row and column derivatives by
\[
\begin{aligned}
 \partial^R_{ab}f(U,V)&=\left.\frac d{ds}f(Ue^{sE^{ab}},V)\right|_{s=0},\\
 \partial^C_{ab}f(U,V)&=\left.\frac d{ds}f(U,Ve^{sE^{ab}})\right|_{s=0}.
\end{aligned}
\]
The integration-by-parts formula follows directly from Haar invariance.  Haar measure is unchanged by right
multiplication, so $Ue^{sE^{ab}}$ has the same distribution as $U$ for every fixed $s$.  Applying
this invariance to the product $fh$ gives
\[
 \E\!\left[(fh)(Ue^{sE^{ab}},V)\right]
 =\E[(fh)(U,V)].
\]
The right-hand side is constant in $s$.  Differentiate at $s=0$ and use the ordinary product rule:
\begin{align*}
 0
 &=\E\!\left[
   \left.\frac d{ds}(fh)(Ue^{sE^{ab}},V)\right|_{s=0}\right]\\
 &=\E[(\partial^R_{ab}f)h]+\E[f(\partial^R_{ab}h)].
\end{align*}
Rearranging proves the desired identity for a row derivative.  The same calculation, now using
the invariance of $V$ under $V\mapsto Ve^{sE^{ab}}$, proves it for a column derivative.  Thus, for
every row or column derivative $\partial_\alpha$,
\begin{equation}
 \E[(\partial_\alpha f)h]=-\E[f(\partial_\alpha h)]     \tag{C.1}\label{eq:C.1}
\end{equation}
This is the rotation-group analogue of integration by parts on a circle: there is no boundary term
because a full rotation returns to the starting point.

Equip the space of smooth functions on $\mathcal O_n\times\mathcal O_n$ with the Haar
$L_2$ inner product
\[
 \langle f,h\rangle_{L_2}:=\E[f(U,V)h(U,V)].
\]
Let $\partial_\alpha$ range over all row and column derivatives defined above, and define
\[
 \Delta_{\mathrm{rot}}=\sum_\alpha \partial_\alpha^2,
 \qquad K_0=-\Delta_{\mathrm{rot}},
\]
together with the bilinear form
\begin{equation}
 \Gamma(f,h)
 =\frac12\{\Delta_{\mathrm{rot}}(fh)
 -f\Delta_{\mathrm{rot}}h-h\Delta_{\mathrm{rot}}f\}
 =\sum_\alpha(\partial_\alpha f)(\partial_\alpha h).     \tag{C.2}\label{eq:C.2}
\end{equation}
The last equality follows by applying the ordinary product rule to each $\partial_\alpha^2(fh)$.

Apply \eqref{eq:C.1} with $h$ replaced by $\partial_\alpha h$ and sum over $\alpha$.  This gives
\begin{align*}
 \langle f,K_0h\rangle_{L_2}
 &=-\sum_\alpha\E[f\partial_\alpha^2h]
 =\sum_\alpha\E[(\partial_\alpha f)(\partial_\alpha h)]
 =\langle K_0f,h\rangle_{L_2}.
\end{align*}
Thus $K_0$ is self-adjoint.  Taking $h=f$ yields the energy identity
\begin{equation*}
 \langle f,K_0f\rangle_{L_2}
 =\sum_\alpha\E[(\partial_\alpha f)^2]\ge0,
\end{equation*}
so $K_0$ is positive semidefinite and $\Delta_{\mathrm{rot}}=-K_0$ is negative semidefinite.
Combining this calculation with \eqref{eq:C.2} gives
\begin{equation}
 \E\Gamma(f,h)=\E[fK_0h]=\E[(K_0f)h].                    \tag{C.3}\label{eq:C.3}
\end{equation}
These formulas hold componentwise for vector- or tensor-valued functions.

The kernel of $K_0$ can be identified explicitly.  If $K_0f=0$, then the quadratic-form identity
above gives
\[
 0=\langle f,K_0f\rangle_{L_2}
  =\sum_\alpha\E[(\partial_\alpha f)^2].
\]
Every summand is nonnegative, so $\partial_\alpha f=0$ almost everywhere for every row and column
rotation direction.  Since $f$ is smooth, these derivatives vanish everywhere.  Hence $f$ is
unchanged along each plane-rotation path:
\[
 f(Ue^{sE^{ab}},V)=f(U,V),
 \qquad
 f(U,Ve^{sE^{ab}})=f(U,V).
\]

Plane rotations generate $\mathrm{SO}(n)$.  One elementary verification is to apply successive
Givens rotations to reduce any $Q\in\mathrm{SO}(n)$ to a diagonal matrix with entries in
$\{+1,-1\}$.  Its determinant is $1$, so the number of $-1$ entries is even; each pair of such
entries is a rotation through angle $\pi$ in the corresponding coordinate plane.  Reversing the
reductions expresses $Q$ as a product of plane rotations.  It follows that
\[
 f(UQ_1,VQ_2)=f(U,V)
 \qquad(Q_1,Q_2\in\mathrm{SO}(n)).
\]

The determinant is continuous and takes only the values $\pm1$ on $\mathcal O_n$.  Thus matrices
with different determinants lie in different connected components.  The preceding generation by
plane rotations also shows that $\mathrm{SO}(n)$ is connected, because every plane rotation is
connected to the identity by its angle parameter.  If $R$ is any fixed reflection, then
$R\,\mathrm{SO}(n)$ is the determinant-$-1$ component and is connected as well.  Consequently,
$\mathcal O_n$ has exactly two connected components, indexed by the determinant, and
$\mathcal O_n\times\mathcal O_n$ has four connected components indexed by
\[
 (\det U,\det V)\in\{+1,-1\}^2.
\]
The invariance under $Q_1,Q_2\in\mathrm{SO}(n)$ therefore says that every $f\in\ker K_0$ is
constant on each of these four components.  Conversely, a function constant on each component
has every $\partial_\alpha f$ equal to zero and hence lies in $\ker K_0$.

It remains to describe this four-dimensional space.  Write $s=\det U$, $t=\det V$, and let
$c_{s,t}$ be the value of $f$ on the component labeled by $(s,t)$.  Every function on
$\{+1,-1\}^2$ has the unique expansion
\[
 c_{s,t}=a_0+a_1s+a_2t+a_{12}st,
\]
where
\[
 a_0=\frac14\sum_{s,t=\pm1}c_{s,t},\qquad
 a_1=\frac14\sum_{s,t=\pm1}s\,c_{s,t},\qquad
 a_2=\frac14\sum_{s,t=\pm1}t\,c_{s,t},\qquad
 a_{12}=\frac14\sum_{s,t=\pm1}st\,c_{s,t}.
\]
Therefore
\begin{equation*}
 \ker K_0
 =\operatorname{span}\{1,\det U,\det V,(\det U)(\det V)\}.
\end{equation*}
The only external input in this subsection is the following standard Poincar\'e inequality.

\begin{theorem}[Normalized orthogonal-group Poincar\'e inequality
  {\cite[Section~3.2]{SaloffCosteCompact}}]\label{thm:orthogonal-poincare}
Let $n\ge3$, and for a smooth function $g:\mathcal O_n\to\R$ define
\[
 \partial_{ab}g(Q)=\left.\frac d{ds}g(Qe^{sE^{ab}})\right|_{s=0}.
\]
On either connected component $\mathcal C$ of $\mathcal O_n$, if $g$ has mean zero under the
normalized Haar law on $\mathcal C$, then
\begin{equation*}
 (n-1)\E_{Q\in\mathcal C}[g(Q)^2]
 \le\sum_{1\le a<b\le n}\E_{Q\in\mathcal C}[(\partial_{ab}g(Q))^2].
\end{equation*}
\end{theorem}

In the normalization of \cite[Section~3.2]{SaloffCosteCompact}, the positive generator is
$\{2(n-2)\}^{-1}(-\sum_{a<b}\partial_{ab}^2)$ and its gap on $\mathrm{SO}(n)$ is
$(n-1)/\{2(n-2)\}$, which gives the constant $n-1$ above after rescaling; left multiplication by
a fixed reflection transports normalized Haar measure and the derivative energy from
$\mathrm{SO}(n)$ to the determinant-$-1$ component, so the same inequality holds there.

The product version follows elementarily from the variance decomposition.  Fix one connected
component of $\mathcal O_n\times\mathcal O_n$, use its normalized Haar law, and suppose that $f$
has mean zero there.  Set $\bar f(V)=\E_U[f(U,V)]$, where the expectation is over the fixed
component of the first factor.  Then
\begin{align*}
 \E[f^2]
 &=\E[(f-\bar f)^2]+\E[\bar f^2]\\
 &\le\frac1{n-1}\left\{
   \sum_{a<b}\E[(\partial_{ab}^Rf)^2]
   +\sum_{a<b}\E[(\partial_{ab}^C\bar f)^2]\right\}\\
 &\le\frac1{n-1}\left\{
   \sum_{a<b}\E[(\partial_{ab}^Rf)^2]
   +\sum_{a<b}\E[(\partial_{ab}^Cf)^2]\right\}.
\end{align*}
The first inequality applies Theorem~\ref{thm:orthogonal-poincare} first to
$f(U,V)-\bar f(V)$ as a function of $U$ and then to the mean-zero function $\bar f(V)$.  The
second uses
$\partial_{ab}^C\bar f=\E_U[\partial_{ab}^Cf]$ and Jensen's inequality.  A function orthogonal to
$\ker K_0$ has mean zero on each of the four components, because $\ker K_0$ is the space of all
componentwise constant functions.  Summing the four component inequalities therefore gives
\begin{equation*}
 (n-1)\norm f_{L_2}^2
 \le\sum_\alpha\E[(\partial_\alpha f)^2]
 =\langle f,K_0f\rangle_{L_2}.
\end{equation*}

The Poincar\'e inequality says that every positive spectral value of $K_0$ is at least $n-1$.
The Moore--Penrose inverse $K_0^\dagger$ is zero on $\ker K_0$ and replaces every positive
spectral value by its reciprocal.  Spectral calculus therefore gives
\begin{equation}
 K_0^\dagger\preceq(n-1)^{-1}I
 \quad\text{on }(\ker K_0)^\perp.                          \tag{C.4}\label{eq:C.4}
\end{equation}

For later use, it remains to check that the low-degree polynomials appearing in the tensor
recursion are orthogonal to this kernel.  Consider first a monomial in the entries of $U$ having
degree strictly less than $n$.  At least one column of $U$ is absent from the monomial.  Right
multiplication by the diagonal reflection that changes the sign of this column leaves the monomial
unchanged, preserves Haar measure, and reverses $\det U$.  The Haar inner product of the monomial
with $\det U$ is therefore its own negative and hence is zero.  Linearity gives the same conclusion
for every polynomial of degree below $n$.  Applying the identical argument to $V$ proves
orthogonality to $\det V$; flipping an unused column of either matrix also proves orthogonality to
$(\det U)(\det V)$.  Thus a polynomial of degree below $n$ in each orthogonal factor and with
mean zero is orthogonal to all four functions spanning $\ker K_0$.

\subsection{Moment cancellation from the preserved power sums}
\label{app:moment-cancellation}

Recall
\[
 F_t=\sqrt n\,UD_tV^T,\qquad
 Q_1=\dot F_t,\qquad
 Q_{j+1}=\operatorname{Sym}\Gamma(K_0^\dagger Q_j,F_t).
\]
All tensor products in the recursion below are over $\mathcal H=\R^{n\times n}$.  More generally,
for an order-$s$ tensor $T\in\mathcal X^{\otimes s}$ over any real Hilbert space $\mathcal X$, let
$P_\pi$ permute its tensor positions according to $\pi\in\mathfrak S_s$; explicitly,
\begin{equation*}
 P_\pi(A_1\otimes\cdots\otimes A_s)
 =A_{\pi^{-1}(1)}\otimes\cdots\otimes A_{\pi^{-1}(s)}.
\end{equation*}
Define
\begin{equation*}
 \operatorname{Sym}T:=\frac1{s!}\sum_{\pi\in\mathfrak S_s}P_\pi T.
\end{equation*}
Thus $\operatorname{Sym}$ is the orthogonal projection onto the tensors invariant under every
permutation of their positions; in particular, it cannot increase the Frobenius norm.  We use the
convention $F_t^{\otimes0}=1$.  For tensor-valued functions $G$ and $H$, the componentwise
extension of \eqref{eq:C.2} is
\begin{equation*}
 \Gamma(G,H)=\sum_\alpha(\partial_\alpha G)\otimes(\partial_\alpha H).
\end{equation*}

Write the diagonal matrix $D_t$ as
\begin{equation*}
 D_t=\operatorname{diag}(d_1(t),\ldots,d_n(t)).
\end{equation*}
Thus $d_i(t)$ is the $i$th diagonal entry of $D_t$.  First consider $\E F_t^{\otimes m}$.
If $m$ is odd, it is zero by changing $U$ to $-U$.
For $m=2r<n$, apply the orthogonal Weingarten formula
\cite[Corollary~3.4]{CollinsSniady} separately to the $U$ entries and the $V$ entries.
We expand this contraction explicitly.  For multi-indices
$\mathbf a,\mathbf b,\mathbf i\in[n]^{2r}$, the entry formula
\[
 (F_t)_{ab}=\sqrt n\sum_{i=1}^n U_{ai}d_i(t)V_{bi}
\]
gives
\begin{align}
 \bigl[\E F_t^{\otimes2r}\bigr]_{\mathbf a,\mathbf b}
 &=n^r\sum_{\mathbf i\in[n]^{2r}}
   \left(\prod_{s=1}^{2r}d_{i_s}(t)\right)
   \left(\E_U\prod_{s=1}^{2r}U_{a_si_s}\right)
   \left(\E_V\prod_{s=1}^{2r}V_{b_si_s}\right).          \tag{C.5a}\label{eq:C.5a}
\end{align}
For a pairing $\pi\in\mathcal P_2(2r)$, recall the notation
$\delta_{\mathbf x}^{\pi}
=\prod_{\{s,t\}\in\pi}\mathbf1_{\{x_s=x_t\}}$ from \eqref{eq:B.8}.
Substituting \eqref{eq:B.8} twice into \eqref{eq:C.5a} yields
\begin{align}
 \bigl[\E F_t^{\otimes2r}\bigr]_{\mathbf a,\mathbf b}
 &=n^r
 \sum_{\pi,\sigma,\rho,\tau\in\mathcal P_2(2r)}
 \delta_{\mathbf a}^{\pi}\delta_{\mathbf b}^{\rho}
 \operatorname{Wg}_{n,r}^{O}(\pi,\sigma)
 \operatorname{Wg}_{n,r}^{O}(\rho,\tau)
 \mathcal S_{\sigma,\tau}(D_t),                         \tag{C.5b}\label{eq:C.5b}\\
 \mathcal S_{\sigma,\tau}(D)
 &\defeq
 \sum_{\mathbf i\in[n]^{2r}}
 \delta_{\mathbf i}^{\sigma}\delta_{\mathbf i}^{\tau}
 \prod_{s=1}^{2r}d_{i_s}.                               \tag{C.5c}\label{eq:C.5c}
\end{align}
The graph on $\{1,\ldots,2r\}$ with edge sets $\sigma$ and $\tau$ is a disjoint union
$\mathcal C(\sigma,\tau)$ of alternating even cycles, allowing a two-vertex cycle when the
two pairings share an edge.  On a connected component $C$, the Kronecker deltas in
\eqref{eq:C.5c} force all indices $i_s$, $s\in C$, to have one common value.  Different
components have independent common values.  Consequently the remaining sum factors exactly as
\begin{equation}
 \mathcal S_{\sigma,\tau}(D_t)
 =\prod_{C\in\mathcal C(\sigma,\tau)}
   \left(\sum_{i=1}^n d_i(t)^{|C|}\right)
 =\prod_{C\in\mathcal C(\sigma,\tau)}
   \tr D_t^{|C|}.                                       \tag{C.5d}\label{eq:C.5d}
\end{equation}
Every $|C|$ is of the form $2\ell$, and the component sizes sum to $2r$.
Thus \eqref{eq:C.5b} shows, coefficient by coefficient, that all dependence on $t$ is through
products of
\[
 \tr D_t^2,\ \tr D_t^4,\ldots,\tr D_t^{2r}.
\]
If $2r\le2K$, each one of these power sums is constant by \eqref{eq:4.9b}; all remaining
factors in \eqref{eq:C.5b} depend only on $n$ and the four pairings.  Together with the odd
case, this proves, on every differentiable segment of the particle path,
\begin{equation}
 \frac d{dt}\E F_t^{\otimes m}=0\qquad(m\le2K),            \tag{C.5}\label{eq:C.5}
\end{equation}
with the finitely many joining times understood segmentwise.

We next prove, simultaneously, the identity
\begin{equation}
 \E\operatorname{Sym}\!\left(Q_j\otimes F_t^{\otimes r}\right)
 =\frac{r!}{(r+j)!}\frac d{dt}\E F_t^{\otimes(r+j)}.       \tag{C.6}\label{eq:C.6}
\end{equation}
for every $1\le j\le2K+1$ and integer $r\ge0$, the fact that every entry of $Q_j$ is a polynomial
of degree at most $j$ in each of $U$ and $V$, and the orthogonality
$Q_j\perp\ker K_0$ for $1\le j\le2K$, understood entrywise.  The orthogonality is part of the
induction, so it will be available before we use $K_0K_0^\dagger Q_j=Q_j$.

For $j=1$, differentiating $F_t^{\otimes(r+1)}$ marks each of its $r+1$ positions once.
Symmetrization averages these marked terms, and therefore gives \eqref{eq:C.6}.  The entries of
$Q_1=\dot F_t$ have degree one in each orthogonal factor.  Taking $r=0$ in \eqref{eq:C.6} and
using \eqref{eq:C.5} gives $\E Q_1=0$.  Since $1<n$, the kernel calculation at the end of
Appendix~\ref{app:rotation-calculus} shows that $Q_1\perp\ker K_0$.  This proves all three
claims at the base level.

Now suppose the three claims hold at an index $j\le2K$.  Since $Q_j\perp\ker K_0$, setting
$g=K_0^\dagger Q_j$ gives
\begin{equation*}
 K_0g=K_0K_0^\dagger Q_j=Q_j.
\end{equation*}
The tensor product rule, followed by symmetrization over all $j+r+1$ output positions, says
\[
 \operatorname{Sym}\Gamma(g,F_t^{\otimes(r+1)})
 =(r+1)\operatorname{Sym}
   \left(\Gamma(g,F_t)\otimes F_t^{\otimes r}\right).
\]
This is the ordinary product rule: $\partial_\alpha$ differentiates one of the $r+1$ copies of $F_t$,
and full symmetrization makes the resulting $r+1$ terms equal.  Since
$Q_{j+1}=\operatorname{Sym}\Gamma(g,F_t)$ and a second full symmetrization does not change the
result, \eqref{eq:C.3}, applied componentwise, gives
\[
\begin{aligned}
 \E\operatorname{Sym}(Q_{j+1}\otimes F_t^{\otimes r})
 &=\frac1{r+1}\E\operatorname{Sym}
      \Gamma(g,F_t^{\otimes(r+1)})\\
 &=\frac1{r+1}\E\operatorname{Sym}
      ((K_0g)\otimes F_t^{\otimes(r+1)})\\
 &=\frac1{r+1}\E\operatorname{Sym}
      (Q_j\otimes F_t^{\otimes(r+1)})\\
 &=\frac{r!}{(r+j+1)!}\frac d{dt}\E F_t^{\otimes(r+j+1)}.
\end{aligned}
\]
Thus \eqref{eq:C.6} holds at $j+1$.

For completeness, the degree assertion also propagates.  Rotation derivatives preserve the
finite-dimensional space of polynomials of degree at most $j$ in each orthogonal factor, and so
does $K_0$.  Because $K_0$ is self-adjoint, the orthogonal complement of this invariant space is
also invariant; hence its spectral inverse $K_0^\dagger$ preserves the same polynomial space.
Consequently $g$ has degree at most $j$, while
$\Gamma(g,F_t)$, and therefore $Q_{j+1}$, has degree at most $j+1$.
If $j+1\le2K$, taking $r=0$ in the newly proved instance of \eqref{eq:C.6} and using
\eqref{eq:C.5} yields $\E Q_{j+1}=0$.  Since $j+1\le2K<n$, the preceding kernel calculation then
gives $Q_{j+1}\perp\ker K_0$.  This completes the simultaneous induction.

In particular, for every $j\le2K$, the $r=0$ case gives the moment cancellation
\begin{equation}
 \E Q_j=\frac1{j!}\frac d{dt}\E F_t^{\otimes j}=0
 \qquad(j\le2K),                                          \tag{C.7}\label{eq:C.7}
\end{equation}
which is the identity used in Lemma~\ref{lem:path-tv}.

\subsection{Density differentiation and the exact Hermite score}

Let
\[
 \phi_y(F)=(2\pi\sigma^2)^{-k/2}
 \exp\!\left[-\frac{\norm{y-aLF}^2}{2\sigma^2}\right],
 \qquad p_t(y)=\E\phi_y(F_t).
\]
Thus $\phi_y(F)$ is the density of $N(aLF,\sigma^2I_k)$ evaluated at $y$, and the expectation
in the formula for $p_t$ is over the Haar rotations $(U,V)$.  The Gaussian kernel is positive, so
$p_t(y)>0$ for every $y$.  Throughout this subsection, $\partial_t\log p_t(y)$ means the
derivative in the parameter $t$ with $y$ held fixed.  Only after taking this derivative do we
evaluate at the random point $y=Y_t$; it is not the total derivative of
$t\mapsto\log p_t(Y_t)$.

Differentiating the Gaussian kernel explicitly, put $r_t=y-aLF_t$.  Since
$\dot r_t=-aL\dot F_t$, the ordinary chain rule gives
\begin{align*}
 \frac d{dt}\phi_y(F_t)
 &=\phi_y(F_t)\frac d{dt}\left(-\frac{\norm{r_t}^2}{2\sigma^2}\right)\\
 &=\frac a{\sigma^2}\ip{r_t}{L\dot F_t}\phi_y(F_t).
\end{align*}
On the other hand,
\[
 \nabla_y\phi_y(F_t)=-\frac{r_t}{\sigma^2}\phi_y(F_t).
\]
Consequently,
\begin{equation*}
 \frac d{dt}\phi_y(F_t)
 =-a\,(L\dot F_t)\mathbin{\cdot}\nabla_y\phi_y(F_t).
\end{equation*}
Differentiation under the expectation is justified because $F_t$ and $\dot F_t$ are bounded on
the compact rotation group and every derivative of the Gaussian kernel is a polynomial times a
Gaussian.  Since $Q_1=\dot F_t$, we have
\[
 \dot p_t(y)=-a\,\E\!\left[(LQ_1)\mathbin{\cdot}\nabla_y\phi_y(F_t)\right].
\]

For $j\ge1$, define
\begin{equation}
 R_j(y)=
 \E\!\left[(L^{\otimes j}Q_j):\nabla_y^j\phi_y(F_t)\right]. \tag{C.8}\label{eq:C.8}
\end{equation}
Here ``$:$'' contracts all $j$ indices: if $A$ and $B$ are order-$j$ tensors, then
\[
 A:B=\sum_{i_1,\ldots,i_j=1}^k A_{i_1\cdots i_j}B_{i_1\cdots i_j}.
\]
Thus the preceding display says simply that
\[
 \dot p_t(y)=-aR_1(y).
\]

To establish the recursion relating $R_j$ and $R_{j+1}$, fix
$1\le j\le m-1=2K$.  The simultaneous induction in
Appendix~\ref{app:moment-cancellation} gives $Q_j\perp\ker K_0$.  Set
$g=K_0^\dagger Q_j$; then
\[
 K_0g=K_0K_0^\dagger Q_j=Q_j.
\]
Let $\partial_\alpha$ range over all elementary row and column rotation derivatives from
\eqref{eq:C.1}.
Applying the integration-by-parts identity \eqref{eq:C.3} component by component, and using that $L$ is
fixed, yields
\begin{align*}
 R_j(y)
 &=\E\!\left[(L^{\otimes j}K_0g):\nabla_y^j\phi_y(F_t)\right]\\
 &=\sum_\alpha\E\!\left[
   (L^{\otimes j}\partial_\alpha g):
   \partial_\alpha\{\nabla_y^j\phi_y(F_t)\}\right].
\end{align*}

The second factor is calculated by the chain rule.  For indices $i_1,\ldots,i_j\in[k]$, the
commutation of the $y$-derivatives with $\partial_\alpha$ gives
\begin{align*}
 \partial_\alpha\{\nabla_y^j\phi_y(F_t)\}_{i_1\cdots i_j}
 &=\partial_{i_1}\cdots\partial_{i_j}
   \{\partial_\alpha\phi_y(F_t)\}\\
 &=-a\sum_{i_{j+1}=1}^k
   (L\partial_\alpha F_t)_{i_{j+1}}
   \partial_{i_1}\cdots\partial_{i_j}\partial_{i_{j+1}}
   \phi_y(F_t).
\end{align*}
The first equality holds because $\partial_\alpha$ changes only $(U,V)$, whereas the $\partial_i$ change
only $y$.  The second is the same Gaussian chain-rule calculation as above, with
$\partial_\alpha F_t$ in place of $\dot F_t$.  Substitution and collection of all $j+1$ tensor indices
therefore give
\begin{align*}
 R_j(y)
 &=-a\,\E\!\left[
   \left\{L^{\otimes(j+1)}
   \sum_\alpha(\partial_\alpha g)\otimes(\partial_\alpha F_t)\right\}:
   \nabla_y^{j+1}\phi_y(F_t)\right]\\
 &=-a\,\E\!\left[
   \{L^{\otimes(j+1)}\Gamma(g,F_t)\}:
   \nabla_y^{j+1}\phi_y(F_t)\right].
\end{align*}
By the product rule \eqref{eq:C.2},
$\Gamma(g,F_t)=\sum_\alpha(\partial_\alpha g)\otimes(\partial_\alpha F_t)$, entry by entry.  Moreover,
$\nabla_y^{j+1}\phi_y(F_t)$ is symmetric in its $j+1$ indices, so contracting it against a tensor
depends only on the symmetric part of that tensor.  Using the definition
$Q_{j+1}=\operatorname{Sym}\Gamma(g,F_t)$ now proves
\begin{equation}
\begin{aligned}
 R_j(y)
 &=-a\,\E\!\left[
   (L^{\otimes(j+1)}\operatorname{Sym}\Gamma(g,F_t)):
   \nabla_y^{j+1}\phi_y(F_t)\right]
 =-aR_{j+1}(y).
\end{aligned}
\tag{C.9}\label{eq:C.9}
\end{equation}
Thus each application of rotation integration by parts changes $Q_j$ into $Q_{j+1}$, adds one
copy of $L$, raises the Gaussian derivative order by one, and contributes the scalar factor $-a$.
Starting from $\dot p_t(y)=-aR_1(y)$ and applying \eqref{eq:C.9} successively gives
\[
 \dot p_t(y)=-aR_1(y)=(-a)^2R_2(y)=\cdots=(-a)^mR_m(y).
\]
The recursion is available through $j=m-1=2K$ precisely because the simultaneous induction in
Appendix~\ref{app:moment-cancellation} proves $Q_j\perp\ker K_0$ at every one of those orders.
Since $m=2K+1$ is odd, we have proved
\begin{equation}
 \dot p_t(y)=(-a)^mR_m(y)=-a^mR_m(y).                     \tag{C.10}\label{eq:C.10}
\end{equation}

For a standard Gaussian vector $g\in\R^k$, define the symmetric Hermite tensor $H_m(g)$ by
\begin{equation}
 H_m(g):v^{\otimes m}
 =\left.\frac{d^m}{ds^m}
   \exp\!\left(s\ip vg-\frac{s^2}{2}\norm v^2\right)
  \right|_{s=0}\qquad(v\in\R^k).                          \tag{C.11}\label{eq:C.11}
\end{equation}
Rank-one symmetric tensors span the symmetric tensor space, so this identity determines
$H_m(g)$.  The definition also gives the Gaussian derivative identity used below.
For $z=(y-aLF)/\sigma$ and any $v\in\R^k$, equation \eqref{eq:C.11}, with the substitution
$s\mapsto-s/\sigma$, gives
\begin{align*}
 \frac{\phi_{y+sv}(F)}{\phi_y(F)}
 &=\exp\!\left[-\frac{s}{\sigma}\ip vz
                -\frac{s^2}{2\sigma^2}\norm v^2\right],\\
 (v\mathbin{\cdot}\nabla_y)^m\phi_y(F)
 &=(-1)^m\sigma^{-m}
   \{H_m(z):v^{\otimes m}\}\phi_y(F).
\end{align*}
Since this holds for every $v$ and both sides are symmetric order-$m$ tensors, polarization gives
\begin{equation}
 \nabla_y^m\phi_y(F)
 =(-1)^m\sigma^{-m}
 H_m\!\left(\frac{y-aLF}{\sigma}\right)\phi_y(F).          \tag{C.12}\label{eq:C.12}
\end{equation}
Substituting \eqref{eq:C.12} into \eqref{eq:C.10}, or equivalently into the last expression in \eqref{eq:C.10}, gives
\begin{align*}
 \dot p_t(y)
 &=(-a)^m(-1)^m\sigma^{-m}
   \E\!\left[(L^{\otimes m}Q_m):
   H_m\!\left(\frac{y-aLF_t}{\sigma}\right)\phi_y(F_t)\right]\\
 &=\left(\frac a\sigma\right)^m
   \E\!\left[(L^{\otimes m}Q_m):
   H_m\!\left(\frac{y-aLF_t}{\sigma}\right)\phi_y(F_t)\right].
\end{align*}
The sign in the second line is positive because
$(-a)^m(-1)^m=(-1)^{2m}a^m=a^m$.

Bayes' formula identifies the last expression as a conditional expectation.  Let $W=(U,V)$
denote all latent rotation randomness, and write $F_t(W)$ and $Q_m(W)$ to display their
dependence on $W$.  For
any integrable function $A(W,g)$, Bayes' formula for the model
$Y_t=aLF_t(W)+\sigma G_k$ is
\[
 \E[A(W,G_k)\mid Y_t=y]
 =\frac{1}{p_t(y)}\E_W\!\left[
 A\!\left(W,\frac{y-aLF_t(W)}{\sigma}\right)
 \phi_y(F_t(W))\right].
\]
Indeed, for fixed $W$ and $Y_t=y$, the noise is
$G_k=(y-aLF_t(W))/\sigma$, while $\phi_y(F_t(W))$ is the likelihood of that value of $W$.
Apply this identity with
\[
 A(W,g)=(L^{\otimes m}Q_m(W)):H_m(g).
\]
Dividing the preceding formula for $\dot p_t(y)$ by $p_t(y)$ gives
\[
 \partial_t\log p_t(y)
 =\left(\frac a\sigma\right)^m
 \E\!\left[(L^{\otimes m}Q_m):H_m(G_k)\mid Y_t=y\right].
\]
Evaluating this fixed-$y$ identity at $y=Y_t$ proves
\begin{equation}
 \left.\partial_t\log p_t(y)\right|_{y=Y_t}
 =\left(\frac a\sigma\right)^m
 \E\!\left[(L^{\otimes m}Q_m):H_m(G_k)\mid Y_t\right].    \tag{C.13}\label{eq:C.13}
\end{equation}
Equation~\eqref{eq:C.13} is the exact score identity used in Lemma~\ref{lem:path-tv}.  The vector $G_k$ is
independent of $(U,V)$ before conditioning, but is generally not independent of $Q_m$ after
conditioning on $Y_t$.  The derivation consists entirely of exact identities and introduces no
Taylor remainder.

For a deterministic order-$m$ tensor $T\in(\R^k)^{\otimes m}$, the following second-moment
identity is called the Gaussian Hermite isometry:
\begin{equation}
 \E[(T:H_m(G_k))^2]=m!\norm{\operatorname{Sym}T}^2.        \tag{C.14}\label{eq:C.14}
\end{equation}
Thus the $L_2$ norm of the random scalar $T:H_m(G_k)$ equals $\sqrt{m!}$ times the Frobenius
norm of the symmetric part of $T$.  To verify the identity, use the Gaussian generating function
\begin{equation*}
 \E\exp\!\left(\ip{v}{G_k}-\frac{\norm v^2}{2}\right)
    \exp\!\left(\ip{w}{G_k}-\frac{\norm w^2}{2}\right)
 =\exp(\ip{v}{w})\qquad(v,w\in\R^k).
\end{equation*}
By \eqref{eq:C.11}, the first exponential expands as
$\sum_{r\ge0}(H_r(G_k):v^{\otimes r})/r!$, and similarly for the second.  Equating the terms
homogeneous of degree $m$ in both $v$ and $w$ gives
\begin{equation*}
 \E[(H_m(G_k):v^{\otimes m})(H_m(G_k):w^{\otimes m})]
 =m!\ip{v}{w}^m.
\end{equation*}
Rank-one symmetric tensors span the symmetric tensor space, so linearity and polarization extend
this identity to arbitrary order-$m$ tensors $S,T$:
\begin{equation*}
 \E[(S:H_m(G_k))(T:H_m(G_k))]
 =m!\ip{\operatorname{Sym}S}{\operatorname{Sym}T}.
\end{equation*}
Taking $S=T$ proves \eqref{eq:C.14}.
Conditional Jensen applied to \eqref{eq:C.13}, followed by \eqref{eq:C.14}, proves
\[
 I_t\le(a/\sigma)^{2m}m!\,\E\norm{L^{\otimes m}Q_m}^2,
\]
which is the score-energy estimate used in Lemma~\ref{lem:path-tv}.

\subsection{The Hilbert-valued energy estimate}

A unit tangent vector is a pair of skew-symmetric matrices $(A,B)$ with
$\tfrac12(\norm A_F^2+\norm B_F^2)=1$.  Here the path time $t$ is fixed.  Define the differential
of $F_t$ in the rotation direction $(A,B)$ by
\begin{equation}
\begin{aligned}
 d_{\mathrm{rot}}F_t(A,B)
 &:=\left.\frac d{ds}\left\{\sqrt n\,Ue^{sA}D_t(Ve^{sB})^T\right\}\right|_{s=0}\\
 &=\sqrt n\,U(AD_t-D_tB)V^T.
\end{aligned}                                             \tag{C.15}\label{eq:C.15}
\end{equation}
This is a derivative with respect to the auxiliary rotation parameter $s$, not the time derivative
$\dot F_t=\partial_tF_t$.
If $\Lambda=\max(1,\norm{D_t}_{\op})$, then
\[
 \norm{d_{\mathrm{rot}}F_t(A,B)}_F^2
 \le n\Lambda^2(\norm A_F+\norm B_F)^2
 \le4n\Lambda^2.                                         \tag{C.16}\label{eq:C.16}
\]
Fix $1\le j\le2K$.  Recall that the recursion
$Q_{j+1}=\operatorname{Sym}\Gamma(g,F_t)$ uses the tensor-valued function
\begin{equation*}
 g:=K_0^\dagger Q_j\in\mathcal H^{\otimes j}.
\end{equation*}
Appendix~\ref{app:moment-cancellation} proves $Q_j\perp\ker K_0$, and therefore
$K_0g=K_0K_0^\dagger Q_j=Q_j$.  We now spell out how \eqref{eq:C.16} controls
$\Gamma(g,F_t)$.  Let $\mathcal T$ be the Euclidean space of pairs of skew-symmetric matrices
with norm
\begin{equation*}
 \norm{(A,B)}_{\mathcal T}^2
 :=\frac12(\norm A_F^2+\norm B_F^2).
\end{equation*}
The vectors $(E^{ab},0)$ and $(0,E^{ab})$, indexed respectively by the row and column rotation
directions, form an orthonormal basis $(e_\alpha)_\alpha$ of $\mathcal T$.  At a fixed $(U,V)$,
define two linear maps
\begin{equation*}
 \mathcal A_g:\mathcal T\to\mathcal H^{\otimes j},
 \quad \mathcal A_ge_\alpha=\partial_\alpha g,
 \qquad
 \mathcal B_t:\mathcal T\to\mathcal H,
 \quad \mathcal B_te_\alpha=\partial_\alpha F_t.
\end{equation*}
The second map is exactly $d_{\mathrm{rot}}F_t$, so \eqref{eq:C.16} gives
\begin{equation*}
 \norm{\mathcal B_t}_{\op}^2\le4n\Lambda^2,
 \qquad
 \norm{\mathcal A_g}_{HS}^2=\sum_\alpha\norm{\partial_\alpha g}^2.
\end{equation*}
Identify $\mathcal H^{\otimes j}\otimes\mathcal H$ with the Hilbert--Schmidt maps from
$\mathcal H$ to $\mathcal H^{\otimes j}$ by sending $a\otimes b$ to the map
$x\mapsto a\ip{b}{x}$.  Under this norm-preserving identification,
\begin{equation*}
 \Gamma(g,F_t)
 =\sum_\alpha(\mathcal A_ge_\alpha)\otimes(\mathcal B_te_\alpha)
 \quad\longleftrightarrow\quad
 \mathcal A_g\mathcal B_t^*.
\end{equation*}
Therefore the Hilbert--Schmidt product inequality gives the pointwise bound
\begin{align*}
 \norm{\Gamma(g,F_t)}^2
 &=\norm{\mathcal A_g\mathcal B_t^*}_{HS}^2\\
 &\le\norm{\mathcal A_g}_{HS}^2\norm{\mathcal B_t^*}_{\op}^2
 \le4n\Lambda^2\sum_\alpha\norm{\partial_\alpha g}^2.
\end{align*}
Taking expectation and using the tensor-valued energy identity \eqref{eq:C.3} gives
\begin{equation}
 \E\norm{\Gamma(g,F_t)}^2
 \le C_{\mathrm{en},1}n\Lambda^2\,\E\sum_\alpha\norm{\partial_\alpha g}^2
 =C_{\mathrm{en},1}n\Lambda^2\langle g,K_0g\rangle_{L_2},
 \qquad C_{\mathrm{en},1}:=4.                            \tag{C.17}\label{eq:C.17}
\end{equation}
For this choice of $g$, \eqref{eq:C.4} gives
\[
\begin{aligned}
 \langle g,K_0g\rangle_{L_2}
 &=\langle K_0^\dagger Q_j,Q_j\rangle_{L_2}
 \le\frac1{n-1}\E\norm{Q_j}^2.
\end{aligned}
\]
Since symmetrization is a contraction and $n/(n-1)\le2$, \eqref{eq:C.17} implies
\begin{equation}
 \E\norm{Q_{j+1}}^2
 \le C_{\mathrm{en},2}\Lambda^2\E\norm{Q_j}^2,
 \qquad C_{\mathrm{en},2}\defeq2C_{\mathrm{en},1},       \tag{C.18}\label{eq:C.18}
\end{equation}
where $C_{\mathrm{en},2}$ is one fixed absolute constant.
At the base,
\begin{equation}
 \E\norm{Q_1}^2
 =n\norm{\dot D_t}_F^2
 =n^2\left(\frac1n\tr\dot D_t^2\right).                   \tag{C.19}\label{eq:C.19}
\end{equation}
Set $C_{\mathrm{en}}=\max\{C_{\mathrm{en},1},C_{\mathrm{en},2}\}$.
Iterating \eqref{eq:C.18} proves
\begin{equation}
 \E\norm{Q_m}^2
 \le (C_{\mathrm{en}}\Lambda^2)^{m-1}n^2
       \left(\frac1n\tr\dot D_t^2\right).                \tag{C.18a}\label{eq:C.18a}
\end{equation}

We finish by proving the two-sided rotation form of $Q_m$.  For fixed $R,S\in\mathcal O_n$, define
the orthogonal map $\rho_{R,S}:\mathcal H\to\mathcal H$ by
\begin{equation*}
 \rho_{R,S}(A):=RAS^T.
\end{equation*}
Under matrix vectorization this map is $R\otimes S$.  Also let
$\mathcal L_{R,S}(U,V):=(RU,SV)$ denote left translation on
$\mathcal O_n\times\mathcal O_n$.  The right-rotation derivatives commute with this left
translation: for every tensor-valued function $G$,
\begin{equation*}
 \partial_\alpha(G\circ\mathcal L_{R,S})
 =(\partial_\alpha G)\circ\mathcal L_{R,S}.
\end{equation*}
Indeed, for a row derivative this is the identity
$RUe^{sE^{ab}}=R(Ue^{sE^{ab}})$, and the column case is identical.  It follows from
$K_0=-\sum_\alpha\partial_\alpha^2$ that $K_0$ commutes with pullback by
$\mathcal L_{R,S}$.  Consequently the pullback preserves $\ker K_0$ and every positive
eigenspace of $K_0$.  On these spaces $K_0^\dagger$ acts respectively as zero and as multiplication
by the reciprocal eigenvalue, so it commutes with the pullback as well:
\begin{equation*}
 K_0^\dagger(G\circ\mathcal L_{R,S})
 =(K_0^\dagger G)\circ\mathcal L_{R,S}.
\end{equation*}

We now induct on $j$.  The base tensor satisfies
\begin{equation*}
 Q_1(\mathcal L_{R,S}(U,V))
 =\dot F_t(RU,SV)
 =\rho_{R,S}Q_1(U,V).
\end{equation*}
Suppose
$Q_j\circ\mathcal L_{R,S}=\rho_{R,S}^{\otimes j}Q_j$ and put
$g=K_0^\dagger Q_j$.  The operator $K_0^\dagger$ acts on the variables $(U,V)$ and componentwise
on tensor entries, so every fixed map on the output tensor space commutes with it.  The preceding
commutation identity therefore gives
\begin{equation*}
 g\circ\mathcal L_{R,S}=\rho_{R,S}^{\otimes j}g.
\end{equation*}
Applying the derivative commutation identity to $g$ and to
$F_t\circ\mathcal L_{R,S}=\rho_{R,S}F_t$ therefore yields
\begin{equation*}
 \Gamma(g,F_t)\circ\mathcal L_{R,S}
 =\rho_{R,S}^{\otimes(j+1)}\Gamma(g,F_t).
\end{equation*}
Finally, $\operatorname{Sym}$ commutes with
$\rho_{R,S}^{\otimes(j+1)}$ because the same map $\rho_{R,S}$ acts on every tensor position.
The recursion defining $Q_{j+1}$ is consequently equivariant as well.  By induction,
\begin{equation*}
 Q_m(RU,SV)=\rho_{R,S}^{\otimes m}Q_m(U,V).
\end{equation*}
Evaluate this identity at the base point $(U,V)=(I,I)$, set $q_m:=Q_m(I,I)$, and then rename
$(R,S)$ as $(U,V)$.  Since $\rho_{U,V}$ is orthogonal, this gives
\[
 Q_m(U,V)=(U\otimes V)^{\otimes m}q_m,\qquad
 \norm{q_m}^2=\E\norm{Q_m}^2.
\]
Put $P=L^*L$.  Because $LL^*=I_k$, the operator $P$ is an orthogonal projection, and hence
\[
 \norm{L^{\otimes m}Q_m}^2
 =\langle Q_m,(L^*L)^{\otimes m}Q_m\rangle
 =\norm{(L^*L)^{\otimes m}Q_m}^2.
\]
Lemma~\ref{lem:twirl} therefore applies to this seed and gives
\begin{equation}
 \E\norm{L^{\otimes m}Q_m}^2
 \le (C_{\mathrm{ten}}m)^{C_{\mathrm{ten}}m}
 \left(\frac{k}{n^2}\right)^{(m+1)/2}\E\norm{Q_m}^2.     \tag{C.19a}\label{eq:C.19a}
\end{equation}

\subsection{Fisher--Rao length controls total variation}\label{app:fisher-tv}

Gaussian convolution makes $y\mapsto p_t(y)$ positive and smooth.  For each fixed $y$, the map
$t\mapsto p_t(y)$ is piecewise continuously differentiable, hence absolutely continuous.  Write
$\dot p_t(y)=\frac{d}{dt}p_t(y)$ wherever the derivative exists.  The fundamental theorem of
calculus gives
\[
 p_1(y)-p_0(y)=\int_0^1\dot p_t(y)\,dt.
\]
The triangle inequality, followed by exchanging two nonnegative integrals, therefore gives
\begin{align*}
 \frac12\int|p_1(y)-p_0(y)|\,dy
 &\le \frac12\int\int_0^1|\dot p_t(y)|\,dt\,dy\\
 &=\frac12\int_0^1\left\{\int|\dot p_t(y)|\,dy\right\}dt.
\end{align*}
Define the score along the path by
\[
 s_t(y)=\frac d{dt}\log p_t(y)=\frac{\dot p_t(y)}{p_t(y)}
\]
and let $Y_t$ have density $p_t$.  Then
\[
 I_t=\E[s_t(Y_t)^2]
 =\int s_t(y)^2p_t(y)\,dy
 =\int\frac{\dot p_t(y)^2}{p_t(y)}\,dy.
\]
Since $\int p_t(y)\,dy=1$, Cauchy--Schwarz yields
\begin{align*}
 \int|\dot p_t(y)|\,dy
 &=\int \frac{|\dot p_t(y)|}{\sqrt{p_t(y)}}\sqrt{p_t(y)}\,dy\\
 &\le
 \left(\int\frac{\dot p_t(y)^2}{p_t(y)}\,dy\right)^{1/2}
 \left(\int p_t(y)\,dy\right)^{1/2}
 =\sqrt{I_t}.
\end{align*}
Combining the preceding displays proves
\begin{equation}
 \TV(\mathbb P_0^L,\mathbb P_1^L)
 =\frac12\int|p_1(y)-p_0(y)|\,dy
 \le\frac12\int_0^1\sqrt{I_t}\,dt.                      \tag{C.20}\label{eq:C.20}
\end{equation}
The factor $1/2$ comes from the definition of total variation in \eqref{eq:3.4}.  Without the absolute
value, $\int(p_1-p_0)=1-1=0$.  The integral on the right is the Fisher--Rao length of the path.

Combining the preceding bounds completes the proof of Lemma~\ref{lem:path-tv}.  The path lies in $[0,B]$, so
$\Lambda^2=\max\{1,\norm{D_t}_{\op}^2\}\le B$ because $B\ge1$.  The score estimate following
\eqref{eq:C.14}, followed in order by \eqref{eq:C.19a} and \eqref{eq:C.18a}, gives
\begin{align*}
 \sqrt{I_t}
 &\le \left(\frac a\sigma\right)^m\sqrt{m!}\,
       \left(\E\norm{L^{\otimes m}Q_m}^2\right)^{1/2}\\
 &\le \left(\frac a\sigma\right)^m\sqrt{m!}\,
       (C_{\mathrm{ten}}m)^{C_{\mathrm{ten}}m/2}
       \left(\frac{k}{n^2}\right)^{(m+1)/4}
       (C_{\mathrm{en}}B)^{(m-1)/2}n
       \left(\frac1n\tr\dot D_t^2\right)^{1/2}\\
 &=\mathcal A_m n
       \left(\frac{k}{n^2}\right)^{(K+1)/2}
       \left(\frac1n\tr\dot D_t^2\right)^{1/2},
\end{align*}
where $m=2K+1$, so $(m+1)/4=(K+1)/2$, and $\mathcal A_m$ is the constant defined in
Lemma~\ref{lem:path-tv}.  Substitute this estimate into \eqref{eq:C.20} and use the definition of
$\mathcal L_K$ in \eqref{eq:4.10}.  The result is precisely \eqref{eq:4.23}.

\section{Tail decoder and probability estimates}\label{app:decoder}

This appendix verifies the normalization, row counts, and success probability of
Algorithms~\ref{alg:upper-sketch} and~\ref{alg:upper-decode}.  Runtime is not part of the sketching
model; the ordered cycle sums may
therefore be evaluated directly.

\subsection{Guarantee of the stable-rank certificate}

Work on the simultaneous covariance event \eqref{eq:5.5} and Frobenius event \eqref{eq:5.19}.  For a candidate
$h=h_j$, abbreviate
\[
 C=C_j=A(I-P_j),\quad F=F_h(A),\quad
 \widehat F=\norm{WC}_F,\quad
 \gamma=\begin{cases}
  \sigma_{h+1}(GA),&h<n,\\
  0,&h=n.
 \end{cases}
\]
Frobenius optimality of the true rank-$h$ approximation gives
\begin{equation}
 F\le\norm C_F,\qquad
 (1-\alpha)\norm C_F^2\le\widehat F^2
 \le(1+\alpha)\norm C_F^2.                                \tag{D.1}\label{eq:D.1}
\end{equation}
Because $P_j$ projects onto the top $h$ right singular directions of $GA$,
\[
 \norm{GA(I-P_j)}_{\op}=\sigma_{h+1}(GA)=\gamma
 \quad(h<n),
\]
while both sides are zero when $h=n$.
Moreover, the lower bound in \eqref{eq:D.1} and $\alpha<1/10$ imply
\[
 F^2\le\norm C_F^2\le\frac{\widehat F^2}{1-\alpha}<2\widehat F^2.
\]
The lower half of \eqref{eq:5.5}, restricted to the orthogonal complement of $P_j$, gives
\begin{equation}
 (1-\alpha)\norm C_{\op}^2\le\gamma^2+\beta_jF^2
 \le\gamma^2+2\beta_j\widehat F^2.                        \tag{D.2}\label{eq:D.2}
\end{equation}
If $C=0$, its tail nuclear norm is exactly zero and Algorithm~\ref{alg:upper-decode} bypasses
the normalized-moment calculation.  Hence assume $C\ne0$ in the following stable-rank
conclusion.  If \eqref{eq:5.20} passes, \eqref{eq:D.1}--\eqref{eq:D.2} imply
\[
 \norm C_{\op}^2\le\frac{D_{\mathrm{cert}}(1+\alpha)}{R}\norm C_F^2,
 \qquad\text{and hence}\qquad
 \sr(C)=\frac{\norm C_F^2}{\norm C_{\op}^2}
 \ge\frac{R}{D_{\mathrm{cert}}(1+\alpha)}.                \tag{D.3}\label{eq:D.3}
\]

Conversely, suppose $h\ge R$ and the exact tail has stable rank at least $R/4$.  Then
$a_h^2\le4F^2/R$.  The upper half of \eqref{eq:5.5} gives
\[
 \gamma^2\le(1+\alpha)a_h^2+\beta_jF^2
 \le\left(\frac{4(1+\alpha)}R+\beta_j\right)F^2.
\]
Use $F^2\le\norm C_F^2\le\widehat F^2/(1-\alpha)$, the lower bound in \eqref{eq:D.1}, and
$\beta_j=\alpha^4/h\le\alpha^4/R$.  The left side of \eqref{eq:5.20} is at most
\[
 \left\{\frac{4(1+\alpha)+\alpha^4}{(1-\alpha)^2}
          +\frac{2\alpha^4}{1-\alpha}\right\}
 \frac{\widehat F^2}{R}
 <6\frac{\widehat F^2}{R}.
\]
Thus the fixed choice $D_{\mathrm{cert}}=12$ makes \eqref{eq:5.20} pass.  The $h=0$ case is identical after
using $C=A$, $a_0=\norm A_{\op}$, $F_0=\norm A_F$, and
$\beta_0=c_0\alpha^4/R$: the hypothesis $\sr(A)\ge R/4$ is precisely
$a_0^2\le4F_0^2/R$, so every displayed inequality above remains valid.  This proves both
directions of the certificate.

\subsection{Unbiased cycle moments}

Let $C\ne0$, let $\Sigma=C^TC$, and let $z_1,\ldots,z_s$ be independent $N(0,\Sigma)$ rows.
For distinct indices $a_1,\ldots,a_j$, integrate $z_{a_1}$ first:
\[
 \E_{z_{a_1}}
 \ip{z_{a_j}}{z_{a_1}}\ip{z_{a_1}}{z_{a_2}}
 =z_{a_j}^T\Sigma z_{a_2}.
\]
Integrating $z_{a_2},z_{a_3},\ldots$ successively closes the cycle and leaves
$\tr(\Sigma^j)$.  Every ordered distinct tuple has the same expectation, so
\[
 \E\widehat p_j=\tr(\Sigma^j).
\]
Let $\widehat p_j^{\,\mathrm{inc}}(z_1,\ldots,z_s)$ denote the entire normalized
increasing-cycle statistic of Kong and Valiant, and let $\pi$ be a uniform random permutation of
$\{1,\ldots,s\}$, independent of the samples.  The ordered statistic in \eqref{eq:5.29} satisfies
\[
 \widehat p_j
 =\E_\pi\!\left[
   \widehat p_j^{\,\mathrm{inc}}(z_{\pi(1)},\ldots,z_{\pi(s)})
   \mathrel{\big|}z_1,\ldots,z_s
  \right].
\]
Indeed, averaging the increasing tuples over $\pi$ makes every ordered distinct tuple equally
likely.  Conditional Jensen therefore shows that this symmetrization cannot increase variance,
so the imported cycle-variance bound applies with exactly the normalization in \eqref{eq:5.29}.

Put $p_j=\tr(\Sigma^j)$ and $p_1=\norm C_F^2$.  For
\[
 \lambda_i=\frac{n\sigma_i(C)^2}{p_1}
\]
we have
\begin{equation}
 \mu_j=\frac1n\sum_i\lambda_i^j
 =\frac{n^{j-1}p_j}{p_1^j}.                               \tag{D.4}\label{eq:D.4}
\end{equation}
This proves the ratio in \eqref{eq:5.32}, including every power of $n$.

\subsection{Propagation of raw-moment error through the ratios}

Define the good raw-moment event
\begin{equation}
 \left|\frac{\widehat p_1}{p_1}-1\right|
 \le\frac{\tau}{8K},\qquad
 \left|\frac{\widehat p_j}{p_j}-1\right|\le\frac\tau4
 \quad(2\le j\le K).                                      \tag{D.5}\label{eq:D.5}
\end{equation}
Write the corresponding relative errors as $e_j$.  From \eqref{eq:D.4},
\[
 \frac{\widehat\mu_j}{\mu_j}
 =\frac{1+e_j}{(1+e_1)^j}.
\]
For $|e_1|\le1/(4K)$, the mean-value theorem gives
\[
 |(1+e_1)^{-j}-1|\le2j|e_1|.
\]
Thus \eqref{eq:D.5} implies
\begin{equation}
 |\widehat\mu_j-\mu_j|\le\tau\mu_j
 \qquad(1\le j\le K).                                     \tag{D.6}\label{eq:D.6}
\end{equation}
The zeroth moment is deterministic: $\mu_0=\widehat\mu_0=1$.

Recall that $B_\lambda=D_0n/R$ is the support bound fixed in Step~6.
Because $0\le\lambda_i\le B_\lambda$ and $\mu_1=1$,
\begin{equation}
 \mu_j\le B_\lambda^{j-1}\quad(j\ge1),\qquad
 q=\frac1n\sum_i\sqrt{\lambda_i}
 \ge\frac1{\sqrt{B_\lambda}}.                             \tag{D.7}\label{eq:D.7}
\end{equation}
The second inequality follows from $\sqrt\lambda\ge\lambda/\sqrt{B_\lambda}$ on
$[0,B_\lambda]$.

Let
\[
 q_K=\sqrt{B_\lambda}\sum_{j=0}^Ka_j\frac{\mu_j}{B_\lambda^j}.
\]
The uniform polynomial approximation \eqref{eq:5.26} gives
\[
 |q_K-q|\le\frac{C_{\mathrm{poly}}\sqrt{B_\lambda}}K
 \le\frac\rho4q
\]
when $K\ge4C_{\mathrm{poly}}B_\lambda/\rho$.  On \eqref{eq:D.6},
\eqref{eq:5.26}--\eqref{eq:5.28} and \eqref{eq:D.7} give
\[
\begin{aligned}
 |\widehat q-q_K|
 &\le\sqrt{B_\lambda}\sum_{j=1}^K
       |a_j|\frac{|\widehat\mu_j-\mu_j|}{B_\lambda^j}\\
 &\le\sqrt{B_\lambda}\,K2^{3K}\tau\,B_\lambda^{-1}
 \le\frac\rho4q
\end{aligned}                                             \tag{D.8}\label{eq:D.8}
\]
because the fixed choice $c_{\mathrm{mom}}=1/4$ in \eqref{eq:5.28} makes the last expression at most
$\rho q/4$.  Therefore
$|\widehat q-q|\le(\rho/2)q$.

Also, \eqref{eq:D.5} and the mean-value theorem imply
\[
 \left|\sqrt{\widehat p_1/p_1}-1\right|
 \le |e_1|
 \le\frac{\tau}{8K}
 \le\frac\rho8.
\]
If $e_F=\sqrt{\widehat p_1/p_1}-1$ and
$e_q=\widehat q/q-1$, then the relative error before taking the maximum with zero is
$e_F+e_q+e_Fe_q$.  We have $|e_F|\le\rho/8$, $|e_q|\le\rho/2$, and $0<\rho<1$; hence its
absolute value is at most $11\rho/16<\rho$.  Combining this with \eqref{eq:D.8} proves
\begin{equation}
 |\widehat T(C)-\norm C_{\Sone}|
 \le\rho\norm C_{\Sone}.                                  \tag{D.9}\label{eq:D.9}
\end{equation}
On the good event $\widehat q>0$, so the maximum
with zero in \eqref{eq:5.32a} does not change the value.

\subsection{Row count and simultaneous success}

For $p_1$, a direct Gaussian calculation gives
\[
 \operatorname{Var}(\widehat p_1)
 =\frac{2\tr(\Sigma^2)}s\le\frac{2p_1^2}s.
\]
In fact $\widehat p_1=\norm{YC}_F^2$, so the proved bound \eqref{eq:5.18a}, with $W$ replaced by $Y$,
supplies the first part of \eqref{eq:D.5}.  With $u_1=\tau/(8K)$ and
$\zeta=(100K)^{-1}$, it is enough that
\begin{equation}
 s\ge16u_1^{-2}\log(2/\zeta)
 =1024K^2\tau^{-2}\log(200K).                            \tag{D.10a}\label{eq:D.10a}
\end{equation}
For $j\ge2$, apply Theorem~\ref{ext:moment} with individual failure
probability $\zeta=(100K)^{-1}$ and required relative error $\tau/4$.  Since
\[
 \log f(j)\le C_{\mathrm{row},5}j\log(j+1),\qquad
 \log(1/\tau)\le C_{\mathrm{row},6}(\rho)K,
\]
the first term of \eqref{eq:5.31} is at most the target once
\begin{equation}
 s\ge n^{1-2/j}
 \exp\!\left\{C_{\mathrm{row},1}\log(j+1)
              +\frac{C_{\mathrm{row},2}K}{j}\right\}.    \tag{D.10}\label{eq:D.10}
\end{equation}
For the remaining comparisons, put
\[
 A_j=\frac{4f(j)}{\tau\sqrt\zeta}.
\]
The second and third terms of \eqref{eq:5.31} are at most $\tau/4$ whenever, respectively,
\begin{equation}
 s\ge A_j^2n^{1/2-1/j},\qquad s\ge A_j^2.                \tag{D.11}\label{eq:D.11}
\end{equation}
The definitions of $f,\tau,\zeta$ give fixed constants
$C_{\mathrm{row},7}$ and $C_{\mathrm{row},8}(\rho)$ such that
\begin{equation}
 \log A_j^2\le
 C_{\mathrm{row},7}j\log(j+1)+C_{\mathrm{row},8}(\rho)K. \tag{D.12}\label{eq:D.12}
\end{equation}
For $4\le j\le K$, the logarithm of the ratio between the first right side in \eqref{eq:D.11} and
$n^{1-2/K}$ is at most
\[
 C_{\mathrm{row},7}K\log(K+1)+C_{\mathrm{row},8}(\rho)K
 -\left(\frac12-\frac1K\right)\log n.
\]
The corresponding ratio for the second right side has the still smaller negative term
$-(1-2/K)\log n$.  Choosing the fixed $c_{\mathrm{row}}$ in \eqref{eq:5.34} sufficiently small makes both
ratios tend to zero, because
$K\log(K+1)\le2c_{\mathrm{row}}\log n$ for all sufficiently large $n$.
For $j=2,3$, \eqref{eq:D.12} is only $O_\rho(K)$, while the powers of $n$ in the two right sides of
\eqref{eq:D.11} are at most $n^{1/6}$; these are also smaller than $n^{1-2/K}$ after decreasing
$c_{\mathrm{row}}$ if necessary.  The finitely many cases $K<4$ are absorbed by the fixed
leading constant.

Finally, the logarithm of the right side of \eqref{eq:D.10a} is at most
$C_{\mathrm{row},9}(\rho)K+C_{\mathrm{row},10}\log(K+1)$.  For $K\ge4$ this is eventually
less than $\tfrac12\log n$, whereas $n^{1-2/K}\ge\sqrt n$; the cases $K<4$ are again absorbed
by the leading constant.  Thus the single integer row count \eqref{eq:5.35} satisfies \eqref{eq:D.10}, \eqref{eq:D.10a},
and \eqref{eq:D.11} simultaneously for every $1\le j\le K$.
A union bound over $p_1,p_2,\ldots,p_K$ makes \eqref{eq:D.5} fail with probability at most $1/100$.

\subsection{Correctness of the complete decoder}

Algorithms~\ref{alg:upper-sketch} and~\ref{alg:upper-decode} state the full sketching and decoding
procedure.  We verify here its
zero-residual convention, success probability, and final error.
Gaussian absolute continuity ensures that $YC=0$ if and only if $C=0$, with probability one.
Thus the zero-residual convention prevents division by zero without masking a nonzero residual.

On the simultaneous covariance and Frobenius events, the cutoff is valid by \eqref{eq:D.3}.  Conditional
on the selected cutoff, the independent regression block succeeds with probability at least
$0.99$, and the independent tail block succeeds with probability at least $0.99$.  A union bound
over these four events gives probability at least $0.96$.  On this intersection,
$\norm C_{\Sone}\le\norm A_{\Sone}$ and \eqref{eq:5.17}, \eqref{eq:5.24}, and \eqref{eq:D.9} give, in the certified
branch,
\[
 |\widehat H+\widehat T(C)-\norm A_{\Sone}|
 \le\alpha\norm A_{\Sone}+\rho\norm C_{\Sone}
      +6\alpha\norm A_{\Sone}
 \le(7\alpha+\rho)\norm A_{\Sone}.
\]
In the negligible branch, \eqref{eq:5.15} gives
$\norm C_{\Sone}\le(\eta+6\alpha)\norm A_{\Sone}$.  Since
$0\le\norm A_{\Sone}-\norm{AP}_{\Sone}\le\norm C_{\Sone}$, \eqref{eq:5.24} gives
\[
 |\widehat H-\norm A_{\Sone}|
 \le(7\alpha+\eta)\norm A_{\Sone}.
\]
The choice $\alpha=\rho=\eta=\eps/20$ makes both errors at most $2\eps/5$.

\section{From real sketch dimension to polynomial-length streams}\label{app:streaming}

This appendix proves Corollary~\ref{cor:streaming}.  The argument uses two external transfer
theorems, stated below in the forms needed here, and then verifies their hypotheses for the
non-even-$p$ hard pair.  Throughout the appendix, fix
\[
 p>0,\qquad p\notin\{2,4,6,\ldots\},\qquad 0<\eps<1,
 \qquad \Lambda_p\defeq n^{1/2+1/p},
\]
and put
\[
 N\defeq n^2
\]
for the dimension obtained by vectorizing an $n\times n$ matrix.  Euclidean norm on
$\mathbb R^N$ is therefore Frobenius norm on matrices.

For $q\ge1$, let $\gamma_q$ denote the product discrete Gaussian on $\mathbb Z^N$, normalized as
in Jiang, Liu, and Yu~\cite{JiangLiuYuLinearization}.  Explicitly,
\[
 \gamma_q(x)
 =\frac{\exp(-\pi\norm{x}_2^2/q^2)}
        {\sum_{u\in\mathbb Z^N}\exp(-\pi\norm{u}_2^2/q^2)},
 \qquad x\in\mathbb Z^N.
\]
Thus its coordinates are independent centered one-dimensional discrete Gaussians of scale $q$;
changing between this convention and the covariance convention for a continuous Gaussian changes
only absolute constants.  We use the elementary tail bounds
\begin{equation}
 \Prb_{Z\sim\gamma_q}\{\norm Z_2>C_\delta q\sqrt N\}\le\delta,
 \qquad
 \Prb_{Z\sim\gamma_q}\{\norm Z_2>Cq\sqrt{N\log q}\}\le q^{-100},             \tag{E.1}\label{eq:E.1}
\end{equation}
for every fixed $\delta>0$ and all sufficiently large $q$.  They follow by applying the usual
one-dimensional discrete-Gaussian tail estimate in each coordinate.  For the isotropic
generalized discrete Gaussian $X_\lambda$ used below, the same bounds hold with $q$ replaced by
$\lambda$.  We will also use, for the corresponding $n\times n$ random matrices,
\begin{equation}
 \Prb_{Z\sim\gamma_q}\{\norm Z_{\op}>C_\delta q\sqrt n\}\le\delta,       \tag{E.1a}\label{eq:E.1a}
\end{equation}
and the analogous bound for $X_\lambda$.  Indeed, a one-dimensional discrete Gaussian of scale
$q$ is centered subgaussian with parameter $Cq$.  For fixed unit vectors $u,v$, the independent
sum $u^TZv$ is therefore subgaussian with the same parameter.  A union bound over fixed
$1/4$-nets of the two unit spheres, each of cardinality at most $9^n$, proves
\eqref{eq:E.1a} after increasing $C_\delta$.  The isotropic generalized law has the same product
and subgaussian properties.

\subsection{The two external interfaces}

The first theorem is the promise-problem case of the mollified transfer theorem of Jiang, Liu,
and Yu~\cite[Theorem~6.2.4]{JiangLiuYuLinearization}.  For the stream length, we use the safe
bound obtained directly from their canonical stream construction
\cite[Definitions~6.1, 6.1.1, and~6.2.2]{JiangLiuYuLinearization}.

We spell out the phrase \emph{associated mollified stream distribution}.  Let
$\gamma_q^{\mathrm{tr}}$ be $\gamma_q$ conditioned on
\[
 \norm{x}_2\le C_{\mathrm{tr}}q(\sqrt N+\log q),
\]
where the fixed constant $C_{\mathrm{tr}}$ is large enough that the discarded mass is at most
$q^{-100}$.  For $v\in\mathbb Z^N$, let $\mathsf{can}(v)$ be the canonical unit-update block
that performs $|v_j|$ updates of sign $\operatorname{sgn}(v_j)$ in coordinate $j$.  Given the
target law $\mathcal I$, independently draw
\[
 Y\sim\mathcal I,\qquad
 X_1,\ldots,X_{q^2},Z^{\mathrm{tr}}\sim\gamma_q^{\mathrm{tr}},
\]
put
\[
 \Delta_{\mathrm{fin}}
 =Y+Z^{\mathrm{tr}}-\sum_{i=1}^{q^2}X_i,
\]
and concatenate
\[
 \mathsf{can}(X_1),\ldots,\mathsf{can}(X_{q^2}),
 \mathsf{can}(\Delta_{\mathrm{fin}}).
\]
The final frequency vector is $Y+Z^{\mathrm{tr}}$.  This random unit-update stream is the
associated mollified stream distribution.  In the smoothness condition \eqref{eq:E.2}, by
contrast, $Z\sim\gamma_q$ is untruncated; the transfer theorem accounts for the truncation.
Here and below $C$ is an absolute constant.

\begin{theorem}[Mollified turnstile transfer]\label{ext:mollified-transfer}
Let $q\ge N^{20}$ be sufficiently large.  Let $\mathcal I$ be a probability law on
$\mathbb Z^N$ whose support has Euclidean diameter at most $q^{5/4}$, and let
$f:\mathbb Z^N\to\{0,1,*\}$ be a promise problem.  Suppose
\begin{equation}
 \Prb_{Y\sim\mathcal I,\,Z\sim\gamma_q}\{f(Y+Z)=f(Y)\}\ge1-\delta.           \tag{E.2}\label{eq:E.2}
\end{equation}
If an $S$-bit randomized turnstile algorithm solves the promise problem on the associated
mollified stream distribution with failure probability at most $\delta$, then there are an
integer matrix $A_0\in\mathbb Z^{r\times N}$ and a decoder
$g:A_0(\mathbb Z^N)\to\{0,1\}$ such that
\begin{equation}
 r\le C\frac{S+2}{\log q},\qquad
 \max_{i,j}|(A_0)_{ij}|\le q,                             \tag{E.3}\label{eq:E.3}
\end{equation}
and
\begin{equation}
 \Prb_{Y\sim\mathcal I}
 \{f(Y)\in\{0,1\},\ g(A_0Y)\ne f(Y)\}
 \le16\delta.                                             \tag{E.3a}\label{eq:E.3a}
\end{equation}
\end{theorem}

\paragraph{Stream length.}
Every stream in the mollified distribution consists of unit updates.  If
$\operatorname{supp}\mathcal I$ lies in the Euclidean ball of radius $H$ about the origin, then
every such stream has length at most
\begin{equation}
 C\left[q^3\{N+\sqrt N\log q\}+\sqrt N\,H\right].         \tag{E.4}\label{eq:E.4}
\end{equation}
Indeed, there are $q^2$ prefix blocks
$\mathsf{can}(X_1),\ldots,\mathsf{can}(X_{q^2})$ and the final block
$\mathsf{can}(\Delta_{\mathrm{fin}})$.  The truncated discrete Gaussians satisfy
$\norm{X_i}_2,\norm{Z^{\mathrm{tr}}}_2\le Cq(\sqrt N+\log q)$, and
$\operatorname{length}(\mathsf{can}(v))=\norm v_1\le\sqrt N\norm v_2$.  The triangle
inequality for the final block then gives \eqref{eq:E.4}.

The second interface is the integer-to-real lifting theorem of Gribelyuk, Lin, Woodruff, Yu, and
Zhou~\cite{GribelyukEtAlLifting}, in the parameter-explicit form given by Jiang, Liu, and
Yu~\cite[Theorem~7.1.2]{JiangLiuYuLinearization}.  We state only its binary consequence.

For the remainder of this appendix, fix
\[
 \delta_{\mathrm{lift}}\defeq\frac1{20}.
\]

\begin{theorem}[Integer-to-real Gaussian lifting]\label{ext:integer-lifting}
Let $q\ge(2N/\delta_{\mathrm{lift}})^{30}$ be sufficiently large and let
$A_0\in\mathbb Z^{r\times N}$ have entries bounded by $q$ and $r\le N/4$.
Preprocessing $A_0$ as in~\cite{GribelyukEtAlLifting} produces an integer matrix
$\widetilde A$ with at most $4r$ rows that retains the information in $A_0x$, and a row
coisometry $L$ with the same real row span as $\widetilde A$.

For each $b\in\{0,1\}$, let $Y_b$ be an arbitrary $\mathbb Z^N$-valued random vector.  Let $X$
be an independent generalized discrete Gaussian with scale matrix $\mathsf S$, while $G$ is the
associated independent continuous Gaussian.
Suppose
\begin{equation}
 q^{6/5}\ge\sigma_1(\mathsf S)\ge\sigma_N(\mathsf S)
 \ge Nq\log q.                                            \tag{E.5}\label{eq:E.5}
\end{equation}
Let $\psi:\mathbb R^N\to\{0,1\}$ satisfy
\begin{equation}
 \Prb\{\psi(X+Y_b)=b\}\ge1-\frac{\delta_{\mathrm{lift}}}{3}
 \qquad\text{for every }b\in\{0,1\}.                                  \tag{E.5a}\label{eq:E.5a}
\end{equation}
If there is a decoder $g$ such that, for each $b\in\{0,1\}$,
\begin{equation}
 \Prb\{g(\widetilde A(X+Y_b))=\psi(X+Y_b)\}
 \ge1-\frac{\delta_{\mathrm{lift}}}{3},                                  \tag{E.5b}\label{eq:E.5b}
\end{equation}
then a decoder of $L(G+Y_b)$ recovers $b$ with probability at least
$1-\delta_{\mathrm{lift}}$ under each label.
\end{theorem}

Only three features of Theorem~\ref{ext:integer-lifting} will matter: the row count grows by at
most four, the output real map has orthonormal rows, and the offset law $Y_b$ may be arbitrary and
label-dependent.  We will make all input errors smaller than a fixed sufficiently small absolute
constant, so the numerical error thresholds in the quoted theorem are satisfied.

\subsection{A rounded and mollification-stable hard pair}

We first extract the distributional content of the lower-bound proof.  Fix a sufficiently small
constant $\delta_0>0$, to be chosen below.  The construction in the proof of
Corollary~\ref{cor:noneven} gives, for all sufficiently large $n$, random matrices
\begin{equation}
 F_b=\sqrt n\,UD_bV^T,\qquad b\in\{0,1\},                 \tag{E.6}\label{eq:E.6}
\end{equation}
where $U,V$ are independent Haar orthogonal matrices and $D_b$ is deterministic.  Put
\[
 s_b=\left(\frac1n\sum_{i=1}^n x_i(b)^{p/2}\right)^{1/p},
\]
where $D_b=\operatorname{diag}(\sqrt{x_1(b)},\ldots,\sqrt{x_n(b)})$.
Equations \eqref{eq:A.p6}--\eqref{eq:A.p7} imply, after increasing the lower threshold for $n$,
that
\begin{equation}
 \norm{F_b}_F\le n\sqrt{B},\qquad
 \norm{F_b}_{\Sp{p}}=\Lambda_p s_b,\qquad
 \frac{s_1-s_0}{s_1+s_0}\ge\eps+\eta_\eps,\qquad
 s_1\ge c_{p,\eps}>0.                                    \tag{E.7}\label{eq:E.7}
\end{equation}
Here $B\le2(K+2)^{M_{p,\eps}}$ is the support bound in the non-even-$p$ construction and
$\eta_\eps=(1-\eps)/4$.

Choose a fixed $\tau=\tau(p,\eps,\delta_0)>0$ sufficiently small.  Repeating the observation
comparison in the proof of Corollary~\ref{cor:noneven} with Gaussian coefficient $\tau$ changes
only constants depending on $p,\eps$, and $\delta_0$.  In particular, the bound
\eqref{eq:A.p11} remains valid after changing its fixed constant.  Consequently there is a
finite $B_{p,\eps}$ such that, with
\begin{equation}
 d_{p,\eps}(n)=\left\lfloor\frac{n^2}{(\log n)^{B_{p,\eps}}}\right\rfloor, \tag{E.8}\label{eq:E.8}
\end{equation}
every row coisometry $\mathsf L$ of rank at most $d_{p,\eps}(n)$ satisfies
\begin{equation}
 \TV\bigl(\operatorname{Law}(\mathsf L(F_0+\tau G)),
           \operatorname{Law}(\mathsf L(F_1+\tau G))\bigr)\le\frac14,     \tag{E.9}\label{eq:E.9}
\end{equation}
where $G$ has independent standard Gaussian entries.  This is the same calculation as
\eqref{eq:A.p11} followed by the optimization in \eqref{eq:4.25}--\eqref{eq:4.26}; making the
fixed smoothing coefficient smaller only enlarges the constant in the exponent $B_{p,\eps}$.

We now choose the integer-sketch entry bound $q$ and let
\begin{equation}
 \lambda=C_{\mathrm{lift}}Nq\log q,\qquad
 a=\frac{\lambda}{\tau},                                  \tag{E.10}\label{eq:E.10}
\end{equation}
where the fixed absolute constant $C_{\mathrm{lift}}$ absorbs the normalization convention in
Theorem~\ref{ext:integer-lifting}.  Round the orbit part entrywise:
\begin{equation}
 Y_b=\operatorname{round}(aF_b)\in\mathbb Z^{n\times n},\qquad
 E_b=Y_b-aF_b.                                            \tag{E.11}\label{eq:E.11}
\end{equation}
Then deterministically
\begin{equation}
 \norm{E_b}_F\le\frac n2,\qquad
 \norm{E_b}_{\Sp{p}}\le r_{p,n}\defeq
 \begin{cases}
  \Lambda_p/2,&0<p<2,\\
  n/2,&p\ge2.
 \end{cases}                                                            \tag{E.12}\label{eq:E.12}
\end{equation}
The Schatten bound follows from
$\norm M_{\Sp{p}}\le n^{1/p-1/2}\norm M_F$ for $p<2$ and
$\norm M_{\Sp{p}}\le\norm M_F$ for $p\ge2$.

Let $G_\lambda$ have independent $N(0,\lambda^2)$ entries.  Conditional on $F_b$, the two
Gaussian laws with means $aF_b$ and $Y_b$ and common covariance $\lambda^2I_N$ have relative
entropy $\norm{E_b}_F^2/(2\lambda^2)$.  Pinsker's inequality and averaging over $F_b$ therefore
give, for every row coisometry $\mathsf L$,
\begin{equation}
 \begin{aligned}
 &\TV\bigl(\operatorname{Law}(\mathsf L(G_\lambda+aF_b)),
            \operatorname{Law}(\mathsf L(G_\lambda+Y_b))\bigr)\\
 &\qquad\le
 \E_{F_b}\!\left[
   \TV\bigl(\operatorname{Law}(\mathsf L(G_\lambda+aF_b)\mid F_b),
             \operatorname{Law}(\mathsf L(G_\lambda+Y_b)\mid F_b)\bigr)
 \right]\\
 &\qquad\le\frac{\E_{F_b}\norm{E_b}_F}{2\lambda}
 \le\frac n{4\lambda}.
 \end{aligned}                                                              \tag{E.13}\label{eq:E.13}
\end{equation}
The unrounded pair $G_\lambda+aF_b=a(F_b+\tau G)$ is an invertible rescaling of the pair in
\eqref{eq:E.9}.  Hence, for all large $n$, the triangle inequality gives
\begin{equation}
 \TV\bigl(\operatorname{Law}(\mathsf L(G_\lambda+Y_0)),
           \operatorname{Law}(\mathsf L(G_\lambda+Y_1))\bigr)<\frac13      \tag{E.14}\label{eq:E.14}
\end{equation}
whenever $\rank\mathsf L\le d_{p,\eps}(n)$.  Thus integer rounding preserves the real-sketch hard pair
uniformly over all relevant row spaces.

It remains to construct a smooth integer promise.  Let $X_\lambda$ be the product generalized
discrete Gaussian whose associated continuous law in Theorem~\ref{ext:integer-lifting} is
$G_\lambda$, independent of $Y_b$.  By \eqref{eq:E.1} and \eqref{eq:E.1a}, after fixing a
constant $C_0=C_0(p,\delta_0)$, each of
\begin{equation}
 \norm{X_\lambda}_{\Sp{p}}\le C_0\lambda\Lambda_p,\qquad
 \norm{Z_q}_{\Sp{p}}\le C_0q\Lambda_p,\quad Z_q\sim\gamma_q,            \tag{E.15}\label{eq:E.15}
\end{equation}
fails with probability at most $\delta_0$.  For $p<2$, these estimates follow from the
Frobenius bounds in \eqref{eq:E.1} and
$\norm M_{\Sp{p}}\le n^{1/p-1/2}\norm M_F$.  For $p\ge2$, they follow from
\eqref{eq:E.1a} and $\norm M_{\Sp{p}}\le n^{1/p}\norm M_{\op}$.

We next define deterministic promise thresholds.  If $p\ge1$, put
\begin{equation}
 \begin{aligned}
 M_b&=a\Lambda_p s_b,\\
 R_{\mathrm{base}}&=C_0\lambda\Lambda_p+r_{p,n},
 &R_{\mathrm{mol}}&=C_0q\Lambda_p,\\
 U_0&=M_0+R_{\mathrm{base}}+R_{\mathrm{mol}},
 &L_1&=M_1-R_{\mathrm{base}}-R_{\mathrm{mol}}.
 \end{aligned}                                             \tag{E.16}\label{eq:E.16}
\end{equation}
Equation \eqref{eq:E.7} gives
\[
 (1-\eps)L_1-(1+\eps)U_0
 \ge a\Lambda_p\eta_\eps(s_0+s_1)
      -2(R_{\mathrm{base}}+R_{\mathrm{mol}}).
\]
Because $a=\lambda/\tau$, $s_0+s_1\ge c_{p,\eps}$, and
$\lambda/q=C_{\mathrm{lift}}N\log q$, choosing the fixed
$\tau=\tau(p,\eps,\delta_0)>0$ sufficiently small makes the last expression positive for all
large $n$.

If $0<p<1$, put
\begin{equation}
 \begin{aligned}
 P_b&=a^p\Lambda_p^p s_b^p,\\
 Q_{\mathrm{base}}&=(C_0\lambda\Lambda_p)^p+r_{p,n}^p,
 &Q_{\mathrm{mol}}&=(C_0q\Lambda_p)^p,\\
 U_0&=(P_0+Q_{\mathrm{base}}+Q_{\mathrm{mol}})^{1/p},
 &L_1&=(P_1-Q_{\mathrm{base}}-Q_{\mathrm{mol}})^{1/p}.
 \end{aligned}                                             \tag{E.16a}\label{eq:E.16a}
\end{equation}
To verify that $L_1$ is real and that the thresholds are separated, define
\[
 \theta_{p,\eps}=\frac{1-\eps-\eta_\eps}{1+\eps+\eta_\eps}
 <\frac{1-\eps}{1+\eps}.
\]
Equation \eqref{eq:E.7} gives $s_0\le\theta_{p,\eps}s_1$ and
$s_1\ge c_{p,\eps}$.  Hence
\[
 (1-\eps)^pP_1-(1+\eps)^pP_0
 \ge a^p\Lambda_p^p c_{p,\eps}^p
 \left\{(1-\eps)^p-(1+\eps)^p\theta_{p,\eps}^p\right\}>0.
\]
After the same choice of a sufficiently small fixed $\tau$, the total error satisfies
\[
 \bigl\{(1-\eps)^p+(1+\eps)^p\bigr\}
 (Q_{\mathrm{base}}+Q_{\mathrm{mol}})
 <(1-\eps)^pP_1-(1+\eps)^pP_0,
\]
because $a=\lambda/\tau$ and $q/\lambda=o(1)$.  In particular, $P_1$ exceeds the total error,
so $L_1$ is well defined.
It follows, after raising to the power $p$, that the same strict separation holds in both ranges:
\begin{equation}
 (1-\eps)L_1>(1+\eps)U_0,                                \tag{E.17}\label{eq:E.17}
\end{equation}
where the thresholds are given by \eqref{eq:E.16} for $p\ge1$ and by \eqref{eq:E.16a} for
$0<p<1$.

Define the Schatten--$p$ promise problem
\begin{equation}
 f(T)=
 \begin{cases}
  0,&\norm T_{\Sp{p}}\le U_0,\\
  1,&\norm T_{\Sp{p}}\ge L_1,\\
  *,&\text{otherwise}.
 \end{cases}                                              \tag{E.19}\label{eq:E.19}
\end{equation}
For the later lifting step, fix the following total binary extension of this promise:
\begin{equation}
 \overline f(T)=\mathbf1_{\{\norm T_{\Sp{p}}\ge L_1\}}.
                                                               \tag{E.19a}\label{eq:E.19a}
\end{equation}
Thus $\overline f=f$ whenever $f\in\{0,1\}$; on the gap where $f=*$, the chosen extension is
zero.
For $T_b=X_\lambda+Y_b$, equations \eqref{eq:E.12} and \eqref{eq:E.15} show that, except with
probability $2\delta_0$,
\begin{equation}
 f(T_b)=\overline f(T_b)=b
 \quad\text{and}\quad f(T_b+Z_q)=f(T_b).                  \tag{E.20}\label{eq:E.20}
\end{equation}
For $p\ge1$, this follows from the triangle and reverse-triangle inequalities.  For $0<p<1$,
it follows instead from \eqref{eq:A.p9}, applied first to the base perturbation and then to the
mollification.  This proves the required average smoothness before truncation in both ranges.

Finally, truncate only to obtain bounded support.  A second application of \eqref{eq:E.1} gives an
event $\mathcal E$ with probability at least $1-q^{-100}$ on which
\[
 \norm{X_\lambda}_F\le C\lambda\sqrt{N\log q}.
\]
Since \eqref{eq:E.7} and \eqref{eq:E.11} give
$\norm{Y_b}_F\le a\sqrt{NB}+\sqrt N/2$, the conditional laws
\begin{equation}
 \mathcal I_b=\operatorname{Law}(T_b\mid\mathcal E),\qquad
 \mathcal I=\tfrac12\mathcal I_0+\tfrac12\mathcal I_1                  \tag{E.21}\label{eq:E.21}
\end{equation}
are supported in the Euclidean ball of radius
\[
 H\defeq C_{p,\eps}\lambda\sqrt{N\log q}\,(\log n)^{M_{p,\eps}/2}.
\]
If $q\ge N^{C_*(p,\eps)}$ for a sufficiently large constant $C_*(p,\eps)$, then
\eqref{eq:E.10} gives
\begin{equation}
 \operatorname{diam}(\operatorname{supp}\mathcal I)
 \le2H
 \le C_{p,\eps}N^{3/2}q(\log q)^{3/2}(\log n)^{M_{p,\eps}/2}
 \le q^{5/4}.                                             \tag{E.22}\label{eq:E.22}
\end{equation}
Conditioning changes the probabilities in \eqref{eq:E.20} by at most $O(q^{-100})$.
Consequently, after making $\delta_0$ a sufficiently small fixed constant, $f$ is
$3\delta_0$-smooth on average on $\mathcal I$ with respect to $\gamma_q$, and
$f(T_b)=b$ with probability at least $1-3\delta_0$ under each truncated label.

\subsection{Proof of the streaming corollary}

Assume that an $S$-bit algorithm as in Corollary~\ref{cor:streaming} exists.  A constant number of
independent repetitions and the median reduce its failure probability below $\delta_0$ using
$O(S)$ bits.  Put
\begin{equation}
 q=\left\lfloor W^{1/10}\right\rfloor.                   \tag{E.23}\label{eq:E.23}
\end{equation}
If $W\ge n^{C_{\mathrm{str}}(p,\eps)}$ and $C_{\mathrm{str}}(p,\eps)$ is sufficiently large,
then
\begin{equation}
 q\ge\max\{N^{20},(2N/\delta_{\mathrm{lift}})^{30},N^{C_*(p,\eps)}\},\qquad
 \lambda\le q^{6/5}.                                     \tag{E.24}\label{eq:E.24}
\end{equation}
The lower inequality in \eqref{eq:E.5} follows directly from \eqref{eq:E.10}; hence the entire
lifting window holds.  The same choice makes the safe stream-length bound \eqref{eq:E.4} at most
$W$.  Indeed, \eqref{eq:E.22} gives $H\le q^{5/4}$, and therefore
\begin{equation}
 \begin{aligned}
 C\left[q^3\{N+\sqrt N\log q\}+\sqrt N\,H\right]
 &\le C\left[W^{3/10}\{n^2+n\log W\}+nW^{1/8}\right]\\
 &\le W,
 \end{aligned}                                            \tag{E.25}\label{eq:E.25}
\end{equation}
once $C_{\mathrm{str}}(p,\eps)$ is sufficiently large.

By \eqref{eq:E.17}, a $(1\pm\eps)$ Schatten--$p$ estimate solves the promise problem
\eqref{eq:E.19}: choose any threshold strictly between $(1+\eps)U_0$ and $(1-\eps)L_1$.
Thus the amplified streaming algorithm solves the promise problem on the mollified stream
distribution associated with \eqref{eq:E.21}.  Theorem~\ref{ext:mollified-transfer} supplies an
integer sketch $A_0$ with
\begin{equation}
 r\le C\frac{S+2}{\log q}                                 \tag{E.26}\label{eq:E.26}
\end{equation}
and entries bounded by $q$.  Apply that theorem with error parameter $3\delta_0$.  Its decoder
fails to solve the promise $f$ on the mixture \eqref{eq:E.21} with probability at most
$48\delta_0$.  Because the two labels have equal weight, its promise-problem failure probability
under either individual truncated label is at most $96\delta_0$.  Moreover,
\eqref{eq:E.20} and the conditioning estimate show that
$f(T_b)=\overline f(T_b)=b$ fails with probability at most $3\delta_0$ under either truncated
label.  On the complementary event, promise correctness is exactly correctness for the total
function $\overline f$.  Thus the same decoder computes $\overline f(T_b)$ with failure
probability at most $99\delta_0$ under each truncated label.

The preprocessing in Theorem~\ref{ext:integer-lifting} retains $A_0x$, so this decoder can be run
from $\widetilde A x$.  Removing the conditioning event $\mathcal E$ costs at most $q^{-100}$.
Choose $\delta_0$ so that
\[
 \max\{2\delta_0,99\delta_0+q^{-100}\}
 <\frac{\delta_{\mathrm{lift}}}{3}
\]
for all sufficiently large $q$.  Equations \eqref{eq:E.20} and the preceding decoder bound then
verify \eqref{eq:E.5a}--\eqref{eq:E.5b}, with $\psi=\overline f$, for the unconditioned pair
$X_\lambda+Y_b$.

If $r>N/4$, then \eqref{eq:E.26} already implies
$S=\Omega(N\log q)$, which is stronger than the desired result.  Otherwise apply
Theorem~\ref{ext:integer-lifting}.  It produces a row coisometry with at most $4r$ rows that
distinguishes the real laws $G_\lambda+Y_0$ and $G_\lambda+Y_1$ with success probability greater
than $2/3$.  This contradicts \eqref{eq:E.14} whenever $4r\le d_{p,\eps}(n)$, because equal-prior
binary testing from two laws at total variation below $1/3$ has success probability below $2/3$.
Therefore
\begin{equation}
 r=\Omega_{p,\eps}\!\left(\frac{n^2}{(\log n)^{B_{p,\eps}}}\right).      \tag{E.27}\label{eq:E.27}
\end{equation}
Combining \eqref{eq:E.23}, \eqref{eq:E.26}, and \eqref{eq:E.27}, and absorbing the additive
constant in \eqref{eq:E.26} for large $n$, gives
\[
 S=\Omega_{p,\eps}\!\left(
   \frac{n^2\log q}{(\log n)^{B_{p,\eps}}}
  \right)
 =\Omega_{p,\eps}\!\left(
   \frac{n^2\log W}{(\log n)^{B_{p,\eps}}}
  \right).
\]
This proves Corollary~\ref{cor:streaming}.

\section{Upper bounds for finite non-even Schatten exponents}
\label{app:noneven-upper}

This appendix proves Corollary~\ref{cor:noneven-upper}.  The construction uses the same three
ingredients as the nuclear-norm upper bound---a multiscale spectral cutoff, a regularized
Gaussian row sketch, and unbiased estimates of integer spectral moments---but makes two changes.
First, the leading singular-value contribution is read directly from the singular values of the
row sketch.  Second, the tail polynomial approximates $x^{p/2}$ rather than $\sqrt x$.  Estimating
the exact spectral head and a near-optimal residual separately avoids any need to assert
additivity between the two projected matrices when $p\ne1$.

Throughout the appendix, fix
\[
 0<p<\infty,\qquad p\notin\{2,4,6,\ldots\},\qquad 0<\eps<1,
 \qquad \vartheta=\frac p2.
\]
All constants below may depend on $p$ and $\eps$.  For a matrix $A$ with decreasing singular
values, write
\[
 T_{p,h}(A)=\left(\sum_{i>h}\sigma_i(A)^p\right)^{1/p},
 \qquad F_h(A)^2=\sum_{i>h}\sigma_i(A)^2.
\]

\subsection{A power-dependent multiscale cutoff}

\begin{lemma}[Schatten--$p$ head--tail dichotomy]\label{lem:F-dichotomy}
Fix $1\le R\le n$ and $0<\eta<1$.  Put
\begin{equation}
 h_j=\min\{jR,n\},\qquad
 J=\left\lceil 2+\log_2\!\left(
       \eta^{-1}(n/R)^{(1/p-1/2)_+}\right)\right\rceil .       \tag{F.1}\label{eq:F.1}
\end{equation}
For at least one $0\le j\le J$, either
\[
 T_{p,h_j}(A)\le\eta\norm A_{\Sp{p}},
\]
or $A-A_{h_j}$ has stable rank at least $R/4$.
\end{lemma}

\begin{proof}
As in Lemma~\ref{lem:dichotomy}, suppose that every inspected tail is nonzero and has stable
rank below $R/4$; no inspected cutoff can then equal $n$.  Set
$b_j=\sigma_{jR+1}(A)$.  The stable-rank inequalities give
\[
 \sum_{i>jR}\sigma_i(A)^2<\frac R4b_j^2,
 \qquad b_{j+1}<\frac12b_j,
\]
and hence $b_j\le2^{-(j-1)}b_1$ for $j\ge1$.

If $0<p<2$, comparison of the $\ell_p$ and $\ell_2$ norms gives
\[
 T_{p,jR}(A)
 \le n^{1/p-1/2}\left(\sum_{i>jR}\sigma_i(A)^2\right)^{1/2}
 <\frac{n^{1/p-1/2}\sqrt R}{2}b_j.
\]
The first $R$ singular values are at least $b_1$, so
$\norm A_{\Sp{p}}\ge R^{1/p}b_1$.  Therefore
\begin{equation*}
 \frac{T_{p,jR}(A)}{\norm A_{\Sp{p}}}
 \le (n/R)^{1/p-1/2}2^{-j}.
\end{equation*}
If $p>2$, every singular value in the $j$th tail is at most $b_j$, and consequently
\[
 T_{p,jR}(A)^p
 \le b_j^{p-2}\sum_{i>jR}\sigma_i(A)^2
 <\frac R4b_j^p,
 \qquad \norm A_{\Sp{p}}^p\ge Rb_1^p.
\]
Thus the ratio of the two norms is at most
$4^{-1/p}2^{-(j-1)}$.  In both ranges, the choice \eqref{eq:F.1} makes the last inspected tail
$\eta$-negligible.  This proves the lemma.
\end{proof}

\subsection{A regularized row sketch estimates the exact head and a near-optimal tail}

The next lemma is the main replacement for Steps~2--3 and~5 of the nuclear-norm upper bound.
The head statistic in the lemma concerns the exact leading singular values of $A$, not the norm
of the projected matrix $AP_h$.

\begin{lemma}[Simultaneous head and residual comparison]\label{lem:F-split}
Fix $0<\kappa<1/10$, and let $h_0,\ldots,h_J$ be the grid in
Lemma~\ref{lem:F-dichotomy}, with $H=h_J$.  There is a Gaussian matrix
$G\in\R^{s_G\times n}$ with independent $N(0,1/s_G)$ entries, where
\begin{equation}
 s_G\le C_{p,\kappa}
 \begin{cases}
   H+\log(J+1),&0<p<2,\\
   n^{1-2/p}H^{2/p}+\log(J+1),&p>2,
 \end{cases}                                                   \tag{F.2}\label{eq:F.2}
\end{equation}
such that the following holds with probability at least $0.99$, simultaneously for every
candidate $h=h_j$.  Let $Q_h$ contain the top $h$ right singular vectors of $GA$, put
$P_h=Q_hQ_h^T$ and $C_h=A(I-P_h)$, and define
\[
 \mathcal H_h=\sum_{i\le h}\sigma_i(A)^p,\qquad
 \widehat{\mathcal H}_h=\sum_{i\le h}\sigma_i(GA)^p,\qquad
 \mathcal T_h=\sum_{i>h}\sigma_i(A)^p.
\]
Then
\begin{align}
 \left|\widehat{\mathcal H}_h-\mathcal H_h\right|
 &\le\kappa\norm A_{\Sp{p}}^p,                           \tag{F.3}\label{eq:F.3}\\
 \mathcal T_h
 \le\norm{C_h}_{\Sp{p}}^p
 &\le\mathcal T_h+\kappa\norm A_{\Sp{p}}^p.             \tag{F.4}\label{eq:F.4}
\end{align}
Moreover, after adding the independent Frobenius block used in \eqref{eq:5.19}, the same
Loewner event supports the stable-rank certificate of \eqref{eq:5.20}, with its multiplicative
parameter and fixed numerical constant adjusted as in the proof below.
\end{lemma}

\begin{proof}
We give the parameter calculation because it is where the two ranges of $p$ diverge.  Put
$\bar h=\max\{h,R\}$ and choose a sufficiently small constant $c_{p,\kappa}>0$.  Set
\begin{equation}
 \beta_h=c_{p,\kappa}
 \begin{cases}
   \bar h^{-1},&0<p<2,\\
   n^{2/p-1}\bar h^{-2/p},&p>2.
 \end{cases}                                                   \tag{F.5}\label{eq:F.5}
\end{equation}
For $h>0$, use $F_h=F_h(A)$; for $h=0$, put $F_0=\norm A_F$.
Apply Theorem~\ref{ext:covariance} as in \eqref{eq:5.4}--\eqref{eq:5.6}, with a fixed
multiplicative parameter $\xi=\xi(p,\kappa)>0$ small enough below.  A union bound gives
simultaneously
\begin{equation}
 (1-\xi)A^TA-\beta_hF_h^2I
 \preceq A^TG^TGA
 \preceq(1+\xi)A^TA+\beta_hF_h^2I.                       \tag{F.6}\label{eq:F.6}
\end{equation}
Indeed, the effective dimension needed for candidate $h$ is at most
\[
 h+\frac{\xi}{\beta_h}
 \le C_{p,\kappa}
 \begin{cases}
   \bar h,&0<p<2,\\
   n^{1-2/p}\bar h^{2/p},&p>2.
 \end{cases}
\]
The second expression dominates $\bar h$ because $\bar h\le n$.  Maximizing over the grid proves
the row count \eqref{eq:F.2}.

Let $e_h=\sqrt{\beta_h}F_h$.  Weyl monotonicity applied to \eqref{eq:F.6}, followed by the
elementary inequality
\[
 (u+v)^p\le(1+\kappa/20)u^p+C_{p,\kappa}v^p
 \qquad(u,v\ge0),                                        \tag{F.7}\label{eq:F.7}
\]
shows, after decreasing $\xi$, that
\begin{equation}
 \left|\sigma_i(GA)^p-\sigma_i(A)^p\right|
 \le\frac\kappa4\sigma_i(A)^p+C_{p,\kappa}e_h^p.         \tag{F.8}\label{eq:F.8}
\end{equation}
For $0<p<2$, putting $a_h=\sigma_{h+1}(A)$ gives
\begin{equation}
 \norm A_{\Sp{p}}^p
 \ge h a_h^p+F_h^2a_h^{p-2}
 \ge d_p h^{1-p/2}F_h^p,                                \tag{F.9}\label{eq:F.9}
\end{equation}
where $d_p>0$ depends only on $p$; the case $a_h=0$ is immediate, and the last inequality follows
by minimizing over $a_h>0$.
For $p>2$, comparison of finite-dimensional $\ell_2$ and $\ell_p$ norms gives
\begin{equation}
 F_h^p\le n^{p/2-1}\mathcal T_h
 \le n^{p/2-1}\norm A_{\Sp{p}}^p.                       \tag{F.10}\label{eq:F.10}
\end{equation}
Equations \eqref{eq:F.5}, \eqref{eq:F.9}, and \eqref{eq:F.10} imply, respectively,
\[
 h e_h^p\le C_pc_{p,\kappa}^{p/2}\norm A_{\Sp{p}}^p
 \quad(0<p<2),
 \qquad
 h e_h^p\le c_{p,\kappa}^{p/2}\norm A_{\Sp{p}}^p
 \quad(p>2).
\]
Taking $c_{p,\kappa}$ sufficiently small and summing \eqref{eq:F.8} over $i\le h$ proves
\eqref{eq:F.3}.

We next compare the residual spectrum with the exact tail.  Because $P_h$ is the top right
singular projector of $GA$, the lower half of \eqref{eq:F.6} gives
\[
 \sigma_i(C_h)
 \le\frac{\sigma_i(GA(I-P_h))+e_h}{\sqrt{1-\xi}}.
\]
If $P_*$ is the exact top-$h$ right singular projector of $A$, individual
Eckart--Young optimality and the upper half of \eqref{eq:F.6} give
\[
 \sigma_i(GA(I-P_h))=\sigma_{h+i}(GA)
 \le\sigma_i(GA(I-P_*))
 \le\sqrt{1+\xi}\,\sigma_{h+i}(A)+e_h.
\]
Consequently, for every $i\ge1$,
\begin{equation}
 \sigma_i(C_h)
 \le(1+C\xi)\sigma_{h+i}(A)+Ce_h,                        \tag{F.11}\label{eq:F.11}
\end{equation}
with zero padding beyond index $n$.  Let
$m=\min\{n,\lceil C_{p,\kappa}h\rceil\}$, where the fixed constant is large enough that
$h/m\le\kappa/20$ whenever $m<n$; the case $h=0$ is exact.  Apply \eqref{eq:F.7} to the first
$m$ singular values in \eqref{eq:F.11}.  For the rest, right multiplication by a projection gives
$\sigma_i(C_h)\le\sigma_i(A)$.  Since the sequence $\sigma_i(A)^p$ is decreasing,
\[
 \sum_{i=m+1}^{m+h}\sigma_i(A)^p
 \le\frac hm\sum_{i=h+1}^{h+m}\sigma_i(A)^p.
\]
Consequently
\begin{equation}
 \norm{C_h}_{\Sp{p}}^p
 \le(1+\kappa/2)\mathcal T_h+C_{p,\kappa}m e_h^p.        \tag{F.12}\label{eq:F.12}
\end{equation}
The bounds following \eqref{eq:F.10}, with $m\le C_{p,\kappa}h$, make the last term at most
$(\kappa/2)\norm A_{\Sp{p}}^p$ after one final decrease of $c_{p,\kappa}$.
This proves the upper half of \eqref{eq:F.4}.  The lower half is the Eckart--Young theorem:
$AP_h$ has rank at most $h$, so \eqref{eq:3.2a} gives
$\sigma_i(C_h)\ge\sigma_{h+i}(A)$ for every $i$.  Summing the $p$th powers proves the claim,
including for the Schatten quasi-norm when $0<p<1$.

Finally, \eqref{eq:F.5} always has $\beta_h\le c_{p,\kappa}/R$.  Thus the proof of the certificate
in \eqref{eq:D.1}--\eqref{eq:D.3} applies verbatim after changing fixed constants.  In the
converse direction, an exact tail of stable rank at least $R/4$ passes the certificate because
$\beta_h\le c_{p,\kappa}/R$.  This completes the proof.
\end{proof}

\subsection{Fractional-power estimation on a certified residual}

We use an elementary polynomial approximation whose rate is deliberately conservative but
sufficient for a polylogarithmic saving.

\begin{lemma}[Bernstein approximation to a fixed power]\label{lem:F-poly}
Fix $\vartheta>0$.  For every integer $K\ge1$, there is a degree-$K$ polynomial
$P_{K,\vartheta}(u)=\sum_{j=0}^Ka_ju^j$ such that
\begin{equation}
 \max_j|a_j|\le4^K,
 \qquad
 \sup_{0\le u\le1}|P_{K,\vartheta}(u)-u^\vartheta|
 \le C_\vartheta K^{-\nu_\vartheta},
 \qquad
 \nu_\vartheta=\frac12\min\{1,\vartheta\}.             \tag{F.13}\label{eq:F.13}
\end{equation}
\end{lemma}

\begin{proof}
Take the Bernstein polynomial
\[
 P_{K,\vartheta}(u)=\sum_{\ell=0}^K
  (\ell/K)^\vartheta\binom K\ell u^\ell(1-u)^{K-\ell}.
\]
If $S\sim\operatorname{Bin}(K,u)$, this polynomial equals
$\E(S/K)^\vartheta$.  When $0<\vartheta\le1$, the inequality
$|x^\vartheta-y^\vartheta|\le|x-y|^\vartheta$ and Lyapunov's inequality give error at most
$\{u(1-u)/K\}^{\vartheta/2}\le(4K)^{-\vartheta/2}$.  When $\vartheta>1$, the function is
$\vartheta$-Lipschitz on $[0,1]$, giving error at most $\vartheta/(2\sqrt K)$.
After expanding $(1-u)^{K-\ell}$, the absolute value of the coefficient of $u^j$ is at most
\[
 \sum_{\ell\le j}\binom K\ell\binom{K-\ell}{j-\ell}
 =\binom Kj\sum_{\ell\le j}\binom j\ell
 =\binom Kj2^j\le4^K.
\]
This proves \eqref{eq:F.13}.
\end{proof}

\begin{lemma}[Certified Schatten--$p$ tail estimate]\label{lem:F-tail}
Let $C\ne0$ be fixed and suppose
\begin{equation}
 \lambda_i=\frac{n\sigma_i(C)^2}{\norm C_F^2},\qquad
 \frac1n\sum_i\lambda_i=1,\qquad 0\le\lambda_i\le B.   \tag{F.14}\label{eq:F.14}
\end{equation}
For every fixed relative accuracy $0<\rho<1$, there is a constant $c_{p,\rho}>0$ such that,
for every integer
\[
 4\le K\le c_{p,\rho}\frac{\log n}{\log\log(e n)},
\]
there is a matrix $Y\in\R^{s_T\times n}$ with independent $N(0,1/s_T)$ entries whose row
count satisfies
\begin{equation}
 s_T\le C_{p,\rho}K^{C_{p,\rho}}n^{1-2/K},               \tag{F.15}\label{eq:F.15}
\end{equation}
and for which $YC$ gives a relative-$\rho$ estimate of $\norm C_{\Sp{p}}^p$ with failure probability at most
$0.01$, provided
\begin{equation}
 B\le c_{p,\rho}
 \begin{cases}
   K^{p/4},&0<p<2,\\
   K^{1/p},&p>2.
 \end{cases}                                             \tag{F.16}\label{eq:F.16}
\end{equation}
Here $Y$ is independent of all randomness used to choose $C$.
\end{lemma}

\begin{proof}
Put
\[
 q_p=\frac1n\sum_i\lambda_i^\vartheta,\qquad \vartheta=p/2.
\]
Then
\begin{equation}
 \norm C_{\Sp{p}}^p
 =n^{1-p/2}\norm C_F^p q_p.                              \tag{F.17}\label{eq:F.17}
\end{equation}
If $0<p<2$, the pointwise inequality
$x^\vartheta\ge B^{\vartheta-1}x$ on $[0,B]$ gives
$q_p\ge B^{\vartheta-1}$.  If $p>2$, Jensen gives $q_p\ge1$.

Scale the polynomial in Lemma~\ref{lem:F-poly} from $[0,1]$ to $[0,B]$.  Its relative bias for
$q_p$ is at most
\begin{equation}
 C_p
 \begin{cases}
   BK^{-p/4},&0<p<2,\\
   B^{p/2}K^{-1/2},&p>2.
 \end{cases}                                             \tag{F.18}\label{eq:F.18}
\end{equation}
Thus \eqref{eq:F.16}, with a sufficiently small fixed leading constant, makes the bias at most
$\rho/4$.

Use the injective-cycle statistics \eqref{eq:5.29} to estimate
$p_j=\sum_i\sigma_i(C)^{2j}$ and form the normalized moment estimates
\[
 \widehat\mu_j=\frac{n^{j-1}\widehat p_j}{\widehat p_1^j},
 \qquad
 \widehat q_p=B^\vartheta\sum_{j=0}^Ka_j\frac{\widehat\mu_j}{B^j}.
\]
As in \eqref{eq:D.4}--\eqref{eq:D.6}, relative accuracy $\tau$ in the raw moments gives relative
accuracy $O(K\tau)$ in the normalized moments.  Since
$\mu_j\le B^{j-1}$ and $|a_j|\le4^K$, it is enough to take
\[
 \tau=\frac{c_{p,\rho}}{K^2 4^K B^{(p/2-1)_+}}.
\]
Then the stochastic error is at most $(\rho/4)q_p$.  Notice that
$\log(1/\tau)=O_{p,\rho}(K+\log B)=O_{p,\rho}(K)$ under \eqref{eq:F.16}.
The variance bound in Theorem~\ref{ext:moment} and the row-count calculation
\eqref{eq:D.10}--\eqref{eq:D.12} therefore give \eqref{eq:F.15}, after changing only the fixed
constants.  The degree-one statistic estimates $\norm C_F^2$ at the same row count.

Finally return the nonnegative estimate
\[
 \widehat{\mathcal T}_p(C)
 =\max\left\{0,
 n^{1-p/2}\widehat p_1^{p/2}\widehat q_p\right\}.
\]
On the simultaneous moment event, \eqref{eq:F.17} and the preceding bounds give relative error
at most $\rho$ after decreasing the fixed internal accuracies.  A union bound over the $K$
moments gives failure probability at most $0.01$.
\end{proof}

\subsection{Sketch, correctness, and optimization}

\begin{proof}[Proof of Corollary~\ref{cor:noneven-upper}]
It suffices to treat $0<\eps<1/2$.  Choose fixed internal accuracies
$\kappa,\eta,\rho>0$, depending only on $p$ and $\eps$, sufficiently small that a relative error
$O_p(\kappa+\eta^p+\rho)$ in the $p$th power yields a $(1\pm\eps)$ estimate after taking the
$p$th root.

For all sufficiently large $n$, set
\begin{equation}
 K=\left\lfloor c_{p,\eps}\frac{\log n}{\log\log n}\right\rfloor,
 \qquad
 b=c'_{p,\eps}
 \begin{cases}
   K^{p/4},&0<p<2,\\
   K^{1/p},&p>2,
 \end{cases}
 \qquad R=\lceil n/b\rceil,                              \tag{F.19}\label{eq:F.19}
\end{equation}
where the two fixed leading constants are chosen sufficiently small.  Use the grid
\eqref{eq:F.1}, and put $H=h_J$.  Independently draw $G$ as in Lemma~\ref{lem:F-split}, a
Frobenius certification matrix $W$ with
$O_{p,\eps}(\log(J+1))$ rows, and $Y$ as in Lemma~\ref{lem:F-tail}.  The fixed linear sketch
stores
\begin{equation}
                         \boxed{GA,\qquad WA,\qquad YA.}  \tag{F.20}\label{eq:F.20}
\end{equation}

The decoder forms $Q_j,P_j,C_j$ for every candidate exactly as in Steps~2--4 of
Algorithm~\ref{alg:upper-decode}.  Writing
$\widehat F_j^2=\norm{WA(I-P_j)}_F^2$ and
$\gamma_j=\sigma_{h_j+1}(GA)$, with the usual zero convention at $h_j=n$, it tests
\begin{equation*}
 \frac{\gamma_j^2+2\beta_j\widehat F_j^2}{1-\xi}
 \le\frac{D_{p,\kappa}\widehat F_j^2}{R},
\end{equation*}
where $\beta_j$ is given by \eqref{eq:F.5} and $D_{p,\kappa}$ is a sufficiently large fixed
constant.
If a candidate passes, select the first one and call the branch certified.  If none passes,
select $h=H$ and call the branch negligible.  In either case compute the exact-head statistic
from the stored row sketch,
\[
 \widehat{\mathcal H}_h=\sum_{i\le h}\sigma_i(GA)^p.
\]
In the negligible branch return $\widehat{\mathcal H}_h^{1/p}$.  In the certified branch put
$B_*=D_{p,\kappa}(1+\xi)n/R$, form $YC_h=YA(I-P_h)$, and use
Lemma~\ref{lem:F-tail} with the known support bound $B_*$ to obtain
$\widehat{\mathcal T}_p(C_h)$, and return
\[
 \left\{\widehat{\mathcal H}_h+\widehat{\mathcal T}_p(C_h)\right\}^{1/p}.
\]
The zero-residual convention is handled as in Appendix~\ref{app:decoder}.

We verify correctness on the simultaneous covariance, certification, and selected-tail events.
If the branch is certified, the certificate gives
$\sr(C_h)\ge R/\{D_{p,\kappa}(1+\xi)\}$, so the normalized spectrum \eqref{eq:F.14} is
bounded by $B_*\le C_{p,\eps}b$.  The choice of the leading constant in \eqref{eq:F.19}
makes \eqref{eq:F.16} hold.  Equations \eqref{eq:F.3}--\eqref{eq:F.4} and
Lemma~\ref{lem:F-tail} then show that the returned quantity before taking the $p$th root differs
from
\[
 \mathcal H_h+\mathcal T_h=\norm A_{\Sp{p}}^p
\]
by at most $O_p(\kappa+\rho)\norm A_{\Sp{p}}^p$.

If no candidate is certified, the stable-rank alternative in Lemma~\ref{lem:F-dichotomy} would
have produced a candidate that passes.  Hence some inspected exact tail is $\eta$-negligible,
and monotonicity gives $T_{p,H}(A)\le\eta\norm A_{\Sp{p}}$.  Equation \eqref{eq:F.3} therefore
shows that omitting the tail changes the $p$th power by at most
$(\kappa+\eta^p)\norm A_{\Sp{p}}^p$.  Our fixed choices of the internal accuracies prove the
$(1\pm\eps)$ guarantee.  The three relevant good events---covariance, simultaneous Frobenius
certification, and the selected tail estimate---have joint probability greater
than $2/3$ after the same constant-budget union bound used in the main upper proof.

It remains to count measurements.  If $0<p<2$, Lemma~\ref{lem:F-dichotomy} gives
$J=O_{p,\eps}(\log(2+b))$, and \eqref{eq:F.2} yields
\begin{equation}
 n s_G\le C_{p,\eps}\frac{n^2\log(2+b)}b.                \tag{F.21}\label{eq:F.21}
\end{equation}
If $p>2$, then $J=O_{p,\eps}(1)$ and
\begin{equation}
 n s_G\le C_{p,\eps}\frac{n^2}{b^{2/p}}.                \tag{F.22}\label{eq:F.22}
\end{equation}
The $W$ block is smaller.  By \eqref{eq:F.15}, the $Y$ block uses at most
\[
 C_{p,\eps}K^{C_{p,\eps}}n^{2-2/K}
\]
measurements.  Choosing the fixed constant multiplying $\log n/\log\log n$ in
\eqref{eq:F.19} sufficiently small makes this last quantity smaller than both
\eqref{eq:F.21} and \eqref{eq:F.22}.  Substitution of $b$ from \eqref{eq:F.19} gives, for all
sufficiently large $n$,
\[
 k_{p,\eps}(n)\le
 \begin{cases}
  C_{p,\eps}n^2\log K/K^{p/4}
     \le C_{p,\eps}n^2/(\log n)^{p/8},&0<p<2,\\[1mm]
  C_{p,\eps}n^2/K^{2/p^2}
     \le C_{p,\eps}n^2/(\log n)^{1/p^2},&p>2.
 \end{cases}
\]
For bounded $n$, enlarge $C_{p,\eps}$ and store the full matrix.  This proves the corollary with
$c_p=p/8$ for $0<p<2$ and $c_p=1/p^2$ for $p>2$.
\end{proof}

\end{document}